\documentclass{article}
\usepackage{graphicx} 

\usepackage[utf8]{inputenc} 
\usepackage[T1]{fontenc} 
\usepackage{url}            
\usepackage{booktabs}       
\usepackage{amsfonts}       
\usepackage{nicefrac}       
\usepackage{microtype}      
\usepackage[dvipsnames]{xcolor}         
\usepackage{geometry}
\usepackage{amsmath, amsthm, amstext, amssymb, mathtools}
\usepackage{hyperref}
\usepackage[nameinlink,capitalize,noabbrev]{cleveref}
\usepackage{dsfont}
\usepackage{subcaption}
\usepackage{graphicx}
\usepackage[shortlabels]{enumitem}
\usepackage{algorithm}
\usepackage{algpseudocode}
\usepackage[most]{tcolorbox}
\usepackage{stmaryrd}\SetSymbolFont{stmry}{bold}{U}{stmry}{m}{n}
\usepackage{tikz,rotating}
\usetikzlibrary{arrows.meta, positioning, calc, fit}
\usepackage{caption}
\usepackage{listings}
\usepackage{thmtools}
\usepackage{thm-restate}
\usepackage[normalem]{ulem}

\DeclareMathOperator*{\argmin}{arg\,min}

\newtheorem{theorem}{Theorem}[section]

\newtheorem{corollary}[theorem]{Corollary}

\newtheorem{lemma}[theorem]{Lemma}
\newtheorem{claim}[theorem]{Claim}
\newtheorem*{claim*}{Claim}

\newtheorem{observation}[theorem]{Observation}

\newtheorem{definition}{Definition}[section]

\allowdisplaybreaks

\crefname{claim}{claim}{claims}
\Crefname{claim}{Claim}{Claims}

\newcommand{\QBF}{\ensuremath{\mathsf{QBF}}}
\newcommand{\TQBF}{\ensuremath{\mathsf{TQBF}}}

\newcommand{\Vars}{[n]}
\newcommand{\Clauses}{[m]}

\usepackage{xspace}
\newcommand{\fsoc}{for the sake of contradiction\xspace}

\newcommand{\CE}{\mathsf{CE}}

\newcommand{\first}{\mathrm{FirstX}}
\newcommand{\firstU}{\mathrm{FirstUnsat}}

\newcommand{\winner}{\ensuremath{C}}      
\newcommand{\loser}{\ensuremath{D}}       
\newcommand{\clause}{\ensuremath{L}}      
\renewcommand{\t}{\ensuremath{T}}         
\newcommand{\f}{\ensuremath{F}}           
\newcommand{\guard}{\ensuremath{X}}       
\newcommand{\lit}{\ell}
\newcommand{\group}{\ensuremath{\mathsf g}}

\newcommand{\is}{ {i^*}}
\newcommand{\ks}{{k^*}}

\newcommand{\weight}{w}
\newcommand{\SV}{\ensuremath{\mathrm{SV}}}
\newcommand{\SSV}{\ensuremath{\mathrm{SSV}}}

\newlength{\casewdii}
\newlength{\casewdiii}
\newlength{\casewdiv}
\newlength{\casesw}
\newlength{\subcasesw}
\newcommand{\pspace}{\ensuremath{\mathsf{PSPACE}}}
\newcommand{\np}{\ensuremath{\mathsf{NP}}}

\newcommand{\tqbf}{\TQBF}

\newcommand{\true}{\textnormal{\textsc{true}}}
\newcommand{\false}{\textnormal{\textsc{false}}}

\lstdefinestyle{pythoncode}{
  language=Python,
  basicstyle=\ttfamily\footnotesize,
  columns=fullflexible,
  keepspaces=true,
  breaklines=true,
  breakatwhitespace=false,
  showstringspaces=false,
  frame=single,
  framerule=0.3pt,
  rulecolor=\color{black!35},
  xleftmargin=0.5em,
  xrightmargin=0.5em,
  aboveskip=0.8em,
  belowskip=0.8em
}

\newtcolorbox{breakablealgorithm}[2][]{%
  breakable,
  enhanced,
  colback=white,
  colframe=lightgray,
  boxrule=0.6pt,
  sharp corners,
  top=6pt,
  bottom=6pt,
  left=6pt,
  right=6pt,
  title={Algorithm~#2},
  fonttitle=\bfseries,
  #1
}


\newtcbtheorem[number within=section]{marginrules}{Construction}{
    enhanced,
    unbreakable,
    colback=white,
    colframe=lightgray,
    boxrule=0.5pt,
    sharp corners,
    fonttitle=\bfseries,
    top=3pt,
    bottom=3pt,
    left=6pt,
    right=6pt,
}{} 

\definecolor{prefixOrange}{HTML}{ffd89e} 
\definecolor{fullBlue}{HTML}{9ac8fc}     
\definecolor{assignGreen}{HTML}{ccffd7}  

\definecolor{borderOrange}{HTML}{FFB74D}
\definecolor{borderBlue}{HTML}{4FC3F7}
\definecolor{borderGreen}{HTML}{81C784}

\pgfdeclarelayer{background}
\pgfdeclarelayer{layer1}
\pgfdeclarelayer{layer2}
\pgfsetlayers{background,layer1,layer2,main}

\definecolor{winnerblue}{RGB}{0, 102, 204}
\definecolor{loserred}{RGB}{204, 0, 0}
\definecolor{literalgreen}{RGB}{0, 153, 51}
\definecolor{clauseorange}{RGB}{255, 140, 0}
\definecolor{guardpurple}{RGB}{153, 51, 204}
\definecolor{neutralgray}{RGB}{70, 70, 70}

\definecolor{groupteal}{RGB}{0, 170, 170}   
\definecolor{groupwine}{RGB}{170, 0, 110}   
\definecolor{groupgold}{RGB}{200, 160, 0}   

\colorlet{gOne}{winnerblue}          
\colorlet{gTwo}{literalgreen}        
\colorlet{gThree}{clauseorange}      
\colorlet{gFour}{clauseorange!70!black}
\colorlet{gFive}{guardpurple}        
\colorlet{gSix}{groupgold}           
\colorlet{gSeven}{loserred}          
\colorlet{gEight}{groupteal}         
\colorlet{gNine}{literalgreen!70!black}
\colorlet{gTen}{neutralgray}
\colorlet{gEleven}{neutralgray!75!black}
\colorlet{gTwelve}{groupwine}        

\newcommand{\groupone}{\textcolor{gOne}{\group1}}
\newcommand{\grouptwo}{\textcolor{gTwo}{\group2}}
\newcommand{\groupthree}{\textcolor{gThree}{\group3}}
\newcommand{\groupfour}{\textcolor{gFour}{\group4}}
\newcommand{\groupfive}{\textcolor{gFive}{\group5}}
\newcommand{\groupsix}{\textcolor{gSix}{\group6}}
\newcommand{\groupseven}{\textcolor{gSeven}{\group7}}
\newcommand{\groupeight}{\textcolor{gEight}{\group8}}
\newcommand{\groupnine}{\textcolor{gNine}{\group9}}
\newcommand{\groupten}{\textcolor{gTen}{\group10}}
\newcommand{\groupeleven}{\textcolor{gEleven}{\group11}}
\newcommand{\grouptwelve}{\textcolor{gTwelve}{\group12}}

\newcommand\redsout{\bgroup\markoverwith{\textcolor{Plum}{\rule[0.5ex]{2pt}{1pt}}}\ULon}

\title{Stable Voting is PSPACE-Complete}
\author{
Ethan Dickey\\
Purdue University\\
\texttt{dickeye@purdue.edu}
\and
Alexandros Psomas\\
Purdue University\\
\texttt{apsomas@cs.purdue.edu}
\and
Athina Terzoglou\\
Purdue University\\
\texttt{aterzogl@purdue.edu}
}

\date{}

\crefname{marginrules}{Construction}{Constructions}
\Crefname{marginrules}{Construction}{Constructions}
\crefname{tcb@cnt@marginrules}{Construction}{Constructions}
\Crefname{tcb@cnt@marginrules}{Construction}{Constructions}

\begin{document}

\maketitle
\begin{abstract}
Stable Voting and Simple Stable Voting, introduced by Holliday and Pacuit, are Condorcet-consistent voting rules defined recursively: a candidate wins if they would win after removing some opponent they beat, taking the pair with the largest margin first.
The computational complexity of winner determination under these rules has been an open question. We resolve this problem: winner determination is \pspace-complete under both Stable Voting and Simple Stable Voting.
\end{abstract}%

\section{Introduction}

Social choice studies the problem of aggregating individual preferences towards a collective decision. In the single-winner setting, a \emph{social choice function}, or \emph{voting rule}, maps a profile of voters' preferences over a set of alternatives to a single winning alternative. %
Social choice functions are typically evaluated by the axioms they satisfy \cite{brandt2016handbook}. A central axiom in the field is \emph{Condorcet consistency}, which states that whenever some candidate wins against every other in a head-to-head majority comparison, that candidate should be the unique winner. While Condorcet consistency is a compelling notion, majority comparisons can cycle: a majority may prefer $A$ to $B$, $B$ to $C$, and yet, $C$ to $A$. Hence, a  Condorcet winner may not always exist. %

Rules that always output a Condorcet winner when one exists are called \emph{Condorcet consistent}. \emph{Split Cycle}, introduced by Holliday and Pacuit \cite{holliday2023split}, and recently characterized axiomatically by Ding, Holliday, and Pacuit \cite{ding2025axiomatic}, is a Condorcet-consistent rule that resolves majority cycles using a simple principle: it considers every cycle of head-to-head majorities and deletes the edge with the smallest margin within that cycle; the winners are the candidates that remain undefeated according to the resulting relation (i.e., candidates with no incoming edges in the resulting defeat relation). Split Cycle satisfies a number of appealing axiomatic properties and can be computed in polynomial time, but it is not resolute: on some elections, it returns multiple winners.

\emph{Stable Voting} (\SV) and \emph{Simple Stable Voting} (\SSV), introduced by Holliday and Pacuit \cite{holliday2023stable}, refine Split Cycle into a resolute rule that selects a single winner. Both rules are defined recursively around the following principle: if a candidate $A$ would win once another candidate $B$ is removed, then it should win the whole election. Among such candidates, it selects the one with the largest $(A,B)$ margin. The two rules, \SV\ and \SSV, differ only in which candidates they consider. \SV\ requires that candidates are undefeated according to Split Cycle, while \SSV\ does not have such a restriction. 
As a result, \SV\ always selects a winner among the Split Cycle winners, but that is not necessary for \SSV (there exists a counterexample with seven alternatives \cite{holliday2026SV}). On any uniquely weighted tournament, both rules return a unique winner. \SV\ powers \url{stablevoting.org}, where it is used to run public polls on a daily basis.

The recursive definition of \SV\ and \SSV\ makes winner determination algorithmically challenging. To determine the winner of an $n$-candidate election, one must determine (at least naively) winners in many $(n-1)$-candidate elections. A straightforward dynamic program over the candidate subsets runs in exponential time. In practice, this limits computation to elections of around 25 to 30 candidates \cite{Hol25,stablevotingcode}. One might hope that this recursion can be bypassed using simple local features of the weighted tournament, such as the number of pairwise victories of each candidate. \Cref{fig:almost-condorcet} shows that this is not the case. In the tournament shown, candidate $c_1$ is nearly Condorcet: it beats ten of its fourteen opponents. Nevertheless, the winner is $c_5$, which beats only one opponent. This also illustrates why the recursive structure cannot be replaced by a simple score-like proxy: even candidates that are locally dominant can lose, and candidates that look locally weak can win.
In this paper, we resolve the computational complexity of \SV\ and \SSV, previously posed as an open question \cite{holliday2026SV,Hol25}: winner determination is PSPACE-complete under both rules.

\begin{figure}[ht]
  \centering

  \begin{subfigure}[b]{0.45\linewidth}
    \centering
    \resizebox{\linewidth}{!}{%
    \begin{tikzpicture}[
      >={Stealth[length=2.3mm]},
      vhi/.style={circle, draw=blue!55!black, very thick, fill=blue!15, minimum size=8mm, inner sep=0pt, font=\sffamily\bfseries},
      vdim/.style={circle, draw=black!45, thin, fill=white, minimum size=8mm, inner sep=0pt, font=\sffamily\bfseries, text=black!60},
      fedge/.style={->, blue!55!black, semithick},
      gedge/.style={->, black!40, thin},
    ]
      \node[vhi]  (c1) at (0,2.4){$c_{1}$};
      \node[vdim] (c2) at (2,2.4){$c_{2}$};
      \node[vdim] (c3) at (4,2.4){$c_{3}$};
      \node[vdim] (c4) at (6,2.4){$c_{4}$};
      \node[vdim] (c5) at (8,2.4){$c_{5}$};
      \node[vdim] (c6) at (0,-0.6){$c_{6}$};
      \node[vdim] (c7) at (2,-0.6){$c_{7}$};
      \node[vdim] (c8) at (4,-0.6){$c_{8}$};
      \node[vdim] (c9) at (6,-0.6){$c_{9}$};
      \node[vdim] (c10) at (-3,3.2){$c_{10}$};
      \node[vdim] (c11) at (-3,2.2){$c_{11}$};
      \node[vdim] (c12) at (-3,1.2){$c_{12}$};
      \node[vdim] (c13) at (-3,0.2){$c_{13}$};
      \node[vdim] (c14) at (-3,-0.8){$c_{14}$};
      \node[vdim] (c15) at (-3,-1.8){$c_{15}$};
      \draw[fedge] (c1) -- (c6);
      \draw[fedge] (c1) to[out=35,in=145] (c2);
      \draw[fedge] (c1) to[out=45,in=135] (c3);
      \draw[fedge] (c1) to[out=52,in=128] (c4);
      \draw[fedge] (c1) to[out=58,in=122] (c5);
      \draw[fedge] (c1) to[out=200,in=0] (c11.east);
      \draw[fedge] (c1) to[out=210,in=0] (c12.east);
      \draw[fedge] (c1) to[out=218,in=0] (c13.east);
      \draw[fedge] (c1) to[out=226,in=0] (c14.east);
      \draw[fedge] (c1) to[out=234,in=0] (c15.east);
      \draw[gedge] (c10.east) to[out=0,in=160] (c1.north west);
      \draw[gedge] (c7) to[out=140,in=295] (c1);
      \draw[gedge] (c8) to[out=140,in=310] (c1);
      \draw[gedge] (c9) to[out=140,in=325] (c1);
    \end{tikzpicture}}
    \caption{$c_{1}$ wins $10$ of the $14$ others, yet loses.}
    \label{fig:c1-loses}
  \end{subfigure}
  \hfill
  \begin{subfigure}[b]{0.45\linewidth}
    \centering
    \resizebox{\linewidth}{!}{%
    \begin{tikzpicture}[
      >={Stealth[length=2.3mm]},
      vhi/.style={circle, draw=green!45!black, very thick, fill=green!20, minimum size=8mm, inner sep=0pt, font=\sffamily\bfseries},
      vdim/.style={circle, draw=black!45, thin, fill=white, minimum size=8mm, inner sep=0pt, font=\sffamily\bfseries, text=black!60},
      tedge/.style={->, green!40!black, very thick},
      gedge/.style={->, black!40, thin},
    ]
      \node[vdim] (c1) at (0,2.4){$c_{1}$};
      \node[vdim] (c2) at (2,2.4){$c_{2}$};
      \node[vdim] (c3) at (4,2.4){$c_{3}$};
      \node[vdim] (c4) at (6,2.4){$c_{4}$};
      \node[vhi]  (c5) at (8,2.4){$c_{5}$};
      \node[vdim] (c6) at (0,-0.6){$c_{6}$};
      \node[vdim] (c7) at (2,-0.6){$c_{7}$};
      \node[vdim] (c8) at (4,-0.6){$c_{8}$};
      \node[vdim] (c9) at (6,-0.6){$c_{9}$};
      \node[vdim] (c10) at (11,3.2){$c_{10}$};
      \node[vdim] (c11) at (11,2.2){$c_{11}$};
      \node[vdim] (c12) at (11,1.2){$c_{12}$};
      \node[vdim] (c13) at (11,0.2){$c_{13}$};
      \node[vdim] (c14) at (11,-0.8){$c_{14}$};
      \node[vdim] (c15) at (11,-1.8){$c_{15}$};
      \draw[tedge] (c5) to[out=20,in=180] (c10.west);
      \draw[gedge] (c4) to[out=45,in=145] (c5);
      \draw[gedge] (c3) to[out=45,in=130] (c5);
      \draw[gedge] (c2) to[out=45,in=115] (c5);
      \draw[gedge] (c1) to[out=45,in=100] (c5);
      \draw[gedge] (c6) to[out=45,in=205] (c5);
      \draw[gedge] (c7) to[out=45,in=220] (c5);
      \draw[gedge] (c8) to[out=45,in=235] (c5);
      \draw[gedge] (c9) to[out=45,in=250] (c5);
      \draw[gedge] (c11.west) to[out=180,in=-10] (c5);
      \draw[gedge] (c12.west) to[out=180,in=-25] (c5);
      \draw[gedge] (c13.west) to[out=180,in=-40] (c5);
      \draw[gedge] (c14.west) to[out=180,in=-55] (c5);
      \draw[gedge] (c15.west) to[out=180,in=-70] (c5);
    \end{tikzpicture}}
    \caption{$c_{5}$ wins only $1$ of the $14$ others, yet wins.}
    \label{fig:c5-wins}
  \end{subfigure}
\caption{ In this tournament, the strong candidate $c_{1}$ (left), which beats ten of the fourteen others, loses, while the weak candidate $c_{5}$ (right), which beats only one, wins according to \SSV. In \Cref{apx:table}, \Cref{tab:ranks} shows the complete weighted tournament for completeness.}
\label{fig:almost-condorcet} 
\end{figure}
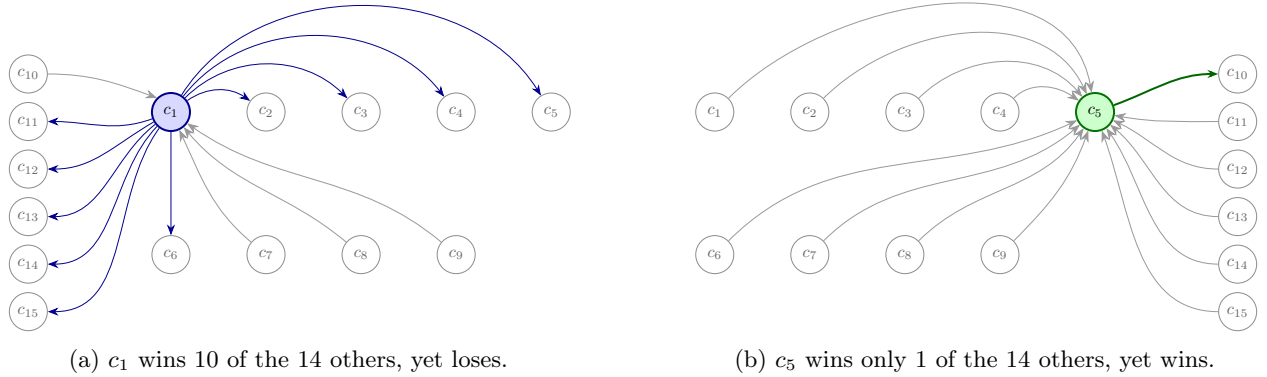

\subsection{Our Contribution}

We settle the computational complexity of winner determination for Stable Voting (\SV) and Simple Stable Voting (\SSV). We prove that, given a weighted tournament and a designated candidate, deciding whether that candidate is the winner is \pspace-complete for both rules. The hardness holds even for uniquely weighted tournaments, and even under the promise that the winner is one of two distinguished candidates, which we denote by $\winner$ and $\loser$.

The proof is by a reduction from \tqbf. At a high level, the construction makes the recursion of winner-determination for \SV\ and \SSV\ simulate the recursive evaluation of a quantified Boolean formula. For each variable $x_i$, the tournament contains two literal candidates, corresponding to setting $x_i$ to true or false. The recursive deletion of one of these candidates encodes the choice of a truth value. The ordering of the relevant margins then implements the quantifier prefix: existential variables allow the recursion to select a branch favorable to $\winner$, while universal variables allow $\loser$ to win unless both branches are favorable to $\winner$. Thus, the winner of the constructed tournament is $\winner$ exactly when the quantified formula is true, and is $\loser$ otherwise.

The main difficulty is that this simulation must be carried out inside a single complete weighted tournament. Since every pair of candidates has a specified comparison, the construction cannot isolate gadgets by omitting edges, as is often possible in graph-based reductions, nor can it rely on a simple global Condorcet-winner or Condorcet-loser structure. Moreover, the recursive rule examines many induced subtournaments, including subtournaments produced by deleting a candidate that the intended simulation would not delete. The auxiliary candidates therefore have to behave correctly not only on the intended \textit{prefix} and \textit{full-assignment} subtournaments, but also on the illegal subtournaments that arise from such deviations. We handle this using guard and clause candidates: guards prevent the recursion from leaving the space of legal partial assignments, while clause candidates check whether a completed assignment satisfies the formula.

We first prove \pspace-completeness for SSV, whose recursive definition is slightly simpler. We then show that the same construction also works for \SV. Although \SV\ imposes the additional requirement that certain candidates be undefeated, we prove that this restriction does not change the outcome on the subtournaments relevant to our construction. Hence, \SV\ and \SSV\ agree on the instances produced by the reduction.

Conceptually, the result differs from the usual sources of high complexity in computational social choice. Many standard tournament solutions are polynomial-time computable, while others have \np-hard or \np-complete winner-membership problems; the tournament equilibrium set is known to be \np-hard and in \pspace, but no matching \pspace-hardness result is known \cite{brandt2016handbookTournament,woeginger2003banks,brandt2010computational}. Nearby \pspace-hardness results typically obtain their power from succinct or combinatorial preference representations, online or candidate-sequential elections, or explicit strategic game structure \cite{lang2004logical,HEMASPAANDRA2014,hemaspaandra2017complexity,ClevelandComplexity}. Here, by contrast, the alternation of the \tqbf\ instance is generated internally by the recursive winner-determination procedure itself. To the best of our knowledge, this is the first \pspace-completeness result for winner determination by a natural voting rule on an explicitly represented complete weighted tournament.

\subsection{Related Work}\label{sec:related}
\paragraph{Pairwise-majority rules and tournament solutions.}
Voting rules based on pairwise comparisons differ in how much information from the weighted majority tournament they use \cite{fishburn1977condorcet,laslier1997tournament}. Tournament solutions such as Copeland, the Smith set, the uncovered set, the Banks set, and the tournament equilibrium set (TEQ) depend only on the orientation of the majority relation \cite{Copeland1951,smith73,miller1980new,banks1985sophisticated,schwartz1990tournament}. By contrast, Kemeny, Minimax, Ranked Pairs, and Beat Path additionally use the magnitudes of the pairwise margins \cite{kemeny,simpson1969,rankedPairs87tideman,schulze11}, as does Split Cycle, which deletes the victories that are weakest on a cycle and returns the candidates that remain undefeated \cite{holliday2023split,ding2025axiomatic}. Stable Voting and Simple Stable Voting are margin-based rules that recursively check claims of the form that $A$ wins after $B$ is removed. Like the rules above, however, they depend only on the relative ordering of the pairwise margins \cite{holliday2023stable,holliday2026SV}.

\paragraph{Computational complexity.}
The systematic study of computational complexity in voting was initiated by Bartholdi, Tovey, and Trick. They showed that winner determination under the rules of Dodgson and Kemeny is \np-hard \cite{bartholdi1989voting}. In follow-up work, they studied the computational hardness of strategic manipulation \cite{bartholdi1989computational} and subsequently introduced the complexity-theoretic study of electoral control \cite{bartholdi1992hard}. These initial hardness results were later sharpened to exact complexity characterizations: winner determination for Dodgson, Kemeny, and Young elections is complete for $\Theta_2^p=\mathsf{P}^{\np}_{\parallel}$, the class of problems solvable in polynomial time using parallel queries to an \np\ oracle \cite{hemaspaandra1997exact,hemaspaandra2005complexity, rothe2003exact}. Among tournament solutions, deciding whether a designated candidate is a Banks winner is \np-complete \cite{woeginger2003banks}, while deciding membership in the recursively defined tournament equilibrium (TEQ) set is \np-hard and belongs to \pspace\ \cite{brandt2010computational}.

Hardness also arises when a voting rule permits multiple outcomes through tie-breaking. Brill and Fischer showed that, for Tideman's original neutral formulation of Ranked Pairs, deciding whether a designated candidate is a Ranked Pairs winner is \np-complete, even though one Ranked Pairs winner can be found in polynomial time \cite{rankedPairs87tideman,rankedPairsComplexity}. Under parallel-universe tie-breaking, deciding whether a candidate can win is also \np-complete for Single Transferable Vote (STV) and Ranked Pairs; Wang et al.\ give practical search and integer-programming algorithms for computing these winner sets \cite{wang2019practical}. Boehmer, Bredereck, and Peters give a fine-grained and parameterized analysis of related winner- and position-determination problems for sequential scoring rules, including STV, Coombs, and Baldwin \cite{boehmer2026rank}.

The closest \pspace-completeness results in computational social choice concern problems with an explicit online, sequential, or strategic component. Hemaspaandra, Hemaspaandra, and Rothe show that online manipulation in sequential elections can be \pspace-complete, even for election systems with computationally simple winner problems \cite{HEMASPAANDRA2014}. The same authors subsequently show that online candidate control in candidate-sequential elections can also be \pspace-complete \cite{hemaspaandra2017complexity}. More recently, Cleveland, de Keijzer, and Polukarov show that deciding whether a subgame-perfect equilibrium exists in a sequential primary election is \pspace-complete \cite{ClevelandComplexity}.
In each of these settings, the complexity arises from an explicit sequence of strategic or irrevocable choices.

\section{Preliminaries}\label{sec:prelims}

We study the computational complexity of winner determination in Stable Voting (SV) and Simple Stable Voting (SSV).
Voting rules are typically defined on voting profiles (a set of candidates, a set of voters, and the voters' preferences over the candidates). For our purposes, however, it is more convenient to work directly with the pairwise majority margins induced by a voting profile, which contain all the information necessary to compute the winner in SV and SSV. We represent these margins by a skew-symmetric weighted tournament:

\begin{definition}[Skew-symmetric weighted tournament]\label{def:weighted-tournament} A \emph{skew-symmetric weighted tournament} is a pair $G=(V,w)$, where $V$ is a finite set of candidates and $w:V\times V\to\mathbb{Z}$ satisfies, for all distinct $A,B\in V$,
\[
w(A,B)=-w(B,A)
\qquad\text{and}\qquad
w(A,B)\neq 0.
\]
We refer to $w(A,B)$ as the weight of the arc $A\to B$.
\end{definition}

We say that $G$ is \emph{uniquely weighted} if all values $w(A,B)$, over ordered pairs $A\neq B$, are distinct.
Throughout the paper, we refer to skew-symmetric weighted tournaments simply as \emph{weighted tournaments}. In figures, we draw only the positive-weight edge between each pair of candidates; the reverse edge, with the opposite weight, is implicit.

For a candidate $A\in V$, let $G\setminus\{ A \}$ denote the weighted tournament obtained from $G$ by deleting $A$ and all of its incident edges. More generally, for $S\subseteq V$, let $G\setminus S$ denote the subtournament induced by $V\setminus S$. For a subset $S\subseteq V$, let $G_{S}$ denote the subtournament induced by $S$.

\paragraph{Stable Voting and Simple Stable Voting.} 

Stable Voting and Simple Stable Voting depend only on the relative order of the pairwise margins. Therefore, the winner of a voting profile can be computed from its induced weighted tournament, which can itself be computed from the profile in polynomial time. Conversely, any realizable weighted tournament can be converted in polynomial time into a voting profile inducing it (see \Cref{apx:vot-profile}).

It is therefore without loss of generality to state our hardness reduction directly in terms of weighted tournaments. In our construction, for simplicity, the positive weights are $(1,\ldots,\binom{|V|}{2})$, so the resulting tournament need not be realizable as written. This is also without loss. Multiplying all weights by the same positive constant preserves the relative order of the margins, and hence preserves the winner under Stable Voting and Simple Stable Voting. Thus, after multiplying all weights by $2$, the constructed tournament has the same winner and can be realized by a voting profile.

We now define Stable Voting and Simple Stable Voting recursively on weighted tournaments. %

\begin{definition}[Simple Stable Voting~\cite{holliday2023stable}]
Let $G=(V,\weight)$ be a weighted tournament. The winner $\SSV(G)\in V$ is defined recursively as follows.
\begin{enumerate}
     \item If $G$ has a unique candidate $A$, then $\SSV(G):=A$.
    \item Otherwise, consider all ordered pairs $(A,B)$ of distinct candidates in decreasing order of $\weight(A,B)$. Let $(A,B)$ be the first pair such that
    $\SSV(G\setminus\{B\})=A$. Then $\SSV(G):=A$.
\end{enumerate}
\end{definition}

To define SV, we first need to define the notion of a \emph{defeat}:
\begin{definition}[$B$ does not defeat $A$]\label{def:does-not-defeat}
Let $G=(V,\weight)$ be a weighted tournament and let $A, B \in V$ be two (distinct) candidates. We say that $B$ \textbf{does not defeat} $A$ if there exists a directed path $A = x_0, x_1, \dots, x_t=B$ in $G$ such that
\[
\weight(x_i,x_{i+1}) \ge \weight(B,A)
\qquad\text{for all } i\in\{0,\dots,t-1\}.
\]
If no such path exists, then $B$ \textbf{defeats} $A$. We say that $A$ is \textbf{undefeated} in $G$ if no candidate defeats $A$; equivalently, $B$ does not defeat $A$ for every $B \in V \setminus \{A\}$.
\end{definition}

Note that for two distinct candidates $A$ and $B$, it is possible that neither is defeated by the other; that is, $B$ does not defeat $A$ and $A$ does not defeat $B$ may hold at the same time. For instance, suppose $w(A,B)>0$. Then $B$ does not defeat $A$, since the edge $A\to B$ already satisfies the condition, and in fact this holds for any positive edge. If the remaining weights also give a directed path from $B$ back to $A$ in which every edge has weight at least $w(A,B)$, then $A$ does not defeat $B$ as well. However, it is not possible that $A$ defeats $B$ and $B$ defeats $A$.

\begin{definition}[Stable Voting \cite{holliday2023stable}]
Let $G=(V,\weight)$ be a weighted tournament. The winner $\SV(G)\in V$ is defined recursively as follows.
\begin{enumerate}
    \item If $G$ has a unique \textbf{undefeated} candidate $A$, then $\SV(G):=A$.
    \item Otherwise, consider all ordered pairs $(A,B)$ such that $A$ is \textbf{undefeated}, in decreasing order of $\weight(A,B)$. Let $(A,B)$ be the first pair such that $\SV(G\setminus\{B\})=A$. Then $\SV(G):=A$.
\end{enumerate}
\end{definition}

Both \SV\ and \SSV\ have a winner, and, in fact, a unique winner when the tournament is uniquely weighted.
Their distinction is based on whether the winner is undefeated, and it might result in a different winner between the two rules. By definition, the \SV\ winner is always among the undefeated candidates, whereas the \SSV\ winner need not be. \Cref{fig:sv-ssv-differ} shows the smallest such example~\cite{holliday2026SV}, with seven candidates, in which the \SSV\ winner $a$ is defeated, while the undefeated candidates are $b$, $c$, and $d$. Hence, \SV\ and \SSV\ select different winners here.

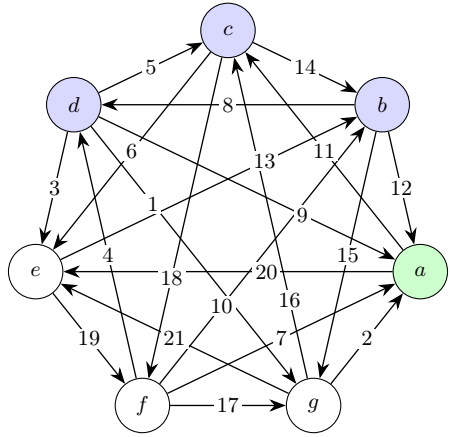
\begin{figure}[ht]
\centering
\begin{tikzpicture}[scale=0.9,  transform shape,
  >={Stealth[length=2.2mm]},
  e/.style={draw,->,line width=0.5pt},
  wlab/.style={fill=white,inner sep=1.2pt,font=\small},
  cand/.style={circle,draw,minimum size=8mm,inner sep=0pt,font=\small},
  win/.style={cand,fill=green!20}, sc/.style={cand,fill=blue!15}, oth/.style={cand,fill=white}]
  \node[sc]  (c) at (90:2.9)      {$c$};
  \node[sc]  (b) at (38.57:2.9)   {$b$};
  \node[win] (a) at (-12.86:2.9)  {$a$};
  \node[oth] (g) at (-64.29:2.9)  {$g$};
  \node[oth] (f) at (-115.71:2.9) {$f$};
  \node[oth] (e) at (-167.14:2.9) {$e$};
  \node[sc]  (d) at (141.43:2.9)  {$d$};
  \draw[e] (c)--(b); \draw[e] (b)--(a); \draw[e] (g)--(a); \draw[e] (f)--(g);
  \draw[e] (e)--(f); \draw[e] (d)--(e); \draw[e] (d)--(c);
  \draw[e] (a)--(c); \draw[e] (b)--(g); \draw[e] (f)--(a); \draw[e] (g)--(e);
  \draw[e] (f)--(d); \draw[e] (c)--(e); \draw[e] (b)--(d);
  \draw[e] (g)--(c); \draw[e] (f)--(b); \draw[e] (a)--(e); \draw[e] (d)--(g);
  \draw[e] (c)--(f); \draw[e] (e)--(b); \draw[e] (d)--(a);
  \node[wlab] at ($(c)!0.5!(b)$)  {14}; \node[wlab] at ($(b)!0.5!(a)$)  {12};
  \node[wlab] at ($(g)!0.5!(a)$)  {2};  \node[wlab] at ($(f)!0.5!(g)$)  {17};
  \node[wlab] at ($(e)!0.5!(f)$)  {19}; \node[wlab] at ($(d)!0.5!(e)$)  {3};
  \node[wlab] at ($(d)!0.5!(c)$)  {5};  \node[wlab] at ($(a)!0.5!(c)$)  {11};
  \node[wlab] at ($(b)!0.5!(g)$)  {15}; \node[wlab] at ($(f)!0.5!(a)$)  {7};
  \node[wlab] at ($(g)!0.5!(e)$)  {21}; \node[wlab] at ($(f)!0.5!(d)$)  {4};
  \node[wlab] at ($(c)!0.5!(e)$)  {6};  \node[wlab] at ($(b)!0.5!(d)$)  {8};
  \node[wlab] at ($(g)!0.28!(c)$) {16}; \node[wlab] at ($(f)!0.33!(b)$) {10};
  \node[wlab] at ($(a)!0.40!(e)$) {20}; \node[wlab] at ($(d)!0.33!(g)$) {1};
  \node[wlab] at ($(c)!0.66!(f)$) {18}; \node[wlab] at ($(e)!0.66!(b)$) {13};
  \node[wlab] at ($(d)!0.66!(a)$) {9};
\end{tikzpicture}
\caption{A uniquely weighted tournament on seven candidates where SV and SSV disagree, shown in \cite{holliday2026SV}. The SSV winner is $a$ (green); the undefeated candidates are $b$, $c$, $d$ (blue), so the SV winner differs from the SSV winner.}
\label{fig:sv-ssv-differ}
\end{figure}

\paragraph{Critical Edges.} The recursive definition of $\SSV$ scans candidate pairs in decreasing order of weight and selects the first pair $(A,B)$ such that $\SSV(G\setminus\{ B \})=A$. Therefore, the winner is ``determined'' by an edge of the form $A\to B$. We extend this idea to the subtournament induced after removing an arbitrary candidate $U\in V(G)$. For weighted tournaments where all weights are unique, $G\setminus\{U\}$ has a unique winner~\cite{holliday2023stable}, and thus the relevant edge is of the form $\SSV(G\setminus\{U\})\to U$. We call this the \emph{critical edge} of $U$. Formally:

\begin{definition}[Critical edge]\label{def:critical-edges}
Let $G=(V,\weight)$ be a weighted tournament. For each candidate $B\in V$, the
\textbf{critical edge} of $B$ is the directed edge $\SSV(G\setminus\{B\})\to B$ (with possibly negative weight). We denote the set of
all critical edges by
\[
\CE_{\SSV}(G):=\{\,(\SSV(G\setminus\{B\}),B): B\in V\,\}.
\]
\end{definition}
Analogously, for Stable Voting, the critical edges are defined

\[
  \CE_{\SV}(G):=\bigl\{\,(\SV(G\setminus\{B\}),B): B\in V \text{ and } \SV(G\setminus\{B\}) \text{ is undefeated in } G\,\bigr\}.
\]
\smallskip

This notion is crucial for our proofs, because it allows us to reason about a set of size at most $|V|$ (the set of critical edges; at most one for each candidate\footnote{$\SSV$ has exactly $|V|$ critical edges, but $\SV$ has at most $|V|$ since defeated candidates do not produce critical edges.}) instead of all $\Theta(|V|^2)$ edges of $G$. In particular, after identifying all critical edges, $\SSV(G)$ and $\SV(G)$ are determined by the maximum-weight critical edge.

\paragraph{Elimination Order.} The recursive definition of $\SV$ and $\SSV$ not only identifies the winner, but also determines a sequence of candidates that are removed one by one along the unique successful recursive branch. Starting from $G$, the maximum-weight critical edge selects the first candidate to remove; applying the same rule to the resulting subtournament selects the next candidate to remove; and so on, until only the winner remains. We refer to this sequence as the \emph{elimination order}. Formally,

\begin{definition}[Elimination order]\label{def:elim-order}
Let $G$ be a uniquely weighted tournament and let $\SSV(G)=A$. Define a sequence of subtournaments $H_0,H_1,\dots,H_{|V|-1}$
recursively by setting $H_0:=G$, and for each $i\in\{1,\dots,|V|-1\}$, letting $A\to B_i$ be the (possibly negatively weighted) maximum-weight critical edge of $H_{i-1}$ and setting $H_i := H_{i-1}\setminus\{ B_i\}.$ The resulting sequence 
\[B_1,B_2,\dots,B_{|V|-1}\]
is called the \textbf{elimination order of $A$ in $G$}.
\end{definition}

\subsection{Example: Simple Stable Voting Winner Determination}

To illustrate the recursive winner determination, consider the uniquely weighted tournament
$G$ on candidates $V=\{a,b,c,d\}$ shown in \Cref{fig:ssv_example}.

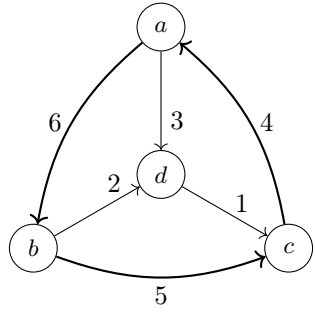
\begin{figure}[ht]
    \centering
    \begin{tikzpicture}[scale=1.3, v/.style={circle, draw, inner sep=1pt, font=\small, minimum size=1.8em}]
        \node[v] (a) at (90:1.5) {$a$};
        \node[v] (b) at (210:1.5) {$b$};
        \node[v] (c) at (330:1.5) {$c$};
        \node[v] (d) at (0,0) {$d$}; %

        \draw[->, thick, bend right=20] (a) to node[left] {6} (b);
        \draw[->, thick, bend right=20] (b) to node[below] {5} (c);
        \draw[->, thick, bend right=20] (c) to node[right] {4} (a);

        \draw[->] (a) to node[right, pos=0.7] {3} (d);
        \draw[->] (b) to node[above, pos=0.7] {2} (d);
        \draw[->] (d) to node[above, pos=0.7] {1} (c); %
    \end{tikzpicture}
    \caption{The tournament $G$ with 4 candidates.}
    \label{fig:ssv_example}
\end{figure}

To determine $\SSV(G)$, we compute the critical edge associated with each candidate $v\in V$. Namely, for each $v$, we first recursively determine the winner of the subtournament $G\setminus\{v\}$, and then form the critical edge $\SSV(G\setminus\{v\})\to v$. Since $G$ is uniquely weighted, $\SSV(G)$ is determined by the (unique) maximum-weight critical edge. Observe that the critical edge is always from the winner of the subtournament to the candidate removed from $G$, even if the original edge is in the other direction (in which case the critical edge has a negative weight).

Four 3-candidate subtournaments result from the removal of a candidate in $G$. For these subtournaments, the winner is easy to identify. If there is a Condorcet winner, they win. Otherwise, the three candidates form a directed cycle, and the winner is the candidate with the incoming edge that has the smallest weight. This follows directly from the recursive definition of \SSV. \Cref{fig:ssv-critical-edges} shows the subtournament winners and the corresponding critical edges.

\begin{figure}[ht]
    \centering

    \begin{subfigure}[t]{0.24\textwidth}        
    \centering
        \begin{tikzpicture}[scale=1.1, baseline=(current bounding box.center),
            v/.style={circle, draw, inner sep=1pt, font=\footnotesize, minimum size=1.5em}]
            \path[use as bounding box] (-1.8,-1.2) rectangle (1.8,1.9);

            \node[v, fill=green!20, thick] (a) at (90:1.3) {$a$};
            \node[v] (b) at (210:1.3) {$b$};
            \node[v] (c) at (330:1.3) {$c$};
            \node[v, gray, dashed] (d) at (0,0) {$d$};

            \draw[->, bend right=20] (a) to node[left] {\small6} (b);
            \draw[->, bend right=20] (b) to node[below] {\small5} (c);
            \draw[->, bend right=20] (c) to node[right] {\small4} (a);
            \draw[->, red, ultra thick] (a) -- node[right, black, font=\bfseries] {3} (d);
        \end{tikzpicture}

        \vspace{1mm}
        \caption{$G\setminus\{d\}: \SSV(G\setminus\{d\})=a$.\\ Critical edge weight: $3$}\label{fig:G-d}
    \end{subfigure}
    \hfill
    \begin{subfigure}[t]{0.24\textwidth}
        \centering
        \begin{tikzpicture}[scale=1.1, baseline=(current bounding box.center),
            v/.style={circle, draw, inner sep=1pt, font=\footnotesize, minimum size=1.5em}]
            \path[use as bounding box] (-1.8,-1.2) rectangle (1.8,1.9);

            \node[v, fill=blue!15, thick] (a) at (90:1.3) {$a$};
            \node[v] (b) at (210:1.3) {$b$};
            \node[v, gray, dashed] (c) at (330:1.3) {$c$};
            \node[v] (d) at (0,0) {$d$};

            \draw[->, bend right=15] (a) to node[left] {\small6} (b);
            \draw[->] (a) -- node[right, pos=0.6] {\small3} (d);
            \draw[->] (b) -- node[above, pos=0.6] {\small2} (d);
            \draw[->, red, ultra thick, bend left=30] (a) to node[right, black, font=\bfseries] {-4} (c);
        \end{tikzpicture}

        \vspace{1mm}
        \caption{$G\setminus\{c\}: \SSV(G\setminus\{c\})=a$.\\ Critical edge weight: $-4$}
    \end{subfigure}
    \hfill
    \begin{subfigure}[t]{0.24\textwidth}
        \centering
        \begin{tikzpicture}[scale=1.1, baseline=(current bounding box.center),
            v/.style={circle, draw, inner sep=1pt, font=\footnotesize, minimum size=1.5em}]
            \path[use as bounding box] (-1.8,-1.2) rectangle (1.8,1.9);

            \node[v] (a) at (90:1.3) {$a$};
            \node[v, gray, dashed] (b) at (210:1.3) {$b$};
            \node[v, fill=blue!15, thick] (c) at (330:1.3) {$c$};
            \node[v] (d) at (0,0) {$d$};

            \draw[->, bend right=15] (c) to node[right] {\small4} (a);
            \draw[->] (a) -- node[right, pos=0.6] {\small3} (d);
            \draw[->] (d) -- node[above, pos=0.6] {\small1} (c);
            \draw[->, red, ultra thick, bend left=30] (c) to node[below, black, font=\bfseries] {-5} (b);
        \end{tikzpicture}

        \vspace{1mm}
        \caption{$G\setminus\{b\}: \SSV(G\setminus\{b\})=c$.\\ Critical edge weight: $-5$}
    \end{subfigure}
    \hfill
    \begin{subfigure}[t]{0.24\textwidth}
        \centering
        \begin{tikzpicture}[scale=1.2, baseline=(current bounding box.center),
            v/.style={circle, draw, inner sep=1pt, font=\footnotesize, minimum size=1.5em}]
            \path[use as bounding box] (-1.8,-1.2) rectangle (1.8,1.9);

            \node[v, gray, dashed] (a) at (90:1.3) {$a$};
            \node[v, fill=blue!15, thick] (b) at (210:1.3) {$b$};
            \node[v] (c) at (330:1.3) {$c$};
            \node[v] (d) at (0,0) {$d$};

            \draw[->, bend right=15] (b) to node[below] {\small5} (c);
            \draw[->] (b) -- node[above, pos=0.6] {\small2} (d);
            \draw[->] (d) -- node[above, pos=0.6] {\small1} (c);
            \draw[->, red, ultra thick, bend left=20] (b) to node[left, black, font=\bfseries] {-6} (a);
        \end{tikzpicture}

        \vspace{1mm}
        \caption{$G\setminus\{a\}: \SSV(G\setminus\{a\})=b$.\\ Critical edge weight: $-6$}
    \end{subfigure}

    \caption{The four critical edges of $G$. The maximum critical edge is $a \to d$ with weight $3$.}
    \label{fig:ssv-critical-edges}
\end{figure}
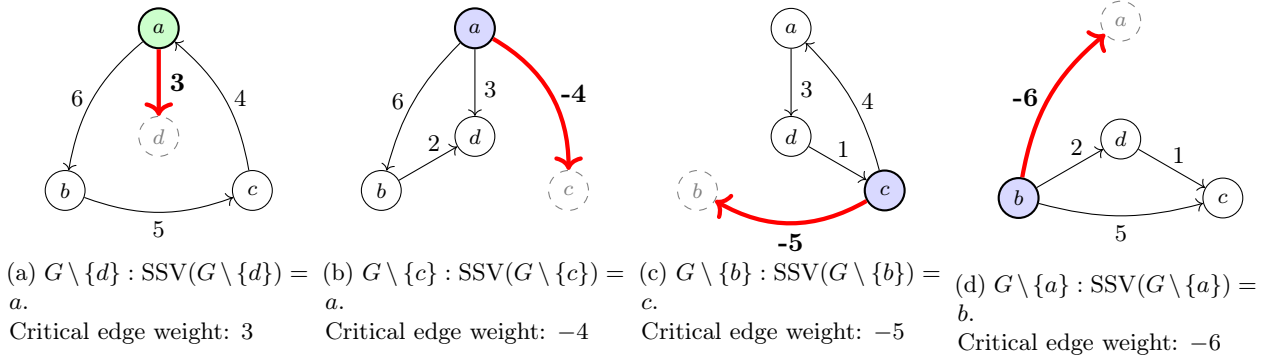

Among the four critical edges shown in \Cref{fig:ssv-critical-edges}, the maximum-weight one is $a\to d$, with weight $3$. Therefore, the winner is $a.$  To determine the elimination order, we recurse on the subtournament induced by removing $d$ (\Cref{fig:G-d}). In the subtournament with candidates $\{a,b,c\}$, each critical edge is again obtained by removing a candidate and taking the winner of the resulting two-candidate subtournament (the candidate with the outgoing edge). Hence, the critical edges are $a\to c$ with weight $-4$ (which we get from removing $c$ from $\{a,b,c\}$), $c\to b$ with weight $-5$ (which we get from removing $b$ from $\{a,b,c\}$), and $b\to a$ with weight $-6$ (which we get from removing $a$ from $\{a,b,c\}$).  The maximum among those is $a\to c$ (observe that the maximum weight here is negative), so $c$ is removed next. Finally, on the remaining subtournament $\{a,b\}$, the winner is $a$, so $b$ is removed. Thus, $\SSV(G)=a$, with elimination order $[d,c,b]$.

\section{The Tournament Construction}\label{sec:construction}

We prove hardness for SV and SSV by reducing from \TQBF\ (true quantified Boolean formula), the canonical \pspace-complete problem of deciding the truth of a fully quantified Boolean formula. An instance of \TQBF\ is a \textit{quantified Boolean formula} expression of the form
\[ \Phi = Q_1x_1Q_2x_2\ldots Q_nx_n \;\varphi(x_1,x_2,\ldots,x_n)\]
where each quantifier $Q_i$ is either $\exists$ or $\forall$ and $\varphi(x_1,x_2,\ldots,x_n)$ is a Boolean formula without quantifiers. Since all variables have a quantifier (i.e., there are no free variables), $\Phi$ is either true or false. The problem \TQBF~asks whether a given \QBF~is true.  In this paper, we assume that $\varphi$ is in CNF form (\TQBF\ remains \pspace-complete under this assumption \cite[\S4.3]{arora2009computational}), namely, $\varphi = C_1 \wedge C_2 \wedge \cdots \wedge C_m$, where each clause $C_j$ is a disjunction of literals. We assume without loss of generality that no clause contains both
$x_i$ and $\neg x_i$ for any variable $x_i$, and that duplicate literals have been removed. Indeed, tautological clauses may simply be deleted from the CNF formula.

In~\Cref{sec:construction:weights}, we describe the weighted tournament $G(\Phi)$ constructed from a given \TQBF\ instance $\Phi$. In~\Cref{sec:construction:dfns}, we introduce several auxiliary definitions for $G(\Phi)$ that are used in the analysis.
In~\Cref{sec:hardness} we prove that \SSV\ is \pspace-hard; in ~\Cref{sec:stable_voting} we extend our proof to \SV. For simplicity, we focus our discussion on \SSV\ for the remainder of this section.

\subsection{Main Construction}\label{sec:construction:weights}

Fix a \TQBF\ instance $\Phi = Q_1 x_1 \cdots Q_n x_n\; \varphi(x_1,\ldots,x_n)$ with $n$ variables and $m$ clauses, where $\varphi$ is in CNF form.
Our reduction constructs a uniquely weighted tournament $G(\Phi)=(V(\Phi),\weight(\Phi))$. We have two designated candidates, $\winner$ (winner) and $\loser$ (loser), that encode whether the formula $\Phi$ is \true\ or \false. For each variable $x_i$ (of $\Phi$), we have two candidates $\t_i,\f_i$ for the two possible values of $x_i$ and a ``guard'' $\guard_i$. We refer to the candidates $\t_i$ and $\f_i$ as the literal candidates, or simply the literals, for variable $x_i$. Finally, for each clause $C_k$, we have a clause candidate $\clause_k$. Formally,
    \[ V(\Phi)\;=\;\{\winner,\loser\}\ \cup\ \{\clause_k:k\in[m]\}\ \cup\ \{\t_i,\f_i,\guard_i:i\in[n]\}. \]

\Cref{margin_groups} describes the rest of the construction, i.e., the orientation and weights of the edges of $G(\Phi)$. Within each group, edges indexed by variables (``for all $i \in [n]$'') are ordered by increasing variable index, and edges indexed by clauses (``for all $k\in[m]$'') are ordered by increasing clause index; earlier edges in this order receive larger weights. Note that this construction can be done in polynomial time ($|V| = 3n+m+2$ and $|E| \in O((n + m)^2)$, and each edge is oriented and weighted by one of the explicit rules below).

We prove that $\winner$ and $\loser$ are the only two candidates that can win in $G(\Phi)$. The remaining candidates have supporting roles that are easier to understand through the way \SSV\ computes the winner. Recall that \SSV\ is recursive, determining the winner by deleting one candidate at a time and recursing on the resulting subtournament. The truth value of a quantified formula is also evaluated recursively, fixing the variables one by one in the order of the quantifier prefix. Our construction makes these two recursions mirror each other. In the context of \SSV, a variable $x_i$ (of $\Phi$), corresponding to the candidate group $i$ in $G(\Phi)$, namely $\{\t_i,\f_i,\guard_i\}$, is ``assigned'' a value (true/false) by deleting one of its two literals and recursing on the resulting subtournament.  The \emph{remaining} literal corresponds to the value being assigned to $x_i$.  The guards $\guard_i$ ensure that this mirroring stays faithful, preventing the recursion from reaching ``illegal'' subtournaments in which both literals of a variable are deleted before the remaining variables are set. The \SSV\ recursion thus first sets all the variables, one by one, and then evaluates the resulting assignment. If the assignment satisfies the formula, then the winner is $\winner$; otherwise, it is $\loser.$ The formula is satisfied by an assignment (where exactly one literal $\lit_i$ from each variable $i$ is left) if each $\clause_k$ candidate has at least one edge $\lit_i\to\clause_k$ with positive weight. While the recursion is still setting variables, each subtournament it reaches encodes a partial assignment whose winner is $\winner$ or $\loser$ according to whether the formula is \true\ or \false\ given that assignment. At the top level, this gives $\SSV(G(\Phi)) = \winner$ when $\Phi$ is \true\ and $\SSV(G(\Phi)) = \loser$ when $\Phi$ is \false.
The orientation and the weights of the edges dictate the order in which \SSV\ examines candidates, and hence the winner it selects. The weights are chosen precisely so that this order achieves the behavior described above: fixing the variables along the quantifier prefix and yielding $\winner$ or $\loser$ according to the truth value at each step.

\begin{marginrules}[label=margin_groups]{Directed edges of $G(\Phi)$ in decreasing priority}{}%
    \begin{enumerate}[label=($\mathsf g$\arabic*),leftmargin=3em]
    
        \item \textbf{$\winner,\loser$ to literals:} for all $i\in[n]$, for each variable $x_i$, the quantifier $Q_i$ determines the relative order of the edges in this group:
        \begin{itemize}[label=--]
            \item if $Q_i=\exists$, then $\winner\rightarrow \f_i,\winner\rightarrow \t_i$, and $\loser\rightarrow \f_i,\loser\rightarrow \t_i$.
            \item if $Q_i=\forall$, then $\loser\rightarrow \f_i,\loser\rightarrow \t_i$, and $\winner\rightarrow \f_i,\winner\rightarrow \t_i$.
        \end{itemize}
    
        \item \textbf{Literals to their own guard:} for all $i\in[n]$, $\f_i\rightarrow \guard_i$ and $\t_i\rightarrow \guard_i$.
    
        \item \textbf{Literals to clauses that contain them:} for clause $k\in[m]$ and variable $i\in[n]$.
        \begin{itemize}[ label=--]
            \item If $x_i$ appears positively in $C_k$, then $\t_i\rightarrow \clause_k$.
            \item If $x_i$ appears negatively in $C_k$, then $\f_i\rightarrow \clause_k$.
        \end{itemize}
    
        \item \textbf{Clauses to literals not in them: }for clause $k\in[m]$ and variable $i\in[n]$.
        \begin{itemize}[label=--]
            \item If $x_i$ appears positively in $C_k$, then  $\clause_k\rightarrow \f_i$. \item If $x_i$ appears negatively in $C_k$, then $\clause_k\rightarrow \t_i$.
            \item If $x_i$ does not appear in $C_k$, then $\clause_k\rightarrow \f_i$ and $\clause_k\rightarrow \t_i$.
        \end{itemize}
    
        \item \textbf{Guards to other literals:} for all $i\neq j$, $\guard_i\rightarrow \f_j$ and $\guard_i\rightarrow \t_j$.
    
        \item \textbf{Guards to clauses:} for all $i\in[n],k\in[m]$, $\guard_i\rightarrow \clause_k$.
    
        \item \textbf{Clauses to $\winner$ then $\loser$ to clauses:} for all $k\in[m]$, $\clause_k\rightarrow \winner$ and then for all $k\in[m]$, $\loser\rightarrow \clause_k$.
    
        \item \textbf{Guards to $\winner$ then $\loser$:} for all $i\in[n]$, $\guard_i\rightarrow \winner$ and $\guard_i\rightarrow \loser$.
    
        \item \textbf{Total order on literals:\footnote{ \;The relative order within $\group9$ does not matter; \Cref{alg:build-tournament} specifies the order we use.}} for all $i\in[n]$, $\f_i\rightarrow \t_i$, and for all $i<j$,
            \[ \f_i\rightarrow \f_j,\quad \f_i\rightarrow \t_j,\quad \t_i\rightarrow \f_j,\quad \t_i\rightarrow \t_j. \]
        \item \textbf{Total order on guards:} for all $i<j$, $\guard_i\rightarrow \guard_j$.
    
        \item \textbf{Total order on clauses:} for all $k<k'$, $\clause_k\rightarrow \clause_{k'}$.
    
        \item \textbf{Designated candidates:} $\winner\rightarrow \loser$.
    \end{enumerate}

\end{marginrules}

\Cref{fig:gadget-overview} gives a high-level view of the construction, organizing the edges into groups by the types of candidates they connect.

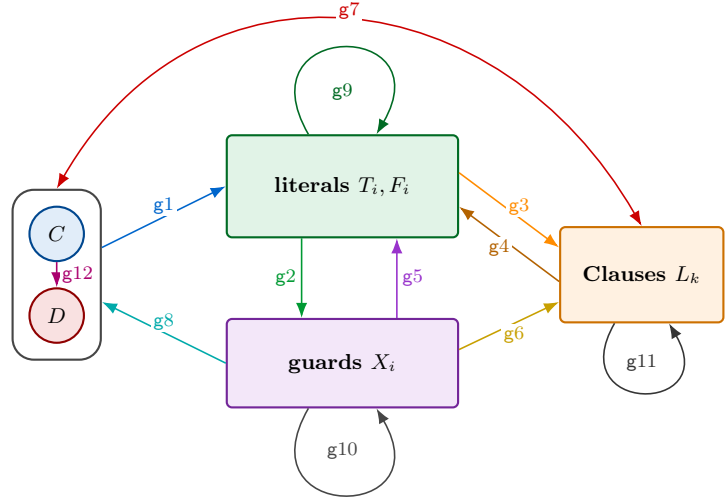
\begin{figure}[ht]
\centering
\begin{tikzpicture}[scale=0.9, transform shape,
    fam/.style={draw, thick, rounded corners=2pt, font=\small\bfseries,
                fill opacity=0.12, text opacity=1},
    cd/.style={circle, draw, thick, minimum size=0.8cm, font=\small\bfseries,
               fill opacity=0.12, text opacity=1},
    g/.style={->, >={Latex[length=2mm,width=1.5mm]}, semithick, font=\footnotesize},
    glab/.style={fill=white, inner sep=1pt, font=\footnotesize}]

  \node[fam, draw=literalgreen!70!black, fill=literalgreen,
        minimum width=3.4cm, minimum height=1.5cm] (Lit) at (4.2, 2.9) {literals $\t_i,\f_i$};
  \node[fam, draw=guardpurple!75!black, fill=guardpurple,
        minimum width=3.4cm, minimum height=1.3cm] (Gu)  at (4.2, 0.3) {guards $\guard_i$};
  \node[fam, draw=clauseorange!80!black, fill=clauseorange,
        minimum width=2.4cm, minimum height=1.4cm] (Cl)  at (8.6, 1.6) {Clauses $\clause_k$};

  \node[cd, draw=winnerblue!70!black, fill=winnerblue] (C) at (0, 2.2) {$\winner$};
  \node[cd, draw=loserred!70!black,   fill=loserred]   (D) at (0, 1.0) {$\loser$};
  \node[draw, thick, rounded corners=8pt, fit=(C)(D), inner sep=6pt,
        fill opacity=0, draw=neutralgray] (CD) {};

  \draw[g, draw=gOne]   ([yshift=4mm]CD.east) -- node[glab, above, pos=0.5]{\groupone} (Lit.west);
  \draw[g, draw=gEight] (Gu.west) -- node[glab, above, pos=0.5]{\groupeight} ([yshift=-4mm]CD.east);
  \draw[g, draw=gTwelve] (C) -- node[glab, right=1pt]{\grouptwelve} (D);

  \draw[g, draw=gSeven, <->] (CD.north) to[out=60,in=120,looseness=1.25]
        node[glab, pos=0.5, above]{\groupseven} (Cl.north);

  \draw[g, draw=gTwo]  ([xshift=-6mm]Lit.south) -- node[glab, left=1pt]{\grouptwo}  ([xshift=-6mm]Gu.north);
  \draw[g, draw=gFive] ([xshift=8mm]Gu.north)    -- node[glab, right=1pt]{\groupfive} ([xshift=8mm]Lit.south);

  \draw[g, draw=gThree]  ([yshift=2mm]Lit.east)-- node[glab, above, pos=0.6]{\groupthree} ([yshift=4mm]Cl.west);
  \draw[g, draw=gFour]       ([yshift=-1mm]Cl.west) -- node[glab, below, pos=0.6]{\groupfour} ([yshift=-3mm]Lit.east);

  \draw[g, draw=gSix] ([yshift=2mm]Gu.east) -- node[glab, below, pos=0.55]{\groupsix} ([yshift=-4mm]Cl.west);

  \draw[g, draw=gNine] ($(Lit.north)+(-0.5,0)$) to[out=120,in=60,looseness=5] ($(Lit.north)+(0.5,0)$);
  \node[glab] at ($(Lit.north)+(0,0.62)$) {\groupnine};
  \draw[g, draw=gTen] ($(Gu.south)+(-0.5,0)$) to[out=-120,in=-60,looseness=5] ($(Gu.south)+(0.5,0)$);
  \node[glab] at ($(Gu.south)+(0,-0.62)$) {\groupten};
  \draw[g, draw=gEleven] ($(Cl.south)+(-0.4,0)$) to[out=-120,in=-60,looseness=5] ($(Cl.south)+(0.4,0)$);
  \node[glab] at ($(Cl.south)+(0,-0.58)$) {\groupeleven};
\end{tikzpicture}
\caption{Overview of the construction. The designated candidates $\winner,\loser$ are single nodes, boxed together on the left; literal, guard, and clause candidates are shown as families. Each labeled arrow is the edge group of the same index in \Cref{margin_groups}, and the self-loops ($\groupnine,\groupten,\groupeleven$) are the internal total orders within a family. Arrows incident to the $\winner,\loser$ box apply to both designated candidates; in particular $\groupseven$ bundles $\clause_k\to\winner$ and $\loser\to\clause_k$. Precise per-edge definitions and the conditional cases ($\groupone$ by quantifier, $\groupthree/\groupfour$ by clause membership) appear in \Cref{margin_groups}.}
\label{fig:gadget-overview}
\end{figure}

\subsection{Useful Definitions}\label{sec:construction:dfns}

Fix a formula $\Phi$ with $n$ variables and $m$ clauses, and let $G(\Phi)=(V(\Phi),\weight(\Phi))$ be the tournament produced by the construction in~\Cref{sec:construction:weights}.  Here, we introduce some notation specific to the construction that will be used in the proofs.

We begin with a property of an induced subtournament that captures whether some guard $\guard_i$ is still present while both of its corresponding literals, $\t_i$ and $\f_i$, are absent. If such a guard exists, the induced subtournament is ``illegal'', since both $\t_i,\f_i$ have been removed before its guard. The following definition returns the smallest such index, or $0$ if no such guard exists.

\begin{definition}[First uncovered guard] \label{def:firstx}
    Let $H$ be an induced subtournament of $G(\Phi)$. Define
        \[ \first(H) := \min\bigl\{ i\in\Vars : \guard_i\in V(H),\; \t_i, \f_i \notin V(H)\bigr\}, \]
    with $\first(H)=0$ if the set is empty.
\end{definition}

Recall that both \tqbf\ and \SSV\ process the quantified formula in order of the variables, setting one at a time. We define the \emph{prefix assignment} where the first $j-1$ variables are set, and the remaining variables are still unresolved. Formally,

\begin{definition}[Prefix assignments and chosen literals] \label{def:assignments}
    For $j\in\{1,\dots,n+1\}$, a \textbf{prefix assignment} of length $j-1$ is an assignment
        \[ \alpha_{< j}\in\{T,F\}^{j-1} \]
    to the first $j-1$ variables $x_1,\dots,x_{j-1}$ in the quantifier prefix. We denote the empty assignment by $\alpha_{<1}$. A prefix assignment of length $n$ is a full assignment.

    For a prefix assignment $\alpha_{< j}$ and each $i\in[j-1]$, define the chosen literal node as
    \[
        \lit_i(\alpha_{< j}) :=
        \begin{cases}
            \t_i, & \text{if } (\alpha_{< j})_i = T,\\[1mm]
            \f_i, & \text{if } (\alpha_{< j})_i = F.
        \end{cases}
    \]
    For a full assignment $\alpha=\alpha_{< n+1}$, we simply write $\lit_i(\alpha)$ instead of $\lit_i(\alpha_{<n+1})$.
\end{definition}

A prefix assignment naturally induces a subtournament of the construction: for each assigned variable, keep only the chosen literal (denoted by $\lit_i$), while for each unassigned variable, keep both literals. The subtournament additionally contains the designated candidates $\winner,\loser$, and all the guards ($\guard_i$'s) and clause ($\clause_k$'s) candidates.  

\begin{definition}[Prefix graph, Full-assignment graph] \label{def:prefix-subgraph}
    Let $\alpha_{< j}$ be a prefix assignment of length $j-1$. Define
    \[
        S_{\alpha_{< j}}
        := \{\winner,\loser\}
        \cup \{\clause_k : k\in\Clauses\}
        \cup \{\guard_i : i\in\Vars\}
        \cup \{\lit_i(\alpha_{< j}) : i\in[j-1]\}
        \cup \{\t_i,\f_i : i\in\Vars\setminus [j-1]\}.
    \]
    The induced subtournament
        \[ G_{\alpha_{< j}} := G(\Phi)_{S_{\alpha_{< j}}} \]
    is called the \textbf{prefix graph} associated with $\alpha_{< j}$. 
    For a full assignment $\alpha=\alpha_{< n+1}$, the prefix graph $G_{\alpha}$ is called the \textbf{full-assignment graph} associated with $\alpha$. 
\end{definition}

Next, we define a special type of induced subtournament, called an \emph{admissible graph}, used throughout the intermediate claims. These are the graphs the \SSV\ recursion reaches after a full-assignment graph (as it keeps deleting candidates). Intuitively, an admissible graph represents a partial assignment on a subset of variables, where each variable has \emph{at most} one of its two literals present. We further impose $\first(\cdot)=0$ to guarantee the assignment is legal, so that every present guard is paired with exactly one of its literals. The converse need not hold, as a present literal may have no corresponding guard.  Aside from literals and guards, it may also contain the designated candidates $\winner$ and $\loser$ along with an arbitrary subset of clauses.

\begin{definition}[Admissible graph] \label{def:admissible-graph}
    An induced subtournament $H$ of $G(\Phi)$ is an \textbf{admissible graph} if $\first(H)=0$ and, for every $i\in\Vars$, at most one of $\{\t_i,\f_i\}$ is in $V(H)$.
\end{definition}

For an admissible graph $H$, write $\lit_i(H)$ (or $\lit_i$ when the context makes it clear) for the unique literal in $V(H)\cap\{\t_i,\f_i\}$ when such a literal is present.

The above definitions capture different stages of the subtournaments during the \SSV\ recursion. First, while the variables are being set, we prove that it stays on prefix graphs until it reaches a full-assignment graph, where all variables are set. From there on, it continues on admissible graphs, which become smaller at every step. We note that the set of prefix graphs intersects the set of admissible graphs at the set of full-assignment graphs. \Cref{fig:assignment-graphs} shows examples of each type.

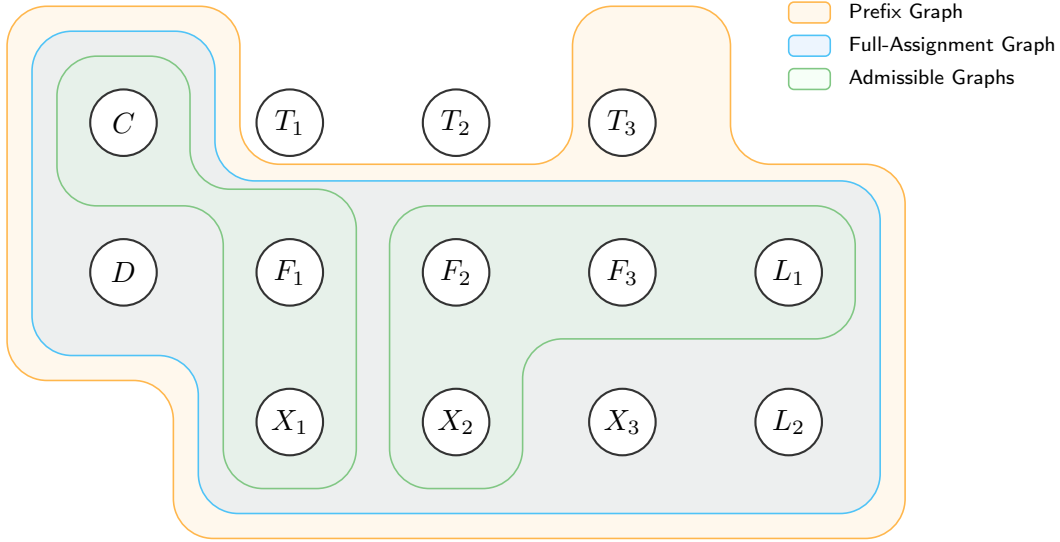
\begin{figure}[ht]
    \centering
    \begin{tikzpicture}[scale=1.1, transform shape,
        v/.style={circle, draw=black!80, thick, fill=white, minimum size=8mm, inner sep=0pt, font=\sffamily\bfseries}
    ]
    \node[v] (C)  at (0, 3.6) {$C$};
    \node[v] (D)  at (0, 1.8) {$D$};
    \node[v] (T1) at (2, 3.6) {$T_1$};
    \node[v] (F1) at (2, 1.8) {$F_1$};
    \node[v] (X1) at (2, 0.0) {$X_1$};
    \node[v] (T2) at (4, 3.6) {$T_2$};
    \node[v] (F2) at (4, 1.8) {$F_2$};
    \node[v] (X2) at (4, 0.0) {$X_2$};
    \node[v] (T3) at (6, 3.6) {$T_3$};
    \node[v] (F3) at (6, 1.8) {$F_3$};
    \node[v] (X3) at (6, 0.0) {$X_3$};
    \node[v] (L1) at (8, 1.8) {$L_1$};
    \node[v] (L2) at (8, 0.0) {$L_2$};
    \tikzset{
        area/.style={line width=0.6pt, rounded corners=15pt, fill opacity=0.18}
    }
    \begin{pgfonlayer}{background}
        \draw [area, fill=prefixOrange, draw=borderOrange]
            (-1.4, 5.0) -- (1.4, 5.0) -- (1.4, 3.1) -- (5.4, 3.1) -- 
            (5.4, 5.0) -- (7.3, 5.0) -- (7.3, 3.1) -- (9.4, 3.1) -- 
            (9.4, -1.4) -- (0.6, -1.4) -- 
            (0.6, 0.5) -- (-1.4, 0.5) -- cycle;
    \end{pgfonlayer}
    \begin{pgfonlayer}{layer1}
        \draw [area, fill=fullBlue, draw=borderBlue]
            (-1.1, 4.7) -- (1.1, 4.7) -- (1.1, 2.9) -- (9.1, 2.9) -- 
            (9.1, -1.1) -- (0.9, -1.1) -- 
            (0.9, 0.8) -- (-1.1, 0.8) -- cycle;
    \end{pgfonlayer}
    \begin{pgfonlayer}{layer2}
        \draw [area, fill=assignGreen, draw=borderGreen]
            (-0.8, 4.4) -- (0.8, 4.4) --
            (0.8, 2.8) -- (2.8, 2.8) --
            (2.8, -0.8) -- (1.2, -0.8) --
            (1.2, 2.6) -- (-0.8, 2.6) -- cycle;
        \draw [area, fill=assignGreen, draw=borderGreen]
            (3.2, 2.6) -- (8.8, 2.6) --
            (8.8, 1.0) -- (4.8, 1.0) --
            (4.8, -0.8) -- (3.2, -0.8) -- cycle;
    \end{pgfonlayer}
    \begin{scope}[shift={(8.0, 4.8)}]
        \draw[area, fill=prefixOrange, draw=borderOrange, rounded corners=2pt] (0,0) rectangle (0.5, 0.25);
        \node[anchor=west, font=\scriptsize\sffamily] at (0.6, 0.12) {Prefix Graph};
        \draw[area, fill=fullBlue, draw=borderBlue, rounded corners=2pt] (0,-0.4) rectangle (0.5, -0.15);
        \node[anchor=west, font=\scriptsize\sffamily] at (0.6, -0.27) {Full-Assignment Graph};
        \draw[area, fill=assignGreen, draw=borderGreen, rounded corners=2pt] (0,-0.8) rectangle (0.5, -0.55);
        \node[anchor=west, font=\scriptsize\sffamily] at (0.6, -0.67) {Admissible Graphs};
    \end{scope}
    \end{tikzpicture}
    \caption{\textbf{Prefix graph} (orange): includes nodes $\{\winner, \loser, \clause_1, \clause_2, \f_1, \f_2, \f_3, \t_3, \guard_1,\guard_2,\guard_3\}$, where $x_1=F$, $x_2=F$, and $x_3$ is not yet set.
    \textbf{Full-assignment graph} (blue): includes nodes $\{\winner, \loser, \clause_1, \clause_2, \f_1, \f_2, \f_3,\guard_1,\guard_2,\guard_3\}$, with exactly one literal per variable.
    \textbf{Admissible graphs} (green): two are shown that include $\{\winner, \f_1, \guard_1\}$ and $\{\f_2, \guard_2, \f_3, \clause_1\}$ respectively.}\label{fig:assignment-graphs}
\end{figure}

Finally, for an admissible graph $H$, we define a property that captures whether the partial assignment satisfies the clauses present in $H$. Recall a clause $\clause_k$ in $H$ is satisfied if there exists an $\lit_i$ such that $\lit_i\to\clause_k$ has a positive weight. If there exists an unsatisfied clause in $H$, we define $\firstU(H)$ as the clause with the smallest index. If all clauses are satisfied, $\firstU(H)=0.$ Formally,

\begin{definition}[First unsatisfied clause]
    \label{def:firstU}
    Let $H$ be an admissible graph. Define
    \[
    \firstU(H):=
    \min\Bigl\{\, k\in\Clauses :
    \clause_k\in V(H)
    \text{ and }
    \nexists i\in\Vars
    \text{ such that }
    \lit_i\text{ is defined and }\weight(\lit_i,\clause_k)>0
    \,\Bigr\},
    \]
    with $\firstU(H)=0$ if the set is empty.
\end{definition}
An admissible graph $H$ is \emph{SAT} if $\firstU(H)=0$, and \emph{UNSAT} otherwise.

\section{\pspace-Hardness of Simple Stable Voting}\label{sec:hardness}

In this section, we prove our main theorem, \pspace-completeness for winner determination in Simple Stable Voting. Here, by winner determination we mean the following decision problem: given a weighted tournament $G$ and a candidate $A$, decide whether $\SSV(G)=A$. We extend the proof to Stable Voting in~\Cref{sec:stable_voting}.

\begin{theorem}\label{thm:main_result}
    Winner determination for Simple Stable Voting is \pspace-complete.
\end{theorem}

\begin{proof}[Proof of~\Cref{thm:main_result}]
Membership in \pspace\ is clear from the recursive definition of \SSV. Given a weighted tournament $G$ on $N$ candidates, we can compute $\SSV(G)$ using a depth-first recursion: scan the ordered pairs $(A,B)$ in decreasing order of $\weight(A,B)$, recursively compute $\SSV(G\setminus\{B\})$, and return the first $A$ for which $\SSV(G\setminus\{B\})=A$. Each recursive call removes one candidate, so the recursion depth is at most $N$. At each level, it suffices to store the current subtournament, the current pair being scanned, and the returned candidate, all of which require only polynomially many bits. Thus, the recursion uses only polynomial space.

In the remainder of this proof, we show that, given a QBF formula $\Phi$, in the tournament $G(\Phi)$ produced by the construction in~\Cref{sec:construction:weights}, $\SSV(G(\Phi))=\winner$ if and only if $\Phi$ is true. Otherwise, $\SSV(G(\Phi))=\loser$.

The \SSV\ rule determines a winner by processing the edges in decreasing order of weight; hence, the first edges to be considered are from $\winner,\loser$ to the literals since they are in the largest weight group ($\group1$). It is thus sufficient to show that the winner of the subtournament without one of those literals is $\winner$ if the residual \QBF\ is true or $\loser$ otherwise, and that the residual tournament respects the \QBF\ encoding.

To encode the \QBF\ formula at every step of the \SSV\ recursion, we need to maintain a valid prefix-assignment such that the winner (of the induced subtournament) correctly encodes the truth of the residual \QBF.
As we establish in~\Cref{lem:guards}, for each variable $x_i$, the guard $\guard_i$ prevents the deletion of both $\t_i$ and $\f_i$ before the recursion reaches a full assignment; in particular, $\guard_i$ wins if both literals are removed, preventing $\winner$ or $\loser$ from winning. 

\begin{restatable}[Guard Lemma]{lemma}{XLemma} \label{lem:guards}
    Let $G$ be an induced subtournament of the constructed tournament $G(\Phi)$. Then, for $\is\in\Vars,$
    \[
    \SSV(G)=\guard_\is \iff \first(G)=\is.
\]
\end{restatable}

\Cref{lem:sat-C-wins,lem:unsat-D-wins} show that when the recursion reaches a full assignment, the winner is the same as the truth value of $\Phi$ (for that assignment). %

\begin{restatable}[If SAT, $\winner$ wins]{lemma}{SATLemma} \label{lem:sat-C-wins}
    Let $\alpha\in\{T,F\}^n$ be a full assignment to the variables of $\varphi$ and let $G_\alpha$ be the corresponding full-assignment graph  (\Cref{def:prefix-subgraph}). If $\alpha$ satisfies $\varphi$, then $\SSV(G_\alpha)=\winner$.
\end{restatable}

\begin{restatable}[If UNSAT, $\loser$ wins]{lemma}{UNSATLemma} \label{lem:unsat-D-wins}
    Let $G$ be an admissible subtournament of $G(\Phi)$ such that $\winner,\loser\in V(G)$. If at least one clause candidate present in $G$ is unsatisfied in $G$, equivalently, $\firstU(G)\neq 0$, then $\SSV(G) = D$. In particular, if $\alpha\in\{T,F\}^n$ is a full assignment that does not satisfy $\varphi$, then $\SSV(G_\alpha)=D$.
\end{restatable}

We prove \Cref{lem:guards,lem:sat-C-wins,lem:unsat-D-wins} in \Cref{subsec:gadget-proofs,subsec:winnerofafullassignmentCD}.

The proof proceeds by backward induction on prefix assignments, where the base case is a full assignment.
For a prefix assignment $\alpha_{< j}\in\{T,F\}^{j-1}$, let $\Phi_{\alpha_{<j}}$ denote the residual QBF obtained from $\Phi$ after fixing the first $j-1$ variables according to $\alpha_{<j}$. For $j\in\{1,\ldots,n+1\}$, consider the following statement:
    \[
        P(j): \text{ for every prefix assignment } \alpha_{< j} \text{ of length } j-1,\,
        \SSV(G_{\alpha_{< j}}) =
        \begin{cases}
            \winner & \text{ if $\Phi_{\alpha_{< j}}$ is true}\\
            \loser  & \text{ otherwise.}
        \end{cases}
    \]
    We prove $P(j)$ for all $j\in\{1,\ldots,n+1\}$ by backward induction.

    \paragraph{Base Case: $(j=n+1)$}
    When $j=n+1$, the prefix assignment is a full assignment. Let $\alpha=\alpha_{<n+1}$ and $G_\alpha$ be the full-assignment graph. If $\alpha$ satisfies $\varphi$, then by \Cref{lem:sat-C-wins}, $\SSV(G_\alpha)=\winner$. If $\alpha$ does not satisfy $\varphi$, then by \Cref{lem:unsat-D-wins}, $\SSV(G_\alpha)=\loser$. Thus $P(n+1)$ holds.

    \paragraph{Induction step.} Assume $P(j+1)$ holds for some $0< j\le n$; we show that $P(j)$ holds. Fix an arbitrary prefix assignment $\alpha:=\alpha_{< j}$, and let $G_\alpha$ be the corresponding prefix graph. The residual \QBF\ has variables $x_1\ldots x_{j-1}$ set according to $\alpha$ and $Q_jx_j$ is the first quantified variable. Hence, we need to show that the winner of $G_\alpha$ correctly encodes the truth of the residual formula $\Phi_\alpha$ based on $Q_j$. 

    To determine $\SSV(G_\alpha)$, we must identify the largest critical edge of $G_\alpha$.  Recall that the \SSV\ rule processes the edges in decreasing order of weight. We show that (i) the edges $\winner,\loser\to \t_i,\f_i$ are not critical edges for all $i<j$ (these are the heaviest edges in $G_\alpha$, since weight decreases as the index $i$ increases), and then (ii) the largest critical edge is one of $\winner\to\t_j$ and $\winner\to\f_j$ if $\Phi_\alpha$ is \true, and the largest critical edge is one of $\loser\to\t_j,\f_j$ if $\Phi_\alpha$ is \false, thus producing the correct winner in each case.

    \begin{enumerate}[(i)]
        \item For every $i<j$, exactly one of $\{\t_i,\f_i\}$ is present in $G_\alpha$; let that be $\lit_i$. Observe that none of the edges $\winner\to\lit_i$ and $\loser\to\lit_i$ are critical, since deleting $\lit_i$ makes $\first(G_\alpha\setminus\{\lit_i\})=i$, and thus, by \Cref{lem:guards}, $\SSV(G_\alpha\setminus\{\lit_i\})=\guard_i$ (that is, the critical edge produced by deleting $\lit_i$ is $\guard_i\to\lit_i$).

        \item The next edges in group $\group1$ to be considered are the four from $\winner,\loser$ to $\f_j,\t_j$, whose order depends on the quantifier $Q_j$. Let $\alpha^T$ and $\alpha^F$ be the two extensions of $\alpha$ setting $x_j=T$ and $x_j=F$, respectively. Then $G_\alpha\setminus\{\f_j\}$ and $G_\alpha\setminus\{\t_j\}$ are the prefix graphs corresponding to $\alpha^T$ and $\alpha^F$, respectively. Since these extensions have length $j$, the induction hypothesis applies to both graphs, so each is won by either $\winner$ or $\loser$ according to the truth value of the corresponding residual formula. Thus, among the four edges from $\winner$ and $\loser$ to $\f_j$ and $\t_j$, exactly two are critical, one for each deletion. By part (i), all earlier group $\group1$ edges have already been ruled out, so the largest critical edge of $G_\alpha$ is the first critical edge among these four.
        \begin{itemize}
            \item \textbf{(Case $Q_j=\exists$)} Recall that in group $\group1$, if the quantifier is $\exists$, then the order of the edges between $\winner,\loser$ and $\t_j,\f_j$ is 
            \[
                \winner\to\f_j\quad
                > \quad\winner\to\t_j\quad
                > \quad\loser\to\f_j\quad
                > \quad\loser\to\t_j
            \]

            Since $Q_j=\exists$, $\Phi_\alpha$ is true exactly when at least one of the two graphs above is won by $\winner$. In that case, the largest critical edge among the four is a $\winner$-edge. If both graphs are won by $\loser$, then neither $\winner$-edge is critical, so the largest critical edge is $\loser\to\f_j$. Thus $\SSV(G_\alpha)=\winner$ if and only if $\Phi_\alpha$ is true, and otherwise $\SSV(G_\alpha)=\loser$.

            \item \textbf{(Case $Q_j=\forall$)} Recall that in group $\group1$, if the quantifier is $\forall$, then the order of the edges between $\winner,\loser$ and $\t_j,\f_j$ is 
            \[
                \loser\to\f_j\quad
                > \quad\loser\to\t_j\quad
                > \quad\winner\to\f_j\quad
                > \quad\winner\to\t_j
            \]
            Since $Q_j=\forall$, $\Phi_\alpha$ is true exactly when both graphs above are won by $\winner$. If at least one of the two graphs is won by $\loser$, then the largest critical edge among the four is a $\loser$-edge. If both graphs are won by $\winner$, then neither $\loser$-edge is critical, so the largest critical edge is $\winner\to\f_j$. Thus $\SSV(G_\alpha)=\winner$ if and only if $\Phi_\alpha$ is true, and otherwise $\SSV(G_\alpha)=\loser$.
        \end{itemize}

    \end{enumerate}
    This proves $P(j)$ and completes the induction step.
    
    \bigskip
    
    Taking $j=1$, the prefix assignment is empty, so $G_{\alpha_{<1}}=G(\Phi)$. Therefore, by $P(1)$, $\SSV(G(\Phi))=\winner$ if and only if $\Phi$ is \true, and otherwise $\SSV(G(\Phi))=\loser$. This concludes the proof of~\Cref{thm:main_result}.
\end{proof}

\subsection{Correctness of Guard and Clause Gadgets}\label{subsec:gadget-proofs}

In this subsection, we prove the two local gadget lemmas underlying the construction. The first is the guard lemma, previously stated as one of the three core lemmas in the proof of \Cref{thm:main_result}. It shows that the guard nodes enforce a valid partial assignment throughout the recursion. The second is the clause lemma, which is used later in the proofs of \Cref{lem:sat-C-wins,lem:unsat-D-wins}. It characterizes what happens when an admissible graph is UNSAT but \loser\ is not present. In that case, the winner is the first unsatisfied clause.

These two lemmas play closely related roles. Each describes a local gadget that takes over precisely when the literals that would otherwise block it have been removed. For the guard gadget, $\guard_i$ wins when both literals $\t_i$ and $\f_i$ are absent while $\guard_i$ is still present. For the clause gadget, $\clause_k$ wins when no literal is still present in the graph that satisfies it. This parallel is also reflected in the edge structure of the gadget. The literals $\t_i,\f_i$ point to their own guard $\guard_i$, while $\guard_i$ points to all other literals; similarly, the literals that satisfy $\clause_k$ point to $\clause_k$, while $\clause_k$ points to the literals that do not satisfy it. Thus, the guard gadget enforces consistency of the represented assignment, while the clause gadget detects unsatisfiability.

At a technical level, the proofs from this point on follow a similar pattern. Since on uniquely weighted tournaments the \SSV\ winner is determined by the maximum critical edge, it suffices to analyze the critical edge produced by each possible deletion. Thus, rather than reasoning about all edges of the tournament at once, we consider each candidate that could be removed, determine the winner of the resulting subtournament, identify the corresponding critical edge, and then compare these edges to find the maximum one. The proof of \Cref{thm:main_result} avoided most of this because the relevant comparisons there occurred entirely among group $\group1$ edges. In contrast, the lemmas below rely heavily on this critical-edge analysis. \Cref{fig:gadget-edges} shows, for each candidate, all incident edges and the groups they belong to, a useful reference for the following proofs.

\begin{figure}[ht]
    \centering
    \tikzset{
         gnode/.style={circle, minimum size=3.4em, inner sep=1pt, font=\boldmath\small},
        winnernode/.style ={gnode, draw=winnerblue!70!black,   fill=winnerblue!10},
        losernode/.style   ={gnode, draw=loserred!70!black,    fill=loserred!10},
        literalnode/.style ={gnode, draw=literalgreen!70!black, fill=literalgreen!10},
        clausenode/.style  ={gnode, draw=clauseorange!80!black, fill=clauseorange!12},
        guardnode/.style   ={gnode, draw=guardpurple!75!black,  fill=guardpurple!10},
        edge/.style={->, >={{Latex[length=2mm, width=1.5mm]}}, thick},
        weight/.style={font=\small\itshape, fill=white, inner sep=1pt, rounded corners=1pt},
        edge1/.style={edge, draw=gOne},
        edge2/.style={edge, draw=gTwo},
        edge3/.style={edge, draw=gThree},
        edge4/.style={edge, draw=gFour},
        edge5/.style={edge, draw=gFive},
        edge6/.style={edge, draw=gSix},
        edge7/.style={edge, draw=gSeven},
        edge8/.style={edge, draw=gEight},
        edge9/.style={edge, draw=gNine},
        edge10/.style={edge, draw=gTen},
        edge11/.style={edge, draw=gEleven},
        edge12/.style={edge, draw=gTwelve},
    }

    \begin{minipage}{0.48\textwidth}
        \centering
        \begin{tikzpicture}[scale=0.72, transform shape]
            \node[winnernode] (V) at (0,0) {$\winner$};
            \draw[edge7] (180:2.2) -- (V) node[at start, left] {$\clause_k$} node[midway, weight] {\groupseven};
            \draw[edge8] (135:2.2) -- (V) node[at start, above left] {$\guard_i$} node[midway, weight] {\groupeight};
            \draw[edge1] (V) -- (45:2.2) node[at end, above right] {$\f_i, \t_i$} node[midway, weight] {\groupone};
            \draw[edge12] (V) -- (0:2.2) node[at end, right] {$\loser$} node[midway, weight] {\grouptwelve};
        \end{tikzpicture}
        \caption*{Winner Node ($\winner$)}
    \end{minipage}
    \hfill
    \begin{minipage}{0.48\textwidth}
        \centering
        \begin{tikzpicture}[scale=0.72, transform shape]
            \node[losernode] (V) at (0,0) {$\loser$};
            \draw[edge12] (180:2.2) -- (V) node[at start, left] {$\winner$} node[midway, weight] {\grouptwelve};
            \draw[edge8] (135:2.2) -- (V) node[at start, above left] {$\guard_i$} node[midway, weight] {\groupeight};
            \draw[edge7] (V) -- (45:2.2) node[at end, above right] {$\clause_k$} node[midway, weight] {\groupseven};
            \draw[edge1] (V) -- (0:2.2) node[at end, right] {$\f_i, \t_i$} node[midway, weight] {\groupone};
        \end{tikzpicture}
        \caption*{Loser Node ($\loser$)}
    \end{minipage}

    \vspace{1em}

    \begin{minipage}{0.32\textwidth}
        \centering
        \begin{tikzpicture}[scale=0.72, transform shape]
            \node[literalnode] (V) at (0,0) {$\t_i/\f_i$};
            \draw[edge1] (180:2.5) -- (V) node[at start, left] {$\winner, \loser$} node[midway, weight] {\groupone};
            \draw[edge9, <->]  (V) -- (144:2.5)  node[at end, above left] {$\{\f_{j \neq i}, \t_{j \neq i}\}$} node[midway, weight] {\groupnine};
            \draw[edge3]  (V) -- (108:2.5)  node[at end, above left] {$\clause_k : i \in C_k$} node[midway, weight] {\groupthree};
            \draw[edge4] (72:2.5) -- (V)   node[at start, above right] {$\clause_k : i \notin C_k$} node[midway, weight] {\groupfour};
            \draw[edge5] (36:2.5) -- (V) node[at start, above right] {$\guard_{j \neq i}$} node[midway, weight] {\groupfive};
            \draw[edge2] (V) -- (0:2.5) node[at end, right] {$\guard_{i}$} node[midway, weight] {\grouptwo};
        \end{tikzpicture}
        \caption*{Literal Node ($\t_i / \f_i$)}
    \end{minipage}
    \hfill
    \begin{minipage}{0.32\textwidth}
        \centering
        \begin{tikzpicture}[scale=0.72, transform shape]
            \node[clausenode] (V) at (0,0) {$\clause_k$};
            \draw[edge11] (180:2.5) -- (V) node[at start, left] {$\clause_{k' < k}$} node[midway, weight] {\groupeleven};
            \draw[edge7] (V) -- (150:2.5) node[at end, above left] {$\winner$} node[midway, weight] {\groupseven};
            \draw[edge7] (120:2.5) -- (V) node[at start, above left] {$\loser$} node[midway, weight] {\groupseven};
            \draw[edge6] (90:2.5) -- (V) node[at start, above] {$\guard_{i}$} node[midway, weight] {\groupsix};
            \draw[edge3] (60:2.5) -- (V) node[at start, above right] {$\{\t/\f\}_{i \in C_k}$} node[midway, weight] {\groupthree};
            \draw[edge4] (V) -- (30:2.5) node[at end, right] {$\{\t/\f\}_{i \notin C_k}$} node[midway, weight] {\groupfour};
            \draw[edge11] (V) -- (0:2.5) node[at end, right] {$\clause_{k' > k}$} node[midway, weight] {\groupeleven};
        \end{tikzpicture}
        \caption*{Clause Node ($\clause_k$)}
    \end{minipage}
    \hfill
    \begin{minipage}{0.32\textwidth}
        \centering
        \begin{tikzpicture}[scale=0.72, transform shape]
            \node[guardnode] (V) at (0,0) {$\guard_{i}$};
            \draw[edge10] (180:2.2) -- (V) node[at start, left] {$\guard_{j<i}$} node[midway, weight] {\groupten};
            \draw[edge8] (V) -- (144:2.2) node[at end, left] {$\{\winner, \loser\}$} node[midway, weight] {\groupeight};
            \draw[edge5] (V) -- (108:2.2) node[at end, above left] {$\{\f_{j \neq i}, \t_{j \neq i}\}$} node[midway, weight] {\groupfive};
            \draw[edge2] (72:2.2) -- (V) node[at start, above right] {$\t_i,\f_i$} node[midway, weight] {\grouptwo};
            \draw[edge6] (V) -- (36:2.2) node[at end, above right] {$\clause_k$} node[midway, weight] {\groupsix};
            \draw[edge10] (V) -- (0:2.2) node[at end, right] {$\guard_{j>i}$} node[midway, weight] {\groupten};
        \end{tikzpicture}
        \caption*{Guard Node ($\guard_i$)}
    \end{minipage}

    \caption{Edges per node type (refer to \Cref{margin_groups,fig:gadget-overview} for more details).}
    \label{fig:gadget-edges} 
\end{figure}

 \subsubsection{The Guard Lemma}

\XLemma*
    \begin{proof}
        We will prove the lemma by induction on $t \coloneqq |V(G)|.$

        \paragraph{Base case ($t=1$).}
        If $V(G)=\{\guard_\is\}$, for some $\is$, then $\first(G)=\is$ and $\SSV(G)=\guard_\is$. If $V(G)=\{u\},$ where $u\neq \guard_\is\;\forall \is$, $\SSV(G)=u$ with $\first(G)=0$. Thus, the claim holds.

        \paragraph{Induction hypothesis.}
        Fix $t\ge 2$ and assume that for every induced subtournament $H$ of $G(\Phi)$ with $|V(H)|< t$, $\forall \guard_\is\in V(H)$ with $\is>0,$

            \[ \SSV(H)=\guard_\is \iff \first(H)=\is. \]

        \paragraph{Induction step.} Let $G$ be an induced subtournament with $|V(G)|=t$. We prove both directions.

        \paragraph{($\Leftarrow$) $\first(G)=\is>0 \Rightarrow \SSV(G)=\guard_\is$.} Since $\first(G)=\is$ we have that $\guard_\is\in V(G)$ and $\t_\is,\f_\is\notin V(G)$. We show that the largest critical edge of $G$ is from $ \guard_\is$, and thus $\SSV(G)=\guard_\is$. We break the proof into two steps.  First, we show that among the critical edges produced by deleting candidates $u\neq \guard_\is$, the largest one has the form $\guard_\is\to u$. Then, we show that the critical edge produced by deleting $\guard_\is$, namely the edge $\SSV(G\setminus\{\guard_\is\})\rightarrow\guard_\is$, is smaller than the edge $\guard_\is\rightarrow u$. %
        
        Toward our first claim, we focus on the critical edges produced by nodes in the set $u\in V(G\setminus\{\guard_\is\})$. First, observe that the residual graph $G\setminus\{u\}$ still has $0<\first(G\setminus\{u\})\le\is,$ since $\guard_\is\in V(G\setminus\{u\})$ and still $\t_\is,\f_\is\notin V(G\setminus\{u\}).$ However, a deletion of a node may cause a guard $\guard_j$ with a smaller index to be without both of its literals. We consider the following two cases:

        \begin{enumerate}[
                align=left,
                itemindent=0pt,
                labelindent=0pt,
                labelwidth=\casesw,
                labelsep=0.4em,
                leftmargin=\dimexpr\casesw+0.5em\relax
                ]
            \item[\textbf{Case 1:}]$\first(G\setminus\{u\})=\is$. Hence, by the induction hypothesis applied to $G\setminus\{u\}$ we get $\SSV(G\setminus\{u\})=\guard_\is,$  and thus $\guard_\is \rightarrow u$ is a critical edge of $G$. The weight of that critical edge in $G$ is bounded below by $\min_u\{\weight(\guard_\is,u)\}\ge-\group10$,  given the edges adjacent to $\guard_\is$ (\Cref{fig:gadget-edges}). (The only incoming edges to $\guard_\is$ are from $\guard_j$ with $j<\is$, as $\t_\is,\f_\is\notin G$ is given by the assumption that $\first(G)=\is$.)
            \item[\textbf{Case 2:}]
            $0<\first(G\setminus\{u\})=j<\is$. This can only happen if there exists a $j$ such that exactly one of $\f_j$ or $\t_j$ is in $V(G)$ and $\guard_j\in V(G)$. Deleting the remaining literal introduces a $\first$ with a smaller index. By the induction hypothesis applied to $G\setminus\{u\}$, we get $\SSV(G\setminus\{u\})=\guard_j.$ Then, the critical edge in $G$ is $\guard_j \rightarrow \f_j/\t_j$ that has weight in group $-\group_2.$
        \end{enumerate}

        Therefore, all critical edges arising from \textbf{Case 1} have strictly larger weight than every critical edge arising from \textbf{Case 2}. Moreover, if \textbf{Case 2} occurs for some $j<\is$ with critical edge $\guard_j\to \lit_j$ ($-\group2$), then deleting $\guard_j$ instead falls under \textbf{Case 1}, the edge $\guard_\is\to \guard_j$($-\group10>-\group2$) is also critical. Hence, the largest critical edge in $G$ produced by removing any node $u\neq \guard_\is$ is produced in \textbf{Case 1}, and has the form $\guard_\is\rightarrow u$.
        
        \medskip
        It remains to compare the largest critical edge that we just showed with the critical edge associated with deleting $\guard_\is$, namely $w\to \guard_\is$ where $w=\SSV(G\setminus\{\guard_\is\})$. Since $\first(G)=\is$, every present guard $\guard_k$ with $k<\is$ has at least one of its literals in $G$. Deleting $\guard_\is$ does not change this for any $k<\is$. Therefore, $\first(G\setminus\{\guard_\is\}) \in \{0\}\cup\{j:j>\is\}.$ By the induction hypothesis, $w\neq \guard_k$ for every $0<k<\is$. We analyze the two possible cases for $\first(G\setminus\{\guard_\is\})$.

        \begin{enumerate}[
                align=left,
                itemindent=0pt,
                labelindent=0pt,
                labelwidth=\casesw,
                labelsep=0.4em,
                leftmargin=\dimexpr\casesw+0.5em\relax
                ]
            \item[\textbf{Case A:}]  $\first(G\setminus\{\guard_\is\})=0$. Then, by the induction hypothesis, $w$ is not a guard node%
            . Since $\t_{i^\ast},\f_{i^\ast}\notin V(G)$, the edge $w\to \guard_{i^\ast}$ lies in one of the groups $-\group5$, $-\group6$, or $-\group8$. Hence its weight is at most $-\group8$. Since we already identified a critical edge of weight at least $-\group10$ and $-\group10>-\group8$, the edge $w\to \guard_{i^\ast}$ cannot be the maximum critical edge.
            
            \item[\textbf{Case B:}] $\first(G\setminus\{\guard_\is\})=k>\is$. Then, by the induction hypothesis, $w=\guard_k$, so the critical edge associated with deleting $\guard_\is$ is $\guard_k\to \guard_\is$, which lies in group $-\group10$.

            Now consider deleting $\guard_k$ instead. Since $k>\is$, the guard $\guard_\is$ remains the first unguarded guard in $G\setminus\{\guard_k\}$. Hence, by \textbf{Case 1} above, the edge $\guard_\is\to \guard_k$ is also critical, and this edge lies in group $\group10$. Therefore, the critical edge produced by deleting $\guard_\is$ cannot be the largest critical edge of $G$.
        \end{enumerate}

        Thus, in all cases, the largest critical edge of $G$ is produced by deleting some vertex $u\neq \guard_\is$ and has the form $\guard_\is\to u$. Therefore $\SSV(G)=\guard_\is$.

        \paragraph{($\Rightarrow$) $\SSV(G)=\guard_\is \Rightarrow \first(G)=\is$.} Suppose for the sake of contradiction that $\first(G)=j\neq \is.$

        \begin{enumerate}[
                align=left,
                itemindent=0pt,
                labelindent=0pt,
                labelwidth=\casesw,
                labelsep=0.4em,
                leftmargin=\dimexpr\casesw+0.5em\relax
                ]
            \item[\textbf{Case 1:}]  $\first(G)=j>0$ with $j\neq \is$. By the direction just proved $(\Leftarrow)$, applied with index $j$, we get $\SSV(G)=\guard_j$, contradicting $\SSV(G)=\guard_\is$. Hence, this case is impossible.
            \item[\textbf{Case 2:}] $\first(G)=0$. Since $\SSV(G)=\guard_\is$, the largest critical edge in $G$ must be of the form $\guard_\is\rightarrow u$ for some $u\neq \guard_\is$. Applying the induction hypothesis to $G\setminus\{u\}$, we get $\first(G\setminus\{u\})=\is$. Since $\first(G)=0$ but $\first(G\setminus\{u\})=\is$, deleting $u$ must have removed the last present literal for $\guard_\is$. Hence $u\in\{\t_\is,\f_\is\}$, and the critical edge $\guard_\is\rightarrow u$ has weight in group $-\group2$.
    \end{enumerate}

        Now consider the critical edge associated with deleting $\guard_\is$. Since $\first(G)=0$, every guard present in $G$ has at least one of its literals present. Removing $\guard_\is$ preserves this property for all remaining guards, and therefore $\first(G\setminus\{\guard_\is\})=0$. By the induction hypothesis, $\SSV(G\setminus\{\guard_\is\})$ is not a guard node. Hence, by the edges adjacent to $\guard_\is$, $\SSV(G\setminus\{\guard_\is\})\rightarrow\guard_\is$ has weight at least $-\group5$, and $-\group5>-\group2$. Therefore, the critical edge that removes $\guard_\is$ is larger than $\guard_\is\rightarrow u$, contradicting the assumption that $\SSV(G)=\guard_\is$.
        
        \medskip
        This completes the proof of $\SSV(G)=\guard_\is \iff \first(G)=\is$ for $\is > 0.$
    \end{proof}

\subsubsection{The Clause Lemma}

Analogously to the guard nodes, a clause $\clause_k$ wins if no present literal satisfies it. The following lemma shows the winner is the \emph{first} unsatisfied clause, and if every present clause is satisfied ($\firstU(\cdot)=0$), no clause candidate wins.

\begin{lemma}[Clause Lemma]
\label{lem:clauses}
Let $G$ be an admissible graph as in \Cref{def:admissible-graph}, with $\loser\notin V(G)$, and let $\firstU$ be as in \Cref{def:firstU}. Then, for every clause candidate $\clause_\ks\in V(G)$,
\[
    \SSV(G)=\clause_\ks \iff \firstU(G)=\ks.
\]

    \begin{proof}
    We prove the claim by induction on $t\coloneqq |V(G)|$.
    
    Recall that an admissible graph may contain $\winner$, literals, guards, and clauses, contains at most one literal from each pair $\{\t_i,\f_i\}$, and satisfies $\first(G)=0$. In this lemma we also assume that $\loser\notin V(G)$.
    
    \paragraph{Base case ($t=1$).}
    If $G=\{\clause_\ks\}$ for some $\ks\in[m]$, then $\SSV(G)=\clause_\ks$. Since no literal is present, $\clause_\ks$ is unsatisfied, and hence $\firstU(G)=\ks$. If $G=\{u\}$ for some $u\neq \clause_k$ for all $k\in[m]$, then $\firstU(G)=0$ and $\SSV(G)=u\neq \clause_k$ for every $k\in[m]$. Hence, the statement holds.
    
    \paragraph{Induction hypothesis.}
    Fix $t\ge 2$ and assume that the lemma holds for every admissible graph $G'$ with $\loser\notin V(G')$ and $|V(G')|<t$.
    
    \paragraph{Induction step.}
    Let $G$ be an admissible graph with $\loser\notin V(G)$ and $|V(G)|=t$. We show that, for every $\ks\in[m]$ such that $\clause_\ks\in V(G)$,
    \[
        \SSV(G)=\clause_\ks \iff \firstU(G)=\ks.
    \]

    \paragraph{($\Rightarrow$) If $\SSV(G)=\clause_\ks$, then
    $\firstU(G)=\ks$.} Assume $\SSV(G)=\clause_\ks$. Suppose, \fsoc, that $\firstU(G)\neq \ks$. Since $\SSV(G)=\clause_\ks$, the largest critical edge of $G$ has the form $\clause_\ks\to u$ for some $u\neq \clause_\ks$ such that $\SSV(G\setminus\{u\})=\clause_\ks.$ Since the winner of $G\setminus\{u\}$ is a clause candidate, \Cref{lem:guards} implies that $\first(G\setminus\{u\})=0$. Hence $G\setminus\{u\}$ is an admissible graph, and by the induction hypothesis, $\firstU(G\setminus\{u\})=\ks.$
    
    We consider all possible cases for $\firstU(G)$: either $\firstU(G)=0$, $\firstU(G)=k>\ks$, or $\firstU(G)=k<\ks$. In each case, we show a critical edge $a\to b$, with $a\neq\clause_\ks$, whose weight is larger than that of $\clause_\ks\to u$, contradicting the assumption that $\SSV(G)=\clause_\ks$.
    
    \begin{enumerate}[
            align=left,
            itemindent=0pt,
            labelindent=0pt,
            labelwidth=\casesw,
            labelsep=0.4em,
            leftmargin=\dimexpr\casesw+0.5em\relax
            ]
        \item[\textbf{Case 1:}] $\firstU(G)=0$. Since $\firstU(G\setminus\{u\})=\ks$ and $\firstU(G)=0$, the clause $\clause_\ks$ is satisfied in $G$, but becomes the first unsatisfied clause in $G\setminus\{u\}$. Therefore, $u$ must be the unique present literal satisfying $\clause_\ks$. Hence, the only critical edge out of $\clause_\ks$ is $\clause_\ks\to u$ with weight in group $-\group3$ (any other possible $u$ would not produce a critical edge as, by the inductive hypothesis, $\clause_\ks$ only wins in the subtournament if and only if it is $\firstU$).
    
        Now consider $G\setminus\{\clause_\ks\}$. Deleting a clause does not change the set of literals, hence it does not change which remaining clauses are satisfied or unsatisfied. Thus, $\firstU(G\setminus\{\clause_\ks\})=0$ and $\first(G\setminus\{\clause_\ks\})=0$. Let $w=\SSV(G\setminus\{\clause_\ks\})$ be the winner of the residual graph. By the induction hypothesis, $w$ is not a clause candidate, and  by \Cref{lem:guards}, it is not a guard candidate. Hence, $w$ is either $\winner$, if $\winner\in V(G)$, or a literal.
        
        If $w=\winner$, the critical edge produced in $G$, $\winner\to\clause_\ks$, is in group $-\group7.$ If $w=\lit_i\in\{\t_i,\f_i\}$, then the critical edge produced in $G$, $\lit_i\to\clause_\ks$, is in $\group3$ if $\lit_i$ satisfies $\clause_\ks$, and in group $-\group4$ otherwise. In either case, $w\to\clause_\ks$ has weight at least $-\group4$. Since $(\group3>-\group7>)-\group4>-\group3$, the critical edge $w\to\clause_\ks$ has a larger weight than the only critical edge out of $\clause_\ks$, $\clause_\ks\to u$, which is in $-\group3$. This contradicts $\SSV(G)=\clause_\ks$.

        \item[\textbf{Case 2:}] $\firstU(G)=k>\ks$, so $\clause_\ks$ is satisfied in $G$. As in Case~1, since $\firstU(G\setminus\{u\})=\ks$ and $\clause_\ks$ is satisfied in $G$, the only way deleting $u$ makes $\clause_\ks$ the first unsatisfied clause is for $u$ to be the unique literal satisfying $\clause_\ks$. Hence the critical edge $\clause_\ks\to u$ lies in group $-\group3$.
        
        Now consider $G\setminus\{\clause_\ks\}$. Deleting a clause does not change the set of literals, so it does not change which clauses are satisfied; since $\clause_k$ remains present and is still the first unsatisfied clause, we have $\firstU(G\setminus\{\clause_\ks\})=k$ and $\first(G\setminus\{\clause_\ks\})=0$. By the induction hypothesis, $\SSV(G\setminus\{\clause_\ks\})=\clause_k$, so $\clause_k\to\clause_\ks$ is a critical edge of $G$. Since $k>\ks$, this edge lies in group $-\group11$, and $-\group11>-\group3$. Thus $\clause_k\to\clause_\ks$ is larger than $\clause_\ks\to u$, contradicting the choice of $\clause_\ks\to u$ as the largest
        critical edge of $G$.

        \item[\textbf{Case 3:}] $\firstU(G)=k<\ks$. As in the previous case, consider $G\setminus\{\clause_\ks\}$. Deleting a clause does not change the set of literals, so it does not change which clauses are satisfied; since $\clause_k$ remains present and is still the first unsatisfied clause, we have $\firstU(G\setminus\{\clause_\ks\})=k$ and $\first(G\setminus\{\clause_\ks\})=0$. By the induction hypothesis, $\SSV(G\setminus\{\clause_\ks\})=\clause_k$. Thus $\clause_k\to\clause_\ks$ is a critical edge of $G$. Since $k<\ks$, this edge lies in group $\group11$. We consider two subcases according to whether $\clause_\ks$ is satisfied in $G$.
        
        \noindent\textbf{Subcase 3a:} $\clause_\ks$ is unsatisfied in $G$. Since $\firstU(G\setminus\{u\})=\ks$ by the inductive hypothesis, the first unsatisfied clause increases from $k$ in $G$ to $\ks$ in $G\setminus\{u\}$. Thus $u$ must be $\clause_k$; otherwise $\clause_k$ remains present and unsatisfied, so the first unsatisfied clause cannot become $\ks$. Hence the critical edge $\clause_\ks\to u$ is $\clause_\ks\to\clause_k$, which has weight in group $-\group11$.  However, we have already shown that $\clause_k\to\clause_\ks$ is a critical edge of $G$, and this edge lies in group $\group11$. Since $\group11>-\group11$, this gives a critical edge of larger weight than $\clause_\ks\to u$, a contradiction.
        
        \noindent\textbf{Subcase 3b:} $\clause_\ks$ is satisfied in $G$. Then there is no $u$ such that $\firstU(G\setminus\{u\})=\ks$. Notice that a literal $\lit_i$ must be removed to make $\clause_\ks$ unsatisfied. However, even if such a literal is removed, the clause $\clause_k$ remains present and unsatisfied, so the first unsatisfied clause cannot be $\clause_\ks$. Removing any other candidate does not make $\clause_\ks$ unsatisfied. This contradicts the fact that $\firstU(G\setminus\{u\})=\ks$.
    \end{enumerate}
    
    All cases lead to contradictions. Therefore, $\firstU(G)=\ks$.
    
    \paragraph{($\Leftarrow$) If $\firstU(G)=\ks$ then $\SSV(G)=\clause_\ks$.}
    Assume $\firstU(G)=\ks$. We show that the largest critical edge of $G$ has the form $\clause_\ks\to u$ for some $u\neq \clause_\ks$. We consider all critical edges produced by deleting candidates in $V(G)\setminus \{\clause_\ks\}$, and then compare them with the critical edge produced by deleting $\clause_\ks$. 
    
    \begin{enumerate}
        \item \textbf{(Removing $\clause_k$ with $k\neq \ks$.)} Deleting $\clause_k$ does not change the set of literals, hence it does not change which remaining clauses are satisfied or unsatisfied. Therefore, $\firstU(G\setminus\{\clause_k\})=\ks$ and $\first(G\setminus\{\clause_k\})=0$. By the induction hypothesis, $\SSV(G\setminus\{\clause_k\})=\clause_\ks$. Hence, $\clause_\ks\to\clause_k$ is a critical edge of $G$, with weight in group $\pm\group11$, depending on the index of $\clause_k.$
    
        \item \textbf{(Removing $\guard_i$.)} Since $G$ is an admissible graph, deleting $\guard_i$ preserves $\first=0$. Moreover, deleting a guard does not change which clauses are satisfied, so $\firstU(G\setminus\{\guard_i\})=\ks$. By the induction hypothesis, $\SSV(G\setminus\{\guard_i\})=\clause_\ks$. Hence, $\clause_\ks\to\guard_i$ is a critical edge of $G$, with weight in group $-\group6$.
    
        \item \textbf{(Removing $\winner$.)} If $\winner\in V(G)$, then deleting $\winner$ preserves both $\first=0$ and $\firstU=\ks$. By the induction hypothesis, $\SSV(G\setminus\{\winner\})=\clause_\ks$. Hence, $\clause_\ks\to\winner$ is a critical edge of $G$, with weight in group $\group7$.
    
        \item \textbf{(Removing $\lit_i$.)} We consider two cases, depending on the existence of $\guard_i.$
    
        \textbf{Case A:} If $\guard_i\in V(G)$ then $\first(G\setminus\{\lit_i\})=i$. By \Cref{lem:guards}, $\SSV(G\setminus\{\lit_i\})=\guard_i$. Hence, the critical edge produced by deleting $\lit_i$ is $\guard_i\to\lit_i$, which has weight in group $-\group2$. On the other hand, deleting $\guard_i$ instead falls under the previous case, so $\clause_\ks\to\guard_i$ is also a critical edge. This edge has weight in group $-\group6$, and $-\group6>-\group2$. Therefore, the edge $\guard_i\to\lit_i$ cannot be the largest critical edge of $G$.

        \textbf{Case B:} If $\guard_i\notin V(G)$ then $\first(G\setminus\{\lit_i\})=0$. Then $G\setminus\{\lit_i\}$ is still an admissible graph. Since $\clause_\ks$ remains present and unsatisfied, we have $0<\firstU(G\setminus\{\lit_i\})\le \ks$, since deleting a literal may cause a clause with a smaller index to be unsatisfied.

        \textbf{Subcase B.1:} If $\firstU(G\setminus\{\lit_i\})=\ks$, then by the induction hypothesis, $\SSV(G\setminus\{\lit_i\})=\clause_\ks$, and hence $\clause_\ks\to\lit_i$ is a critical edge of $G$.

        \textbf{Subcase B.2:} If instead $\firstU(G\setminus\{\lit_i\})=r<\ks$, then by the induction hypothesis, $\SSV(G\setminus\{\lit_i\})=\clause_r$, so the critical edge produced by deleting $\lit_i$ is $\clause_r\to\lit_i$. Since $\firstU(G)=\ks$, the clause $\clause_r$ was satisfied in $G$, and deleting $\lit_i$ made it unsatisfied. Thus, $\lit_i$ was the only satisfying literal for $\clause_r$, and $\clause_r\to\lit_i$ has weight in group $-\group3$. But by the earlier case of removing a clause $\clause_k$ with $k\neq \ks$, deleting $\clause_r$ instead produces the critical edge $\clause_\ks\to\clause_r$. Since $r<\ks$, this edge lies in group $-\group11$, and $-\group11>-\group3$. Therefore, the edge $\clause_r\to\lit_i$ cannot be the largest critical edge of $G$.
    
        \item \textbf{(Removing $\clause_\ks$.)} First observe that deleting $\clause_\ks$ either removes the only unsatisfied clause or leaves a later unsatisfied clause as the first one. Thus, $\firstU(G\setminus\{\clause_\ks\})\in \{0\}\cup\{k:k>\ks\}$. Let $w=\SSV(G\setminus\{\clause_\ks\})$. We consider both cases.
    
        \textbf{Case A:} $\firstU(G\setminus\{\clause_\ks\})=k>\ks$. By the induction hypothesis, $w=\clause_k$. Hence, the critical edge produced by deleting $\clause_\ks$ is $\clause_k\to\clause_\ks$, which lies in group $-\group11$. However, we showed that removing $\clause_k$ produces the critical edge $\clause_\ks\to\clause_k$ in group $\group11.$ Hence, $\clause_k\to\clause_\ks$ is not the largest critical edge.

        \textbf{Case B:} $\firstU(G\setminus\{\clause_\ks\})=0$. By the induction hypothesis $w$ is not a clause candidate, and since $\first(G\setminus\{\clause_\ks\})=0$, by \Cref{lem:guards} it is not a guard candidate. Hence $w$ is $\winner$ (if $\winner\in V(G)$) or a literal $\lit_i$. In each case we show a critical edge coming out of $\clause_\ks$ that is larger than $w\to\clause_\ks$, so the latter is not the largest critical edge of $G$.

        \textbf{Subcase B.1:} If $w=\winner$, then $w\to\clause_\ks$ has negative weight, whereas $\clause_\ks\to\winner$ (from removing $\winner$ above) is also a critical edge with a positive weight.
        
        \textbf{Subcase B.2:} If $w=\lit_i$, then since $\clause_\ks$ is unsatisfied in $G$, the critical edge $\lit_i\to\clause_\ks$ is in group $-\group4$ (no positive edges point from $\lit_i \to \clause_\ks$). We show there is always a larger critical edge out of $\clause_\ks$.

        \begin{itemize}
            \item[-]If $\guard_i\in V(G)$, then $\clause_\ks\to\guard_i$ (from removing $\guard_i$ above) lies in group $-\group6>-\group4$. 
            \item[-] If $\guard_i\notin V(G)$, by the case of removing $\lit_i$ (Case $4.B$ above) there is always a larger critical edge out of $\clause_\ks$. (Either $\clause_\ks\to\lit_i$ in $\group_4>-\group4$ or $\clause_\ks\to\clause_r$ in $-\group11>-\group4.$)
        \end{itemize}
    \end{enumerate}
    
    Therefore, every critical edge that is not of the form $\clause_\ks\to u$ is dominated by a critical edge of the form $\clause_\ks\to v$ (for some $v$). Hence, the largest critical edge of $G$ has the form $\clause_\ks\to u$, and so $\SSV(G)=\clause_\ks$.
    \end{proof}

\end{lemma}

\subsection{Winner Determination for Full-Assignment  Graphs}\label{subsec:winnerofafullassignmentCD}
In this subsection, we prove the two remaining core lemmas, which determine the \SSV\ winner of full-assignment graphs and serve as the base cases in the proof of \Cref{thm:main_result}.

The proofs of the two core lemmas rely on a few auxiliary claims characterizing the winners of the intermediate admissible graphs that arise along the way. We defer the proofs of these claims to \Cref{apx:missing-proofs}. Here we state two of these claims used in both proofs. %
The first claim shows that \winner\ and \loser\ are ``robust'' under the addition of further single literals. Once either is already the \SSV\ winner, adding single literals that are not yet present does not change that outcome.

\begin{restatable}[Adding single literals preserves a  $\winner$/$\loser$ winner]{claim}{addingLits}\label{lem:adding_lits_C_D_wins}
    Let $G$ be an admissible graph (\Cref{def:admissible-graph}) with $\SSV(G) \in \{\winner, \loser\}$ such that for each $\lit_i \in G$, $\guard_i \in G$.
    Let $B$ be a set of literals containing at most one of $\{\t_j, \f_j\}$ for each variable $j \in \Vars$, with $\{\t_j, \f_j\} \cap V(G) = \emptyset$ for every $\lit_j \in B$ (i.e., $B$ adds only literals that are missing from $G$). 
    Then $G \cup B$ is admissible and
        \[ \SSV(G \cup B) = \SSV(G) \,. \]

\end{restatable}

\smallskip
The second claim states that $\loser$ wins every admissible graph (SAT or UNSAT), whenever $\winner$ is not present.

\begin{restatable}[$\loser$ wins in admissible graphs without $\winner$]{claim}{DwithoutC}
\label{lem:D-wins-with-clauses-no-C-general}
    Let $G$ be an admissible graph as in \Cref{def:admissible-graph} such that $\loser\in V(G)$ and $\winner\notin V(G)$. Then  $\SSV(G)=\loser$.
\end{restatable}

Notice that the converse is not true. In admissible graphs where $\winner$ is present but $\loser$ is not, $\winner$ is not always the winner; it depends on whether the graph is SAT (\Cref{lem:C-wins-SAT-no-D} below) or UNSAT (\Cref{lem:clauses}, as UNSAT implies the presence of clauses).

\subsubsection{D Wins UNSAT Full-Assignment Graphs}\label{subsec:D-UNSAT}

Next we prove \Cref{lem:unsat-D-wins} by induction on the number of clauses. The proof of \Cref{lem:unsat-D-wins} is organized around the elimination order (\Cref{def:elim-order}) of an UNSAT full-assignment graph, which has the following form:
\[
\underbrace{
  \clause_1, \ldots, \clause_{r-1}, \clause_{r+1}, \ldots, \clause_m,
  \underbrace{ \winner, 
    \clause_r, 
    \underbrace{
      \guard_n, \lit_n, \ldots, \guard_1, \lit_1
    }_{\text{\Cref{lem:D-wins-with-clauses-no-C-general}}}
  }_{\text{\Cref{lem:D-wins-with-C-and-one-unsat-clause}}}
}_{\text{\Cref{lem:unsat-D-wins}}}
\]
where $r\in\Clauses$ denotes the largest index of an unsatisfied clause in the admissible graph. The following claim captures the base case of \Cref{lem:unsat-D-wins}, where there is a single unsatisfied clause with both designated candidates $\winner$ and $\loser$.

\begin{restatable}[$\loser$ wins with $\winner$ and a single unsatisfied clause]{claim}{DBaseCase}\label{lem:D-wins-with-C-and-one-unsat-clause}
    Let $G$ be an admissible graph as in \Cref{def:admissible-graph} such that $\winner,\loser\in V(G)$. Suppose that the only clause candidate in $G$ is $\clause_k$, and that $\clause_k$ is unsatisfied in $G$. Then $ \SSV(G)=\loser.$
\end{restatable}

With these auxiliary claims in place, we now prove \Cref{lem:unsat-D-wins}. The proof starts from the base case and then builds up by induction on the number of clauses.

\UNSATLemma*
    
\begin{proof}
    By \Cref{lem:adding_lits_C_D_wins}, it suffices to prove the lemma under the additional assumption that every present literal is accompanied by its guard. We prove this restricted statement by induction on the number of clauses.
    
    \medskip
    \noindent
    \textbf{Claim.}
    Let $I\subseteq[n]$, let $K\subseteq[m],\, |K| > 0$, and, for every $i\in I$, fix exactly one literal $\lit_i\in\{\t_i,\f_i\}$. Let $G(I,K)$ be the induced subtournament on
        \[ V(G(I,K))=\{\winner,\loser\}\cup\{\guard_i,\lit_i:i\in I\}\cup\{\clause_k:k\in K\}. \]
    If at least one clause $\clause_k$, $k\in K$ is unsatisfied by the present literals, then $\SSV(G(I,K))=\loser$.

    \paragraph{Base case.}
    For any set $I'\subseteq I$. If  $|K|=1$, then the unique clause in $K$ is unsatisfied. Hence, by \Cref{lem:D-wins-with-C-and-one-unsat-clause}, we get $\SSV(G(I,K))=\loser$.
    
    \paragraph{Induction hypothesis.}
    Fix $t\ge 2$ and assume the statement holds for every graph $G(I',K')$ of the above form with $I'\subseteq I$ and $|K'|<t$.
    
    \paragraph{Inductive step.}
    Let $G:=G(I,K)$ be a graph of the above form with $|K|=t$, and fix an unsatisfied clause $\clause_{r}\in K$. We show that the largest critical edge of $G$ is of the form $\loser\to\clause_k$ (in group $\group7$) for some $\clause_k\in K$. We examine the critical edge produced by removing each possible type of node.
    
    \begin{enumerate}
    
        \item \textbf{(Removing a clause $\clause_k$.)} If $k\neq r$, then $G\setminus\{\clause_k\}$ still contains the unsatisfied clause $\clause_{r}$, so by the induction hypothesis $\SSV(G\setminus\{\clause_k\})=\loser$. Therefore, the critical edge produced by removing $\clause_k$ is $\loser\to\clause_k$, in group $\group7$. Such an edge always exists since the number of clauses $t\ge 2.$
    
        It remains to consider removing $\clause_{r}$. If there is another unsatisfied clause in $G$, then the same induction argument applies and the critical edge produced by removing $\clause_{r}$ is again an edge out of $\loser$ in group $\group7$.

        So suppose that $\clause_{r}$ is the unique unsatisfied clause. Then $G\setminus\{\clause_{r}\}$ is SAT. We do not need to determine its winner exactly; we need to show that the critical edge produced has weight smaller than group $\group7$. Now, the only edges into $\clause_{r}$ with weight larger than group $\group7$ are edges in groups $\group3$ ($\lit_i\to\clause_r$) and 
        $\group6$ ($\guard_i\to\clause_r$), see \Cref{fig:gadget-edges}. First, observe that removing a clause does not affect guards or literals; we still have $\first(G\setminus\{\clause_{r}\})=0$. By \Cref{lem:guards} no guard can win $G\setminus\{\clause_{r}\}$, so no edge $\guard_i\to\clause_r$ is critical. Additionally, we know that $\clause_{r}$ is unsatisfied, so no present literal satisfies $\clause_{r}$, and therefore no edge $\lit_i\to\clause_{r}$ in group $\group3$ is present. Therefore, $ w(\SSV(G\setminus\{\clause_r\}),\clause_r) \leq \group7$. Moreover, equality is possible only when $\SSV(G\setminus\{L_r\})=D$, in which case the critical edge is $D\to L_r$; otherwise the critical edge has weight strictly below group $g_7$.
        Thus removing $\clause_{r}$ produces either another group $\group7$ edge out of $\loser$, or a strictly smaller critical edge.

        \item \textbf{(Removing $\winner$.)}
        The graph $G\setminus\{\winner\}$ is exactly of the form covered by \Cref{lem:D-wins-with-clauses-no-C-general}, and it is still UNSAT. Hence $\SSV(G\setminus\{\winner\})=\loser$. Therefore the critical edge produced by removing $\winner$ is $\loser\to\winner$, which lies in group $-\group12<\group7$.
    
        \item \textbf{(Removing $\loser$.)}
        Since removing $\loser$ does not affect literals or guards, we still have $\first(G\setminus\{\loser\})=0$. Moreover, $G\setminus\{\loser\}$ is still UNSAT. Hence, by \Cref{lem:clauses}, the winner of $G\setminus\{\loser\}$ is the first unsatisfied clause, say $\firstU(G\setminus\{\loser\})=\ks$. Therefore the critical edge produced by removing $\loser$ is $\clause_{\ks}\to\loser$, which lies in group $-\group7<\group7$.
    
        \item \textbf{(Removing a literal $\lit_i$.)}
        Since $\guard_i$ is still present and $\lit_i$ was the unique present literal of variable $i$, we have $\first(G\setminus\{\lit_i\})=i$. Hence, by \Cref{lem:guards}, $\SSV(G\setminus\{\lit_i\})=\guard_i$. Therefore the critical edge produced by removing $\lit_i$ is $\guard_i\to\lit_i$, which lies in group $-\group2<\group7$.
    
        \item \textbf{(Removing a guard $\guard_i$.)} For ease of notation, let $H:=G\setminus\{\guard_i\}$. We show that the critical edge $\SSV(H)\to\guard_i$ has weight $<\group7.$ By \Cref{fig:gadget-edges}, the only incoming edge to $\guard_i$ with larger weight is $\lit_i\to\guard_i$ in group $\group2.$ Hence it suffices to rule out $\SSV(H)=\lit_i$.
        
        We are going to break the proof into two steps; first we show that $H$ has a critical edge with weight in $\group7.$ Then we contradict that $\SSV(H)=\lit_i.$

        \textbf{Step 1:} We consider the graph $H'= H\setminus\{\lit_i,\clause_k\}$ (for an arbitrary fixed $k\in K\setminus\{r\}$, such a clause always exists since $|K|\ge2$). This graph still contains the unsatisfied clause $\clause_{r}$, and it has $t-1$ clauses. Therefore, by the induction hypothesis, $\SSV(H')=\loser$. By \Cref{lem:adding_lits_C_D_wins}, adding any number of single non-present literals does not change a winner in $\{\winner,\loser\}$, we get $\SSV(H'\cup\{\lit_i\}) =\SSV(H\setminus\{\clause_k\})=\loser$. Hence $\loser\to\clause_k$ is a critical edge of $H$, in group $\group7$.
        
        \textbf{Step 2:} Next we show that $\SSV(H)\neq\lit_i.$ Assume for the sake of contradiction that $\SSV(H)=\lit_i.$ Since we have established a critical edge of $H$ in $\group7,$ to get $\SSV(H)=\lit_i$ the maximum critical edge of $H$ must be of the form $\lit_i\to u$ with weight of at least $\group7.$ The only possibility is $\lit_i\to\clause_j$ for some clause satisfied by $\lit_i$ which is in group $\group3,$ since $\guard_i\notin V(H).$ However, we have shown in Step 1 that in $H$, removing any clause $\clause_j\neq\clause_r$ makes the winner $\loser$, which contradicts the assumption $\SSV(H)=\lit_i$.

        \smallskip

        Thus, since $\SSV(H) \neq \lit_i$, the critical edge produced by removing $\guard_i$, $\SSV(H)\to\guard_i$, has weight at most $\group10 < \group7$, by \Cref{fig:gadget-edges}.
    \end{enumerate}

    We have now considered every possible removal. Removing any clause $\clause_k\neq\clause_{r}$ produces a critical edge $\loser\to\clause_k$ in group $\group7$. Every critical edge produced by removing $\winner$, $\loser$, a literal, or a guard has strictly smaller weight, and removing $\clause_{r}$ produces at most another group $\group7$ edge out of $\loser$. Therefore the maximum critical edge of $G$ is of the form $\loser\to\clause_k$ for some $\clause_k\in K$, and hence $\SSV(G)=\loser$.

    This proves the restricted statement. We now deduce the lemma for an arbitrary admissible subtournament.

    Let $G$ be an arbitrary admissible subtournament satisfying the hypotheses of the lemma. Let
        \[ B:=\{\lit_i\in V(G): \guard_i\notin V(G)\} \]
    be the set of present single literals whose guards are absent, and let
        \[ G^\circ:=G\setminus B. \]
    Because $G$ is admissible, $G^\circ$ has the form $G(I,K)$ for some $I\subseteq[n]$ and $K\subseteq[m]$. Moreover, since $\firstU(G)\neq 0$, at least one clause candidate present in $G$ is unsatisfied in $G$. Removing literals cannot make an unsatisfied clause satisfied, so at least one clause candidate present in $G^\circ$ is unsatisfied in $G^\circ$. Hence the inductive claim gives
        \[ \SSV(G^\circ)=D. \]
    By \Cref{lem:adding_lits_C_D_wins}, applied to $G^\circ$ and the set $B$, we obtain
        \[ \SSV(G)=D. \]
    This proves the main statement of the lemma.
    
    The full-assignment statement follows by applying the main statement to $G_\alpha$.
\end{proof}

\subsubsection{C wins SAT Full-Assignment Graphs}\label{subsec:C-SAT}

Now we prove \Cref{lem:sat-C-wins}. First, we state two auxiliary claims that we use in the proofs. \Cref{lem:c-wins-without-clauses} serves as the base case of the main lemma, and the second (\Cref{lem:C-wins-SAT-no-D}) shows that $\winner$ wins in admissible graphs in which \loser~has already been removed. That implies that $\winner\to\loser$ is critical. After establishing this claim, the proof of \Cref{lem:sat-C-wins} follows by comparing the critical edges produced by deleting each possible type of candidate.

\begin{restatable}[$\winner$ wins in admissible graphs without clauses]{claim}{CwinsNoClause}\label{lem:c-wins-without-clauses}
    Let $G$ be an admissible graph as in \Cref{def:admissible-graph} such that $\winner,\loser\in V(G)$ and $\clause_k\notin V(G)$, for all $k\in[m]$. Then $\SSV(G)=\winner.$
\end{restatable}

\begin{restatable}[$\winner$ wins in SAT admissible graphs without $\loser$]{claim}{CwinsNoD}\label{lem:C-wins-SAT-no-D}
    Let $G$ be an admissible graph (\Cref{def:admissible-graph}) with $\winner \in V(G)$, $\loser \notin V(G)$, and $\firstU(G) = 0$.
    Then $\SSV(G) = \winner$.
\end{restatable}

\smallskip
\SATLemma*
\begin{proof}
    We prove the following stronger statement by strong induction on $t:=|I|+|K|$.
    
    \medskip
    \noindent
    \textbf{Induction statement.}
    Let $I\subseteq[n]$, let $K\subseteq[m]$, and for every $i\in I$ fix exactly one literal $\lit_i\in\{\t_i,\f_i\}$. Let $G(I,K)$ be the induced subtournament on
    \[
    V(G(I,K))=\{\winner,\loser\}\cup\{\guard_i,\lit_i:i\in I\}\cup\{\clause_k:k\in K\}.
    \]
    Assume that every clause $\clause_k\in K$ is satisfied by at least one present literal $\lit_i$. Then $\SSV(G(I,K))=\winner$. The lemma follows by taking $I=[n]$, $K=[m]$, and choosing the literals according to $\alpha$.
    
    \paragraph{Base case.}
    If $K=\emptyset$, then $G(I,K)$ has no clause nodes for every $I\subseteq[n]$. Hence, by \Cref{lem:c-wins-without-clauses}, we get $\SSV(G(I,\emptyset))=\winner$. 
    
    \paragraph{Induction hypothesis.}
    Fix $t\ge 1$ and assume the statement holds for every satisfiable graph $G(I',K')$ with $|I'|+|K'|<t$.
    
    \paragraph{Induction step.}
    Let $G:=G(I,K)$ be satisfiable with $|I|+|K|=t$. We show that the largest critical edge is $\winner\to\loser$ with weight in $\group12>0$. We examine the critical edge produced by removing each possible type of node and show that they are either from $\winner$ or negative. 
    
    \begin{enumerate}
        \item \textbf{(Removing $\loser$.)}
        The graph $G\setminus\{\loser\}$ is a SAT admissible graph with $\winner\in V(G)$, $\loser\notin V(G)$. By  \Cref{lem:C-wins-SAT-no-D}, the winner is $\winner.$ Hence the edge $\winner\to\loser$ is a critical edge with weight in group $\group12$.

        \item \textbf{(Removing a clause $\clause_k$.)}
        The graph $G\setminus\{\clause_k\}=G(I,K\setminus\{k\})$ is still satisfied, and $|I|+|K\setminus\{k\}|=t-1$. Hence, by the induction hypothesis, $\SSV(G\setminus\{\clause_k\})=\winner$. Therefore, the critical edge produced by removing $\clause_k$ is $\winner\to\clause_k$, which lies in group $-\group7$.
        
        \item \textbf{(Removing $\winner$.)}
        The graph $G\setminus\{\winner\}$ is an admissible graph without $\winner$, with exactly one literal for every present variable and with $\first(G\setminus\{\winner\})=0$. Hence, by \Cref{lem:D-wins-with-clauses-no-C-general}, we get $\SSV(G\setminus\{\winner\})=\loser$. Therefore, the critical edge produced by removing $\winner$ is $\loser\to\winner$, which lies in group $-\group12$. ($\winner\to\loser$ is also a critical edge with positive weight.)
        
        \item \textbf{(Removing a literal $\lit_i$.)}
        Since $\guard_i$ is still present and $\lit_i$ was the unique present literal of variable $i$, we have $\first(G\setminus\{\lit_i\})=i$. Hence, by \Cref{lem:guards}, $\SSV(G\setminus\{\lit_i\})=\guard_i$. Therefore, the critical edge produced by removing $\lit_i$ is $\guard_i\to\lit_i$, which lies in group $-\group2$.

        \item \textbf{(Removing a guard $\guard_i$.)}
        We show that $\SSV(G\setminus\{\guard_i\})\in\{\winner,\loser\}.$ Consider $G\setminus\{\guard_i,\lit_i\}.$ This is a smaller graph of the same paired form, with variable set $I\setminus\{i\}$ and the same clause set $K$. If every clause in $K$ remains satisfied after deleting $\lit_i$, then by the induction hypothesis, $\SSV(G\setminus\{\guard_i,\lit_i\})=\winner.$ Hence $\winner\to\lit_i$ is a critical edge of $G\setminus\{\guard_i\}$, in group $\group1$. If instead some clause becomes unsatisfied after deleting $\lit_i$, then by \Cref{lem:unsat-D-wins}, $\SSV(G\setminus\{\guard_i,\lit_i\})=\loser.$ Hence $\loser\to\lit_i$ is a critical edge of $G\setminus\{\guard_i\}$, also in group $\group1$. Thus $G\setminus\{\guard_i\}$ has a group $\group1$ critical edge from $\winner$ or $\loser$ into $\lit_i$. Since all group $\group1$ edges originate from $\winner$ or $\loser$, it follows that
        \[ \SSV(G\setminus\{\guard_i\})\in\{\winner,\loser\}. \]
        Therefore the critical edge produced by removing $\guard_i$ is a negative edge $\{\winner, \loser\} \to \guard_i$ in group $-\group8$.

    \end{enumerate}
    In all cases, every critical edge is either an edge out of $\winner$ or has negative weight. Moreover, removing $\loser$ produces the positive critical edge $\winner\to\loser$. Hence $\SSV(G)=\winner$. Applying the statement to the full-assignment graph $G_\alpha$ completes the proof.
\end{proof}

\section{Extension to Stable Voting}\label{sec:stable_voting}

In this section, we show that the same construction (\Cref{margin_groups}) gives \pspace-completeness for Stable Voting (\SV) as well. The main difference between the two rules is the ``does not defeat'' requirement (\Cref{def:does-not-defeat}). Specifically, \SV\ restricts the set of critical edges to $A\to B$ such that $A$ is undefeated, whereas \SSV\ does not have such a restriction. We show that throughout the recursion on $G(\Phi)$, every \SSV\ winner that can affect the recursion is also undefeated in that subtournament. Hence, \SV\ never discards the largest critical edge that \SSV\ would use, and so \SV\ and \SSV\ make the same decision at every step. We proceed in three steps: we first characterize the types of subtournaments that the recursion can actually reach (\Cref{claim:reachability}); we then show that on those subtournaments the \SSV\ winner is undefeated whenever it is not a literal candidate (\Cref{lem:sv-undefeated}); and finally we show that subtournaments won by literal candidates never change the recursion (\Cref{lem:literals-immaterial}). Hence, \SV\ and \SSV\ agree wherever it matters (\Cref{lem:sv-agrees-ssv}).

Recall (\Cref{def:does-not-defeat}) that $B$ does not defeat $A$ in $G$ if there is a (directed) path $A=x_0,x_1,\ldots,x_t=B$ such that $w(x_i,x_{i+1})\ge w(B,A)$ for all $i$, and that $A$ is undefeated in $G$ if no candidate defeats it. We use the fact that, if $w(A,B)>0$, then $B$ trivially does not defeat $A$ (the edge $A\to B$ is itself such a path). Hence, to show that $A$ is undefeated, it suffices to check every candidate $B$ with $w(A,B)<0$. Since the weight groups satisfy $\group1 > \group2 > \cdots > \group12$, a path certifies that $B$ does not defeat $A$ as long as every edge it contains lies in a group with index strictly smaller than that of the edge $B \to A$. All the paths exhibited in this section are of this form, so no comparisons within a group are ever needed. 
Conversely, since a certifying path must leave $A$ along an edge of weight at least $w(B,A)$, and must enter $B$ along such an edge, we get a simple sufficient condition for a defeat. All omitted proofs can be found in \Cref{apx:sv-extension}.

\begin{observation}[Blocking]\label{obs:blocking}
Let $G$ be a weighted tournament and let $A,B\in V(G)$ with $w(B,A)>0$. If every edge of $G$ leaving $A$, or every edge of $G$ entering $B$, has weight smaller than $w(B,A)$, then $B$ defeats $A$.
\end{observation}

\paragraph{Reachable subtournaments.}\label{subsec:reachability}
During the evaluation of $G(\Phi)$ according to \SSV\ (or \SV), only a small subset of subtournaments is reached: those evaluated before the scan finds its first critical pair. Formally, the recursion on $G$ evaluates $\SSV(G\setminus\{B\})$ exactly for the pairs $(A,B)$ of weight at least that of the maximum critical edge of $G$. Call these subtournaments, together with $G(\Phi)$ itself, the ones \emph{arising in the recursive evaluation} of $\SSV(G(\Phi))$. The next claim shows that only three types of subtournaments ever arise: prefix graphs, graphs whose winner is forced to be a guard, and admissible graphs. In particular, the recursion never reaches a subtournament that has $\first(\cdot)=0$ while containing both literals of some variable with a smaller index (i.e., the induced subtournament on $\{\t_1,\f_1,\t_2,\guard_2,\clause_1\}$ which is neither a prefix graph nor admissible).

\begin{restatable}[Reachability]{claim}{Reachability}\label{claim:reachability}
    Every induced subtournament that arises in the recursive evaluation of $\SSV(G(\Phi))$ belongs to at least one
    of the following families:
    \begin{itemize}
        \item[($\mathcal{P}$)] the prefix graphs (\Cref{def:prefix-subgraph});
        \item[($\mathcal{A}$)] the admissible graphs (\Cref{def:admissible-graph}).
        \item[($\mathcal{F}$)] the graphs $H$ with $\first(H) \neq 0$;
    \end{itemize}
\end{restatable}

\begin{corollary}\label{cor:families}
    Let $G$ be an induced subtournament of $G(\Phi)$ that arises in the recursive evaluation of
    $\SSV(G(\Phi))$.
    \begin{enumerate}
        \item If no guard candidate is present in $G$, then $G$ is admissible.
        \item If $\SSV(G)$ is a clause candidate or a literal candidate, then $G$ is admissible.
    \end{enumerate}
\end{corollary}

\begin{proof}
    By \Cref{claim:reachability}, $G$ is a prefix graph, has $\first(G)\neq0$, or is admissible. Every prefix graph contains all guard candidates, and every graph with $\first(G)=\is\neq0$ contains the guard $\guard_{\is}$; this proves the first part. For the second part: every prefix graph is won by $\winner$ or $\loser$ (statement $P(j)$ in the proof of \Cref{thm:main_result}), and every graph with $\first(G)\neq0$ is won by a guard (\Cref{lem:guards}); so a clause or literal winner is only possible when $G$ is admissible.
\end{proof}

\paragraph{Winners of admissible graphs.}
The recursion below the full-assignment graphs runs on admissible graphs, whose winners were characterized in \Cref{sec:hardness} and \Cref{apx:missing-proofs}. The next claim characterizes the different types of admissible graphs depending on the winner and the candidates present. 

\begin{restatable}[Winners of admissible graphs]{claim}{winnersAdmissible}\label{claim:admissible-winners}
    Let $G$ be an admissible graph and let $A = \SSV(G)$, then the following hold.
    \begin{enumerate}[label=(\alph*)]
        \item If $\winner, \loser \in V(G)$, then $A \in \{\winner, \loser\}$.
        \item If $A = \winner$, then $\firstU(G) = 0$; that is, every present clause is satisfied by a present literal.
        \item If $A = \clause_{\ks}$ for some clause candidate, then $\loser \notin V(G)$ and $\firstU(G) = \ks$; in particular, $\clause_{\ks}$ is unsatisfied and every present clause $\clause_{k'}$ with $k' < \ks$ is satisfied by a present literal.
        \item If $A$ is a literal candidate, then $\winner, \loser \notin V(G)$ and $\firstU(G) = 0$.
    \end{enumerate}
\end{restatable}

\paragraph{\SSV\ winners are undefeated.}
We now show that when $\SSV(G) \in \{\winner, \loser, \guard_{\is}, \clause_{\ks}\}$, the winner is undefeated in $G$, so that \SV\ does not discard the maximum critical edge of $G$. We make use of \Cref{fig:dnd-paths} for a visual representation of the does-not-defeat paths used in the following lemma.

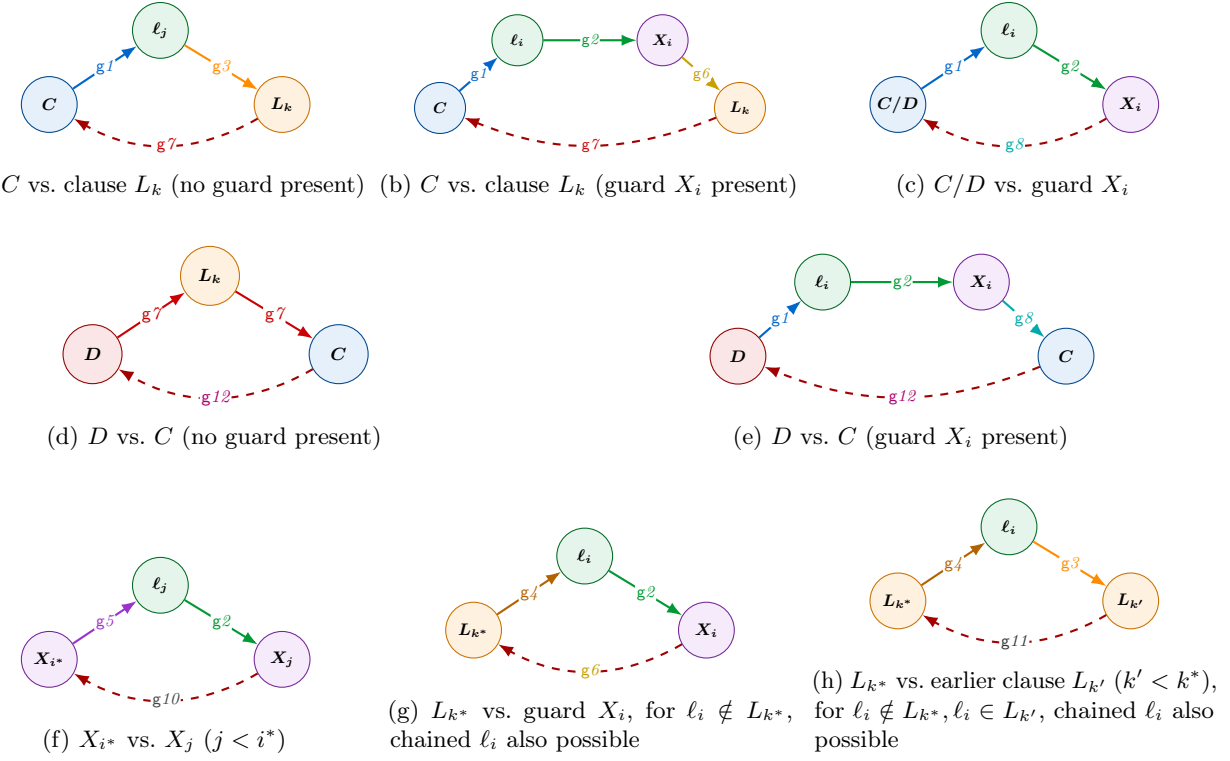
\begin{figure}[ht]
    \centering
    \tikzset{
        gnode/.style={circle, minimum size=3.0em, inner sep=1pt, font=\boldmath\small},
        winnernode/.style ={gnode, draw=winnerblue!70!black,   fill=winnerblue!10},
        losernode/.style   ={gnode, draw=loserred!70!black,    fill=loserred!10},
        literalnode/.style ={gnode, draw=literalgreen!70!black, fill=literalgreen!10},
        clausenode/.style  ={gnode, draw=clauseorange!80!black, fill=clauseorange!12},
        guardnode/.style   ={gnode, draw=guardpurple!75!black,  fill=guardpurple!10},
        edge/.style={->, >={{Latex[length=2mm, width=1.5mm]}}, thick},
        weight/.style={font=\small\itshape, fill=white, inner sep=1pt, rounded corners=1pt},
        edge1/.style={edge, draw=gOne},
        edge2/.style={edge, draw=gTwo},
        edge3/.style={edge, draw=gThree},
        edge4/.style={edge, draw=gFour},
        edge5/.style={edge, draw=gFive},
        edge6/.style={edge, draw=gSix},
        edge7/.style={edge, draw=gSeven},
        edge8/.style={edge, draw=gEight},
        edge10/.style={edge, draw=gTen},
        edge11/.style={edge, draw=gEleven},
        edge12/.style={edge, draw=gTwelve},
        defeats/.style={->, >={{Latex[length=2mm, width=1.5mm]}}, thick,
                        draw=loserred!80!black, dashed},
    }

    \begin{subfigure}[b]{0.32\textwidth}
        \centering
        \begin{tikzpicture}[scale=0.7, transform shape]
            \node[winnernode] (C)  at (0,0)      {$\winner$};
            \node[literalnode] (l) at (2.1,1.4)  {$\lit_j$};
            \node[clausenode]  (L) at (4.4,0)    {$\clause_k$};
            \draw[edge1] (C) -- (l) node[midway, weight] {\groupone};
            \draw[edge3] (l) -- (L) node[midway, weight] {\groupthree};
            \draw[defeats] (L) to[bend left=30] node[midway, weight] {\groupseven} (C);
        \end{tikzpicture}
        \caption{$\winner$ vs.\ clause $\clause_k$ (no guard present)}
        \label{fig:dnd-C-clause}
    \end{subfigure}
    \hfill
    \begin{subfigure}[b]{0.34\textwidth}
        \centering
        \begin{tikzpicture}[scale=0.64, transform shape]
            \node[winnernode]  (C)  at (0,0)      {$\winner$};
            \node[literalnode] (l)  at (1.6,1.4)  {$\lit_i$};
            \node[guardnode]   (X)  at (4.6,1.4)  {$\guard_i$};
            \node[clausenode]  (L)  at (6.2,0)    {$\clause_k$};
            \draw[edge1] (C) -- (l) node[midway, weight] {\groupone};
            \draw[edge2] (l) -- (X) node[midway, weight] {\grouptwo};
            \draw[edge6] (X) -- (L) node[midway, weight] {\groupsix};
            \draw[defeats] (L) to[bend left=22] node[midway, weight] {\groupseven} (C);
        \end{tikzpicture}
        \caption{$\winner$ vs.\ clause $\clause_k$ (guard $\guard_i$ present)}
        \label{fig:dnd-C-clause-guard}
    \end{subfigure}
    \hfill
    \begin{subfigure}[b]{0.32\textwidth}
        \centering
        \begin{tikzpicture}[scale=0.7, transform shape]
            \node[winnernode]  (C)  at (0,0)     {$\winner/\loser$};
            \node[literalnode] (l)  at (2.1,1.4) {$\lit_i$};
            \node[guardnode]   (X)  at (4.4,0)   {$\guard_i$};
            \draw[edge1] (C) -- (l) node[midway, weight] {\groupone};
            \draw[edge2] (l) -- (X) node[midway, weight] {\grouptwo};
            \draw[defeats] (X) to[bend left=30] node[midway, weight] {\groupeight} (C);
        \end{tikzpicture}
        \caption{$\winner/\loser$ vs.\ guard $\guard_i$}
        \label{fig:dnd-CD-guard}
    \end{subfigure}

    \vspace{1.6em}

    \begin{subfigure}[b]{0.40\textwidth}
        \centering
        \begin{tikzpicture}[scale=0.74, transform shape]
            \node[losernode]  (D)  at (0,0)      {$\loser$};
            \node[clausenode] (L)  at (2.1,1.4)  {$\clause_k$};
            \node[winnernode] (C)  at (4.4,0)    {$\winner$};
            \draw[edge7] (D) -- (L) node[midway, weight] {\groupseven};
            \draw[edge7] (L) -- (C) node[midway, weight] {\groupseven};
            \draw[defeats] (C) to[bend left=30] node[midway, weight] {\grouptwelve} (D);
        \end{tikzpicture}
        \caption{$\loser$ vs.\ $\winner$ (no guard present)}
        \label{fig:dnd-D-C}
    \end{subfigure}
    \hfill
    \begin{subfigure}[b]{0.50\textwidth}
        \centering
        \begin{tikzpicture}[scale=0.7, transform shape]
            \node[losernode]   (D)  at (0,0)      {$\loser$};
            \node[literalnode] (l)  at (1.6,1.4)  {$\lit_i$};
            \node[guardnode]   (X)  at (4.6,1.4)  {$\guard_i$};
            \node[winnernode]  (C)  at (6.2,0)    {$\winner$};
            \draw[edge1] (D) -- (l) node[midway, weight] {\groupone};
            \draw[edge2] (l) -- (X) node[midway, weight] {\grouptwo};
            \draw[edge8] (X) -- (C) node[midway, weight] {\groupeight};
            \draw[defeats] (C) to[bend left=22] node[midway, weight] {\grouptwelve} (D);
        \end{tikzpicture}
        \caption{$\loser$ vs.\ $\winner$ (guard $\guard_i$ present)}
        \label{fig:dnd-D-C-guard}
    \end{subfigure}

    \vspace{1.6em}

    \begin{subfigure}[b]{0.32\textwidth}
        \centering
        \begin{tikzpicture}[scale=0.7, transform shape]
            \node[guardnode]   (Xi) at (0,0)     {$\guard_{\is}$};
            \node[literalnode] (l)  at (2.1,1.4) {$\lit_j$};
            \node[guardnode]   (Xj) at (4.4,0)   {$\guard_j$};
            \draw[edge5] (Xi) -- (l)  node[midway, weight] {\groupfive};
            \draw[edge2] (l)  -- (Xj) node[midway, weight] {\grouptwo};
            \draw[defeats] (Xj) to[bend left=30] node[midway, weight] {\groupten} (Xi);
        \end{tikzpicture}
        \caption{$\guard_{\is}$ vs.\ $\guard_{j}$ ($j<\is$)}
        \label{fig:dnd-guard-guard}
    \end{subfigure}
    \hfill
    \begin{subfigure}[b]{0.32\textwidth}
        \centering
        \begin{tikzpicture}[scale=0.7, transform shape]
            \node[clausenode]  (Lk) at (0,0)     {$\clause_{\ks}$};
            \node[literalnode] (l)  at (2.1,1.4) {$\lit_i$};
            \node[guardnode]   (X)  at (4.4,0)   {$\guard_i$};
            \draw[edge4] (Lk) -- (l) node[midway, weight] {\groupfour};
            \draw[edge2] (l)  -- (X) node[midway, weight] {\grouptwo};
            \draw[defeats] (X) to[bend left=30] node[midway, weight] {\groupsix} (Lk);
        \end{tikzpicture}
        \caption{$\clause_{\ks}$ vs.\ guard $\guard_i$, for $\lit_i \notin \clause_\ks$, chained $\lit_i$ also possible}
        \label{fig:dnd-clause-guard}
    \end{subfigure}
    \hfill
    \begin{subfigure}[b]{0.32\textwidth}
        \centering
        \begin{tikzpicture}[scale=0.7, transform shape]
            \node[clausenode]  (Lk) at (0,0)     {$\clause_{\ks}$};
            \node[literalnode] (l)  at (2.1,1.4) {$\lit_i$};
            \node[clausenode]  (Lp) at (4.4,0)   {$\clause_{k'}$};
            \draw[edge4] (Lk) -- (l)  node[midway, weight] {\groupfour};
            \draw[edge3] (l)  -- (Lp) node[midway, weight] {\groupthree};
            \draw[defeats] (Lp) to[bend left=30] node[midway, weight] {\groupeleven} (Lk);
        \end{tikzpicture}
        \caption{$\clause_{\ks}$ vs.\ earlier clause $\clause_{k'}$ ($k'<\ks$), for $\lit_i \notin \clause_\ks, \lit_i \in \clause_{k'}$, chained $\lit_i$ also possible}
        \label{fig:dnd-clause-clause}
    \end{subfigure}

    \caption{Does-not-defeat paths verifying that each \SSV\ winner that is not a literal candidate is undefeated (\Cref{lem:sv-undefeated}), grouped by winner: $\winner$ (top row, together with the guard attacker shared by $\winner$ and $\loser$), $\loser$ (middle row), and the structural winners $\guard_{\is}$ and $\clause_{\ks}$ (bottom row). In each subfigure, the dashed edge is the edge by which the opponent beats the winner; the solid, group-colored path is the does-not-defeat witness, since every edge on the path lies in a group of strictly smaller index than the dashed edge. In Cases 1 and 2 of \Cref{lem:sv-undefeated}, the applicable witness depends on whether a guard candidate is present in $G$: if so, the three-edge routes of (\subref{fig:dnd-C-clause-guard}) and (\subref{fig:dnd-D-C-guard}) apply; if not, $G$ is admissible and the two-edge routes of (\subref{fig:dnd-C-clause}) and (\subref{fig:dnd-D-C}) apply.}
    \label{fig:dnd-paths}
\end{figure}

\begin{lemma}[\SSV\ winners are undefeated]\label{lem:sv-undefeated}
Let $G$ be any induced subtournament of $G(\Phi)$ that arises in the recursive evaluation of $\SSV(G(\Phi))$, and suppose that $\SSV(G)$ is not a literal candidate. Then $\SSV(G)$ is undefeated in $G$.
\end{lemma}

\begin{proof}
    Fix such a subtournament $G$, and let $A = \SSV(G)$ be its winner; by assumption, $A$ is one of the designated candidates $\winner$ or $\loser$, a guard $\guard_{\is}$, or a clause $\clause_{\ks}$. We treat each case separately, showing that for every candidate $B$ that beats $A$ head-to-head, there is a does-not-defeat path $A \rightsquigarrow B$ all of whose edges have strictly larger weight than $B \to A$. In each case, we also show that the candidates of the path are present in $G$; \Cref{fig:dnd-paths} shows all such paths. Throughout, note that whenever the winner is not a guard, \Cref{lem:guards} gives $\first(G) = 0$, so every present guard $\guard_i$ has one of its literals present; we denote such a literal by $\lit_i$.

    \begin{enumerate}
        \item ($\SSV(G) = \winner$.) The only in-edges to $\winner$ are from clauses $\clause_k$
        (group $\group7$) and from guards $\guard_i$ (group $\group8$). For a guard $\guard_i$, the path is
        $\winner \to \lit_i \to \guard_i$ ($\group1,\group2$); $\lit_i$ must be present for $\guard_i$ to not win. For a clause $\clause_k$, we distinguish two
        subcases.
        \begin{itemize}
            \item If some guard $\guard_i$ is present in $G$, the path is $\winner \to \lit_i \to \guard_i \to \clause_k$ ($\group1,\group2,\group6$).
            \item If no guard is present in $G$, then $G$ is admissible by \Cref{cor:families}, and by \Cref{claim:admissible-winners}(b) the clause $\clause_k$ is satisfied by some present literal $\lit_j$; the path is $\winner \to \lit_j \to \clause_k$ ($\group1,\group3$).
        \end{itemize}
        All path edges lie in groups of smaller index than the edges into $\winner$, so $\winner$ is undefeated.
        
        \item  ($\SSV(G) = \loser$.) The only in-edges to $\loser$ are from $\winner$ (group $\group{12}$) and from guards $\guard_i$ (group $\group8$). For a guard $\guard_i$, the path is $\loser \to \lit_i \to \guard_i$ ($\group1,\group2$); $\lit_i$ must be present for $\guard_i$ to not win. Now suppose $\winner \in V(G)$; we show that $\winner$ does not defeat $\loser$, in two subcases.
        \begin{itemize}
            \item If some guard $\guard_i$ is present in $G$, the path is $\loser \to \lit_i \to \guard_i \to \winner$ ($\group1,\group2,\group8$), and all three groups have index smaller than $12$.
            \item If no guard is present in $G$, then $G$ is admissible by \Cref{cor:families}. For $\loser$ to win in $G$, a clause candidate must then be present; otherwise, by  \Cref{lem:c-wins-without-clauses} we would get $\SSV(G) = \winner \neq \loser$ for an admissible graph with $\winner, \loser \in V(G)$ and $\clause_k\notin V(G).$ The path $\loser \to \clause_k \to \winner$ ($\group7,\group7$) then shows that $\winner$ does not defeat $\loser$.
            
        \end{itemize}
        All path edges lie in groups of smaller index than the edges into $\loser$, so $\loser$ is undefeated.
        
        \item ($\SSV(G) = \guard_{\is}$.) By \Cref{lem:guards}, $\first(G) = \is$, so the literals $\t_{\is}, \f_{\is}$ are absent, and the edges $\{\t_{\is}, \f_{\is}\}\to\guard_{\is}$ (group $\group2$) are gone. The only remaining in-edges to $\guard_{\is}$ are from guards $\guard_j$ with $j < \is$ (group $\group{10}$). For each such $\guard_j$, one of its literals $\lit_j$ is present, since $j < \is = \first(G)$. Hence, the path $\guard_{\is} \to \lit_j \to \guard_j$ has edges in groups $\group5$ and $\group2$, both of index smaller than $10$. Hence $\guard_j$ does not defeat $\guard_{\is}$, and $\guard_{\is}$ is undefeated in $G$.
        
        \item ($\SSV(G) = \clause_{\ks}$.) By \Cref{cor:families}, $G$ is admissible, and by \Cref{claim:admissible-winners}(c), $\loser \notin V(G)$ and $\firstU(G) = \ks$. Since $\clause_{\ks}$ is unsatisfied, there is no edge $\lit_i\to \clause_\ks$ in group $\group3$. So the only candidates beating $\clause_{\ks}$ are the guards $\guard_i$ (group $\group6$) and the earlier clauses $\clause_{k'}$ with $k' < \ks$ (group $\group{11}$). Since $\clause_{\ks}$ is unsatisfied, the edge $\clause_{\ks} \to \lit_i$ is positive (group $\group4$) for every present literal $\lit_i$; and $\first(G) = 0$, so every present guard has its literal present. Hence, for a guard $\guard_i$, the path is $\clause_{\ks} \to \lit_i \to \guard_i$ ($\group4,\group2$). For an earlier clause $\clause_{k'}$, it is satisfied (as $k' < \ks = \firstU(G)$), so a satisfying literal $\lit_j$ gives the path $\clause_{\ks} \to \lit_j \to \clause_{k'}$ ($\group4,\group3$). All path edges lie in groups of smaller index than the edges into $\clause_{\ks}$, so $\clause_{\ks}$ is undefeated.\qedhere
    \end{enumerate}
\end{proof}

\paragraph{Literal candidates are immaterial}

It remains to deal with the subtournaments whose \SSV\ winner is a literal candidate. Such subtournaments do arise in the recursion (e.g., deleting $\winner$ and $\loser$ from a small admissible graph makes a literal win) and the winning literal need not be undefeated. The next lemma shows that these subtournaments can never cause \SV\ and \SSV\ to disagree in our construction. Whenever the winner is not a literal, no literal is even eligible as an \SV\ winner since they are defeated. Whenever the \SSV\ winner is a literal, the \SV\ winner is also a literal (though possibly a different one) since they are the only undefeated candidates (which \SV\ must select from). This is all that the comparison in \Cref{lem:sv-agrees-ssv} requires.

\begin{lemma}[Literal candidates and undefeatedness]\label{lem:literals-immaterial}
    Let $G$ be an induced subtournament of $G(\Phi)$ that arises in the recursive evaluation of $\SSV(G(\Phi))$.
    \begin{enumerate}[label=(\alph*)]
        \item If $\SSV(G)\neq \lit_i$ for some literal candidate $\lit_i$, then $\SSV(G)$ defeats \textbf{every} literal candidate in $V(G)$.
        \item If $\SSV(G)=\lit_i$ for some literal candidate $\lit_i$, then \textbf{every} undefeated candidate of $G$ is a literal candidate, hence $\SV(G)$ is a literal candidate (potentially different than $\lit_i$).
    \end{enumerate}
\end{lemma}

\begin{proof}
    (a) Let $A = \SSV(G)$ and let $\lit_i \in V(G)$ be any literal candidate for some $i\in [n]$. In each case below, every incoming edge to $A$ lies in a group of strictly larger index (hence smaller weight) than the edge $A \to \lit_i$, so $A$ defeats $\lit_i$ by \Cref{obs:blocking}.
    \begin{itemize}
        \item $A = \winner$: the edge $\winner \to \lit_i$ lies in group $\group1$; the in-edges of $\winner$ lie in groups $\group7$ and $\group8$.
        \item $A = \loser$: the edge $\loser \to \lit_i$ lies in group $\group1$; the in-edges of $\loser$ lie in groups $\group{12}$ and $\group8$.
        \item $A = \guard_{\is}$: by \Cref{lem:guards}, $\first(G) = \is$, so $\t_{\is},\f_{\is} \notin V(G)$, and every present literal belongs to another variable; hence the edge $\guard_{\is} \to \lit_i$ lies in group $\group5$. The in-edges of $\guard_{\is}$ from its own literals (group $\group2$) are absent, leaving only the edges from guards $\guard_j$ with $j<\is$, in group $\group{10}$.
        \item $A = \clause_{\ks}$: by \Cref{cor:families}, $G$ is admissible, and by \Cref{claim:admissible-winners}(c), $\clause_{\ks}$ is unsatisfied and $\loser \notin V(G)$; hence the edge $\clause_{\ks} \to \lit$ lies in group $\group4$ for every present literal $\lit$. The in-edges of $\clause_{\ks}$ from satisfying literals (group $\group3$) and from $\loser$ (group $\group7$) are absent, leaving only the edges from guards (group $\group6$) and earlier clauses (group $\group{11}$).
    \end{itemize}
    
    \noindent(b) Let $\SSV(G)=\lit_i$ for some $i\in[n]$. By \Cref{cor:families}, $G$ is admissible, and by \Cref{claim:admissible-winners}(d), $\winner, \loser \notin V(G)$ and every present clause is satisfied by a present literal ($\firstU(G)=0$); moreover $\first(G) = 0$ since $G$ is an admissible graph. We show that every guard and every clause of $G$ is defeated (which only requires one candidate to defeat them, as opposed to being undefeated by \emph{any} other candidate); since $\winner$ and $\loser$ are absent, only literal candidates can then be undefeated.
    \begin{itemize}
        \item Let $\guard_j \in V(G)$. Its literal $\lit_j$ is present, and the edge $\lit_j\to \guard_j$ is in group $\group2$. Every positive edge leaving $\guard_j$ lies in group $\group5$ (other literals), $\group6$ (clauses), or $\group{10}$ (later guards)  all strictly smaller than $\group2$. By \Cref{obs:blocking}, $\lit_j$ defeats $\guard_j$.
        \item Let $\clause_k \in V(G)$, and let $\lit_j$ be a present literal satisfying it; the edge $\lit_j \to \clause_k$ lies in group $\group3$. Every positive edge leaving $\clause_k$ lies in group $\group4$ (non-satisfying literals) or $\group{11}$ (later clauses) all strictly smaller than $\group3$. By \Cref{obs:blocking}, $\lit_j$ defeats $\clause_k$.
    \end{itemize}
    Finally, by \Cref{def:does-not-defeat} and the definition of \SV, the candidate $\SV(G)$ is undefeated in $G$; since literals are the only undefeated candidates, $\SV(G)$ is a literal candidate.
\end{proof}

\paragraph{\SV\ agrees with \SSV.}
Having shown that every \SSV\ winner that is not a literal is undefeated, we now use this to prove that \SV\
and \SSV\ agree on every subtournament that the recursion reaches.

\begin{lemma}[\SV\ agrees with \SSV]\label{lem:sv-agrees-ssv}
    Let $G$ be any induced subtournament of $G(\Phi)$ that arises in the recursive evaluation of $\SSV(G(\Phi))$. If $\SSV(G)\in\{\winner,\loser,\clause_k,\guard_i:k\in[m], i\in[n]\}$, then $\SV(G) = \SSV(G)$; otherwise, $\SV(G),\SSV(G)\in\{\lit_i: i\in[n]\}$.
\end{lemma}

\begin{proof}
    The second assertion is \Cref{lem:literals-immaterial}(b). We prove the first by strong induction on $|V(G)|$.
    
    \paragraph{Base case} ($|V(G)| = 1$). Both rules return the unique candidate, so $\SV(G) = \SSV(G)$.
    
    \paragraph{Induction hypothesis.} Assume the statement holds for every induced subtournament $G'$ arising in the recursive evaluation of $\SSV(G(\Phi))$ with $|V(G')| < |V(G)|$.
    
    \paragraph{Induction step.} Let $A= \SSV(G)$ be a non-literal winner, determined by the maximum-weight critical edge $A\to B$; that is, $\SSV(G\setminus\{B\}) = A$. By \Cref{lem:sv-undefeated}, $A$ is undefeated in $G$.
    
    If $G$ has a unique undefeated candidate $U$, then $A= U$, and \SV\ returns $U = A$ directly, so $\SV(G) = \SSV(G)$. Otherwise, recall that \SV\ scans the ordered pairs $(U,V)$ with $U$ undefeated in $G$ in decreasing $w(U,V)$, returning the first with $\SV(G\setminus\{V\}) = U$. We make two observations. First, since \SSV\ scans \emph{all} ordered pairs in decreasing weight down to $(A,B)$, every pair with weight at least $w(A,B)$ was scanned by \SSV, so for every such pair $(U,V)$ the subtournament $G\setminus\{V\}$ arises in the recursive evaluation of $\SSV(G(\Phi))$ and the induction hypothesis applies to it. Second, by \Cref{lem:literals-immaterial}(a), no literal candidate is undefeated in $G$, so every tail $U$ that \SV\ scans is a non-literal candidate.
    
    The pair $(A, B)$ is critical for \SV, since the winner $A$ is undefeated in $G$, and
    $\SSV(G\setminus\{B\}) = A$ is not a literal, the induction hypothesis gives
    $\SV(G\setminus\{B\}) = \SSV(G\setminus\{B\}) = A$. 
    
    It remains to show that \SV\ does not stop at any earlier pair. Suppose, for contradiction, that \SV\ stops at a pair $(U,V)$ with $w(U,V) > w(A,B)$; then $U$ is undefeated in $G$ (hence not a literal) and $\SV(G\setminus\{V\}) = U$. The subtournament $G\setminus\{V\}$ arises in the evaluation of $\SSV(G(\Phi))$, so the induction hypothesis applies to it, and we consider two cases. If $\SSV(G\setminus\{V\})$ is a literal candidate, then $\SV(G\setminus\{V\})$ is also a literal candidate, contradicting $\SV(G\setminus\{V\}) = U$ with $U$ not a literal. If $\SSV(G\setminus\{V\})$ is not a literal, then $\SSV(G\setminus\{V\}) = \SV(G\setminus\{V\}) = U$, so $(U,V)$ is a critical edge of $G$ for \SSV\ with weight larger than $w(A,B)$, contradicting the maximality of $(A,B)$.
    
    Hence \SV\ stops exactly at $(A,B)$, giving $\SV(G) = A= \SSV(G)$.
\end{proof}

\begin{theorem}\label{thm:sv-pspace}
    Winner determination for Stable Voting is \pspace-complete.
\end{theorem}
\begin{proof}
    Membership in \pspace\ follows as in \Cref{thm:main_result}: \SV\ admits the same depth-first recursion, and the additional test of whether a candidate is undefeated is a reachability computation that requires only polynomial space. For hardness, let $\Phi$ be a TQBF instance and let $G(\Phi)$ be the tournament of \Cref{margin_groups}. The tournament $G(\Phi)$ trivially arises in the recursive evaluation of $\SSV(G(\Phi))$, and by \Cref{thm:main_result} its winner $\SSV(G(\Phi))$ is $\winner$ or $\loser$ (in particular, not a literal candidate). Hence, by \Cref{lem:sv-agrees-ssv}, $\SV(G(\Phi)) = \SSV(G(\Phi))$, and so $\SV(G(\Phi)) = \winner$ if and only if $\Phi$ is true. The same polynomial-time reduction therefore establishes \pspace-hardness for \SV.
\end{proof}

\section{Discussion}

In this paper, we prove that winner determination is a \pspace-complete problem under both Stable Voting and Simple Stable Voting. In fact, our reduction establishes hardness for a somewhat stronger, promise version: even when we are promised that the winner is one of two distinguished candidates, $\winner$ or $\loser$, it is \pspace-hard to decide whether the winner is $\winner$. Thus, hardness persists even when the outcome has been narrowed down to a binary choice.

Several natural directions remain open. First is to identify tractable special cases. Since the recursive definitions of SV and SSV can be evaluated by brute force on a small candidate set, winner determination is fixed-parameter tractable in the number of relevant candidates. This suggests looking for structural conditions under which the recursion can be confined to a small subset. More generally, it would be interesting to find graph-theoretic or margin-based assumptions that allow for tractable winner determination.

A second direction is to understand the average-case complexity of winner determination. Our construction relies on carefully chosen edge orientations and weights. This leaves open whether natural stochastic models of preferences or pairwise margins lead to polynomial-time winner determination with high probability.

Finally, one can study the associated search problem: given a tournament, output an \SV\ or \SSV\ winner. This problem is total, since the rules always select a winner. But, it is not obviously a $\mathsf{TFNP}$ problem: $\mathsf{TFNP}$ concerns total search problems whose solutions are polynomial-time verifiable, whereas deciding whether a proposed candidate is the \SV\ or \SSV\ winner is \pspace-complete. 
We conjecture that this search problem is complete for a natural total-search class beyond $\mathsf{TFNP}$. Recent work has introduced several such classes~\cite{kleinberg2021total,pasarkar2023extremal}; identifying the right class that captures recursive voting rules such as \SV\ and \SSV\ is an interesting open problem.

\paragraph{Implementation.} We implemented \SV\ and \SSV\ in C++~\cite{stablevotingcode} to check our construction as we developed it and to see how many candidates exact winner determination can handle in practice. Our solver memoizes the winner of each candidate subset it reaches during the recursion. It agrees with the reference package \texttt{pref\_voting}~\cite{holliday2025prefvoting} on every winner we tested, and is one or two orders of magnitude faster on random uniquely weighted tournaments. Even so, it reaches only about $30$ candidates, in line with the range reported in~\cite{Hol25}: the memo has one entry per subset, so both time and space double with each additional candidate. Memoizing trades the exponential running time of the depth-first recursion of \Cref{thm:main_result} for exponential space; the theorem says that, unless $\mathsf{P} = \pspace$, no implementation avoids both.

\section*{Acknowledgements}
The authors used ChatGPT as a sounding board for quickly exploring and stress-testing possible directions the authors were considering during gadget development, as well as for proofreading, identifying potential issues in proofs and figures, and assisting with portions of the C++ testing code, which was used by the authors to evaluate candidate constructions and inform the development of the final construction. The authors made all substantive design decisions, and reviewed and verified all AI-assisted material. The authors take full responsibility for all claims and results in the paper.

\clearpage
\bibliographystyle{alpha}
\bibliography{ref}

\clearpage
\appendix
\section{Constructing a Voting Profile}\label{apx:vot-profile} 

In \Cref{sec:prelims} we noted that any weighted tournament can be induced by a voting profile after appropriately rescaling the weights. Here we present this construction formally, using classic techniques of McGarvey~\cite{McGarvey53} and Debord~\cite{debord1987caracterisation}.

Recall that \SSV\ and \SV\ depend only on the relative order of their margins, so rescaling the weights does not change the winner as long as their order is preserved. Let $N=|V(G(\Phi))|$ Our construction produces a weighted tournament $G(\Phi)$ with weights $1,\ldots, \binom{N}{2}$. For a weighted tournament to be realizable by a voting profile, all weights must have the same parity. Multiplying every weight by 2 preserves their order and makes all of them even.

\paragraph{Construction} Let $S=\{A_1,\ldots,A_N\}$  be a set of $N$ alternatives and $G$ a weighted tournament (with all even weights). Fix an arbitrary ordering of the alternatives. For any positive edge $A\to B$ with weight $2x$, let $C_1, \ldots, C_{N-2}$ be the remaining alternatives in this fixed order, and add the following $2x$ voters:
\begin{itemize}
    \item $x$ voters with preferences $A\succ B \succ C_1 \succ \ldots \succ
    C_{N-2}$
    \item $x$ voters with preferences $C_{N-2} \succ \ldots \succ
    C_{1} \succ A \succ B$
\end{itemize}
All $2x$ voters rank $A$ above $B$, so this contributes $+2x$ to the margin of $A$ over $B$, matching the edge weight. For every other pair of alternatives $(U,W)$ there are $x$ voters with $U\succ W$; the two groups rank them in opposite orders, so $x$ voters have $U\succ W$ and $x$ have $W\succ U$, contributing $0$ to that margin. Thus each edge is realized independently of the others. Repeating this for every edge, the total number of voters is
\[
 \sum_{i=1}^{\binom{N}{2}} 2i \;=\; \binom{N}{2}\!\left(\binom{N}{2}+1\right) \;=\; O(N^4),
 \]
 and each ballot has length $N$, thus the profile is constructed in polynomial time. Since its margins induce the same edge order as $G(\Phi)$, the profile has the same \SSV\ and \SV\ winner.

\section[Edge Order for Tournament in Figure 1]{Edge Order for Tournament in \protect\Cref{fig:almost-condorcet}}\label{apx:table}

For completeness, \Cref{tab:ranks} gives the order of edges in the full weighted tournament shown in \Cref{fig:almost-condorcet}, from which the \SSV\ winner can be verified (see \texttt{fig1\_tournament.cpp} in~\cite{stablevotingcode}.)

\begin{table}[ht]
\centering
\small
\setlength{\tabcolsep}{4pt}
\begin{tabular}{c|rrrrrrrrrrrrrrr}
 & $c_{1}$ & $c_{2}$ & $c_{3}$ & $c_{4}$ & $c_{5}$ & $c_{6}$ & $c_{7}$ & $c_{8}$ & $c_{9}$ & $c_{10}$ & $c_{11}$ & $c_{12}$ & $c_{13}$ & $c_{14}$ & $c_{15}$ \\ \hline
$c_{1}$ & . & 75 & 81 & 76 & 77 & 1 & -39 & -43 & -47 & -17 & 6 & 8 & 11 & 13 & 15 \\
$c_{2}$ & -75 & . & 78 & 82 & 83 & -35 & 2 & -44 & -48 & -18 & 7 & 9 & -26 & -29 & 16 \\
$c_{3}$ & -81 & -78 & . & 79 & 80 & -36 & -40 & 3 & -49 & -19 & -21 & 10 & -27 & 14 & -32 \\
$c_{4}$ & -76 & -82 & -79 & . & 84 & -37 & -41 & -45 & 4 & -20 & -22 & -24 & 12 & -30 & -33 \\
$c_{5}$ & -77 & -83 & -80 & -84 & . & -38 & -42 & -46 & -50 & 5 & -23 & -25 & -28 & -31 & -34 \\
$c_{6}$ & -1 & 35 & 36 & 37 & 38 & . & 85 & 86 & 87 & 51 & 52 & 53 & 54 & 55 & 56 \\
$c_{7}$ & 39 & -2 & 40 & 41 & 42 & -85 & . & 88 & 89 & 57 & 58 & 59 & 60 & 61 & 62 \\
$c_{8}$ & 43 & 44 & -3 & 45 & 46 & -86 & -88 & . & 90 & 63 & 64 & 65 & 66 & 67 & 68 \\
$c_{9}$ & 47 & 48 & 49 & -4 & 50 & -87 & -89 & -90 & . & 69 & 70 & 71 & 72 & 73 & 74 \\
$c_{10}$ & 17 & 18 & 19 & 20 & -5 & -51 & -57 & -63 & -69 & . & 91 & 92 & 93 & 94 & 95 \\
$c_{11}$ & -6 & -7 & 21 & 22 & 23 & -52 & -58 & -64 & -70 & -91 & . & 96 & 97 & 98 & 99 \\
$c_{12}$ & -8 & -9 & -10 & 24 & 25 & -53 & -59 & -65 & -71 & -92 & -96 & . & 100 & 101 & 102 \\
$c_{13}$ & -11 & 26 & 27 & -12 & 28 & -54 & -60 & -66 & -72 & -93 & -97 & -100 & . & 103 & 104 \\
$c_{14}$ & -13 & 29 & -14 & 30 & 31 & -55 & -61 & -67 & -73 & -94 & -98 & -101 & -103 & . & 105 \\
$c_{15}$ & -15 & -16 & 32 & 33 & 34 & -56 & -62 & -68 & -74 & -95 & -99 & -102 & -104 & -105 & . \\
\end{tabular}
\caption{The rank of the weights of the tournament in \Cref{fig:almost-condorcet}. Entry $(i,j)$ is the rank of the margin between $c_i$ and $c_j$ among all $105$ pairs, where rank $1$ is the largest margin; a positive sign means $c_i$ defeats
$c_j$. Because \SV\ and \SSV\ depend only on the order of the margins, this ranking alone determines the winner.}\label{tab:ranks}
\end{table}
\newpage
\section{Construction}\label{apx:code}
Algorithm~\ref{alg:build-tournament} gives an explicit procedure for constructing the tournament $G(\Phi)$ from a \TQBF\ instance $\Phi$.
\refstepcounter{algorithm}

\begin{breakablealgorithm}{\thealgorithm\ \textsc{BuildTournament}$(\Phi)$}
\label{alg:build-tournament}

\begin{algorithmic}[1]
\Require A prenex \TQBF\ instance $\Phi = Q_1x_1\cdots Q_nx_n\,\varphi$, where
$\varphi = C_1\wedge \cdots \wedge C_m$ is a CNF formula.
\Ensure A uniquely weighted tournament $G(\Phi)=(V(\Phi),\weight(\Phi))$.

\State $V(\Phi)\gets \{\winner,\loser\}
\cup \{\clause_k : k\in[m]\}
\cup \{\t_i,\f_i,\guard_i : i\in[n]\}.$

\State Let $W \gets \binom{|V(\Phi)|}{2}$.

\Procedure{Assign}{$A,B$}
    \State $\weight(A,B)\gets W$
    \State $\weight(B,A)\gets -W$
    \State $W\gets W-1$
\EndProcedure

\Statex
\Statex \textbf{Group $\mathsf g_1$: designated candidates to literals}
\For{$i=1$ to $n$}
    \If{$Q_i=\exists$}
        \State \Call{Assign}{$\winner,\f_i$}; \Call{Assign}{$\winner,\t_i$}; \Call{Assign}{$\loser,\f_i$}; \Call{Assign}{$\loser,\t_i$}
    \ElsIf{$Q_i=\forall$}
        \State \Call{Assign}{$\loser,\f_i$}; \Call{Assign}{$\loser,\t_i$}; \Call{Assign}{$\winner,\f_i$}; \Call{Assign}{$\winner,\t_i$}
    \EndIf
\EndFor

\Statex
\Statex \textbf{Group $\mathsf g_2$: literals to their own guards}
\For{$i=1$ to $n$}
    \State \Call{Assign}{$\f_i,\guard_i$}; \Call{Assign}{$\t_i,\guard_i$}
\EndFor

\Statex
\Statex \textbf{Group $\mathsf g_3$: literals to clauses they satisfy}
\For{$k=1$ to $m$}
    \For{$i=1$ to $n$}
        \If{$x_i$ appears positively in $C_k$}
            \State \Call{Assign}{$\t_i,\clause_k$}
        \ElsIf{$x_i$ appears negatively in $C_k$}
            \State \Call{Assign}{$\f_i,\clause_k$}
        \EndIf
    \EndFor
\EndFor

\Statex
\Statex \textbf{Group $\mathsf g_4$: clauses to literals not appearing in them}
\For{$k=1$ to $m$}
    \For{$i=1$ to $n$}  
        \If{$x_i$ appears positively in $C_k$}
            \State \Call{Assign}{$\clause_k,\f_i$}
        \ElsIf{$x_i$ appears negatively in $C_k$}
            \State  \Call{Assign}{$\clause_k,\t_i$}
        \EndIf
        \If{$x_i$ does not appear in $C_k$}
            \State \Call{Assign}{$\clause_k,\f_i$}; \Call{Assign}{$\clause_k,\t_i$}
        \EndIf
    \EndFor
\EndFor
\end{algorithmic}
\end{breakablealgorithm}

\begin{breakablealgorithm}{\thealgorithm\ \textsc{BuildTournament}$(\Phi)$ (continued)}

\begin{algorithmic}[1]
\Statex \textbf{Group $\mathsf g_5$: guards to other literals}
\For{$i=1$ to $n$}
    \For{$j=1$ to $n$}
        \If{$i\neq j$}
            \State \Call{Assign}{$\guard_i,\f_j$}; \Call{Assign}{$\guard_i,\t_j$}
        \EndIf
    \EndFor
\EndFor

\Statex
\Statex \textbf{Group $\mathsf g_6$: guards to clauses}
\For{$i=1$ to $n$}
    \For{$k=1$ to $m$}
        \State \Call{Assign}{$\guard_i,\clause_k$}
    \EndFor
\EndFor

\Statex
\Statex \textbf{Group $\mathsf g_7$: clauses to $\winner$, and $\loser$ to clauses}
\For{$k=1$ to $m$}
    \State \Call{Assign}{$\clause_k,\winner$};
\EndFor
\For{$k=1$ to $m$}
     \Call{Assign}{$\loser,\clause_k$}
\EndFor

\Statex
\Statex \textbf{Group $\mathsf g_8$: guards to designated candidates}
\For{$i=1$ to $n$}
    \State \Call{Assign}{$\guard_i,\winner$}; \Call{Assign}{$\guard_i,\loser$}
\EndFor

\Statex
\Statex \textbf{Group $\mathsf g_9$: total order on literals}

\For{$i=1$ to $n$}
    \State \Call{Assign}{$\f_i,\t_i$}
    \For{$j=i+1$ to $n$}
        \State \Call{Assign}{$\f_i,\t_j$}
    \EndFor
    \For{$j=i+1$ to $n$}
        \State \Call{Assign}{$\t_i,\f_j$}
    \EndFor
\EndFor
\For{$i=1$ to $n$}
    \For{$j=i+1$ to $n$}
        \State \Call{Assign}{$\f_i,\f_j$}
    \EndFor
\EndFor

\For{$i=1$ to $n$}
    \For{$j=i+1$ to $n$}
        \State \Call{Assign}{$\t_i,\t_j$}
    \EndFor
\EndFor

\end{algorithmic}
\end{breakablealgorithm}
\newpage
\begin{breakablealgorithm}{\thealgorithm\ \textsc{BuildTournament}$(\Phi)$ (continued)}

\begin{algorithmic}[1]

\Statex \textbf{Group $\mathsf g_{10}$: total order on guards}
\For{$i=1$ to $n$}
    \For{$j=i+1$ to $n$}
        \State \Call{Assign}{$\guard_i,\guard_j$}
    \EndFor
\EndFor

\Statex
\Statex \textbf{Group $\mathsf g_{11}$: total order on clauses}
\For{$k=1$ to $m$}
    \For{$k'=k+1$ to $m$}
        \State \Call{Assign}{$\clause_k,\clause_{k'}$}
    \EndFor
\EndFor

\Statex
\Statex \textbf{Group $\mathsf g_{12}$: designated candidates}
\State \Call{Assign}{$\winner,\loser$}

\State \Return $G(\Phi)=(V(\Phi),\weight(\Phi))$
\end{algorithmic}
\end{breakablealgorithm}

\newpage
\section[Deferred Proofs of Section 4]{Deferred Proofs of \Cref{sec:hardness}}\label{apx:missing-proofs}

\subsection{Subtournament-related lemmas}

The next lemma is the basic robustness statement used throughout the appendix. It allows us to safely ignore extra single literals and focus on the part of the graph where literals appear together with their guards.

\addingLits*

    \begin{proof}
        We proceed by induction on $t\coloneqq|B|$ for some $B$ as defined above.

        \paragraph{Base case} ($t=0$) This is trivially true.

        \paragraph{Inductive Hypothesis} Fix $t \ge 1$ and assume the claim holds for $t-1$. That is, given any $B$ of size $t$ (as defined above),
            \[ \SSV(G \cup ( B \setminus \{\lit_i\} )) = \SSV(G) \quad  \forall \, \lit_i \in B \]

        \paragraph{Inductive Step} Fix any arbitrary $B$ of size $t$ and define $H \coloneqq G \cup B$. Let $M \in \{\winner, \loser\}$ be the designated winner, $\SSV(G) = M$.
        
        By the inductive hypothesis, we know that the edges $M \rightarrow \lit_i$ (in group $\group1$) are critical edges of $H$ for all $\lit_i \in B$, since $\SSV(H\setminus\{\lit_i\}) = M$. Let $M \rightarrow \lit_{i^*}$ be the largest among those ($\is = \argmin_i \{\lit_i\in B\}$).  Note that the only other possible critical edge of larger weight is $M \rightarrow \lit_j$ for $j<\is,\;\lit_j \in V(G)$, i.e., the designated winner to some other literal already present in $G$ with a smaller index. However, the critical edge produced by removing $\lit_j$ is $\guard_j \rightarrow \lit_j$ by \Cref{lem:guards} since $\first(H\setminus\{\lit_j\})=j$, therefore the largest set of critical edges in $H$ is $M \rightarrow \lit_\is,\, \lit_\is \in B$ and $\SSV(H)=M$.
        
    \end{proof}

\Cref{lem:adding_lits_C_D_wins} allows us to focus on admissible graphs for which all literals are paired with their guards whenever we want to prove the winner is $\winner$ or $\loser$. We formalize this in the observation below.

\begin{observation}\label{obs:decomposition}
Let $G$ be an arbitrary admissible graph. Let $S=\{\t_i,\f_i,\guard_i : i\in\Vars\}\cap V(G)$ be the set of present literals and guards. By definition of admissible graphs, $\first(G)=0$ and for every $i$ at most one literal is present. Define $I \coloneqq \{\,i\in\Vars : \guard_i \in S\,\}$ as the set of present guards. Since $\first(G)=0$, every guard present in $S$ has its corresponding literal present as well. Hence, for each $i\in I$, there is a unique literal
$\lit_i \in \{\t_i,\f_i\}\cap S$. We write
\[
S_I \coloneqq \{\guard_i,\lit_i : i\in I\}
\qquad\text{and}\qquad
B \coloneqq S\setminus S_I.
\]
Thus, $S_I$ is the paired core, while $B$ consists of the remaining single literals whose corresponding guards are absent. By \Cref{lem:adding_lits_C_D_wins}, once \winner\ or \loser\ is shown to win on the paired core, adding back the literals in $B$ does not change the winner. Therefore, in the proofs below, we first work on the paired core and then restore the literals in $B$ at the end.
\end{observation}

\subsubsection{\texorpdfstring{$\loser$ wins in Admissible Graphs without $\winner$ or clauses}{D wins in Admissible Graphs without C or clauses}}

\begin{claim}[$\loser$ wins admissible graphs without $\winner$ or $\clause_k$'s] \label{lem:d_wins_without_C_or_L}
Let $G$ be an admissible graph such that $\winner,\clause_k\notin V(G)$ for all $k\in[m]$ and $\loser\in V(G).$ Then $\SSV(G)=\loser$.

    \begin{proof}

Using \Cref{obs:decomposition}, it suffices to prove the claim for the paired core. We write $S=S_I\cup B$, where $S_I$ is the paired core and $B$ is the set of extra single literals. Let $G_I$ be the induced subtournament on $V(G_I)=\{\loser\}\cup S_I$. We show that $\SSV(G_I)=\loser$ by induction on the number of pairs, then by \Cref{lem:adding_lits_C_D_wins}, it follows that $\SSV(G)=\loser.$ For each integer $t\ge 0$, define $P(t)$ to be the statement: 

\medskip

$P(t):$  $\forall J\subseteq I$ with $|J|=t$, if $G_J$ is the induced subtournament on $V(G_J)=\{\loser\}\cup S_J,$ then $\SSV(G_J)=\loser$.

\medskip

We prove $P(t)$ by induction on $t$. We note that at each step the graphs differ by two nodes, a literal $\lit_i$ and its guard $\guard_i$ for some $i\in \Vars$.

        \paragraph{Base Case: $(t=0)$} then $J=\emptyset$ and $G=\{\loser\}$, by definition $\SSV(G)=\loser.$ 
        \paragraph{Induction hypothesis.}
        Fix $t\ge 1$ and assume $P(t-1)$ holds. 

        \paragraph{Inductive step:} Fix $J$ with $|J|=t.$ Let $G:=G_J$. We will show that $\SSV(G)=\loser$ by showing that the edge $D\rightarrow X_r,$ where $r=\max\{J\}$, is the largest critical edge of $G$. We will break the proof into the following three steps, building up from smaller subtournaments where we know the winner:

        \begin{enumerate}[
                    align=left,
                    itemindent=0pt,
                    labelindent=0pt,
                    labelwidth=2pt,
                    labelsep=0.5em,
                    leftmargin=\dimexpr2.3em\relax
                ]
            \item[\textbf{Step 1:}] Show that $\SSV(G\setminus\{\guard_i,\lit_i\})=\loser$ for all $i\in J.$
            \item[\textbf{Step 2:}]
            Show $\SSV(G\setminus\{\guard_i\})=\loser$ with the largest critical edge being $\loser\rightarrow\lit_i$ for all $i\in J.$            
            \item[\textbf{Step 3:}] Show $\SSV(G)=\loser$ with the largest critical edge being $\loser\rightarrow\guard_r.$
        \end{enumerate}

    \noindent Towards \textbf{Step 1}, consider the subtournament $G\setminus\{\lit_i,\guard_i\}$ for some $i\in J.$ By the induction hypothesis, $\SSV(G\setminus\{\lit_i,\guard_i\})=\loser$. By the definition of critical edges (\Cref{def:critical-edges})
        \begin{itemize}
            \item the edge $\loser\rightarrow\lit_i$ is a critical edge of $G\setminus\{\guard_i\}$, and
            \item the edge $\loser\rightarrow\guard_i$ is a critical edge of $G\setminus\{\lit_i\}$ .
        \end{itemize}

        \noindent Towards \textbf{Step 2}, let us focus on the subtournament $G\setminus\{\guard_i\}.$ We have shown that $\loser\to\lit_i$ is a critical edge of $G\setminus\{\guard_i\},$ next we will show that is also the largest. Notice that the only edges that could have a larger weight than $\loser\rightarrow\lit_i$ (which are in $\group1$) are $\loser\rightarrow\lit_j$ for $j<i.$ However, the critical edge produced by removing $\lit_j$ is $\guard_j\to\lit_j$ ($-\group2$) since $\first(G\setminus\{\guard_i,\lit_j\})=j,$ and by \Cref{lem:guards} $\SSV(G\setminus\{\guard_i,\lit_j\})=\guard_j.$ Hence, we can conclude that $\SSV(G\setminus\{\guard_i\})=\loser.$

        Finally, toward \textbf{Step 3,}We just shown that  $\SSV(G\setminus\{\guard_i\})=\loser$. Hence, $D\rightarrow\guard_i$ ($-\group8$) is a critical edge in $G$ for all $i.$ The largest among those is $\loser\rightarrow\guard_r,$ where $r=\max\{J\}$.  To conclude the proof that $\SSV(G)=D$, we must show that $\loser\rightarrow\guard_r$ is larger (less negative) than the remaining critical edges, defined by removing any other node ($\lit_i$ or $\loser$). We will consider both cases: 
        
        \begin{enumerate}
            \item \textbf{(Removing $\lit_i$).} By \Cref{lem:guards} that $\SSV(G\setminus\{\lit_i\})=\guard_i$, since $\first(G\setminus\{\lit_i\})=i$. Thus, the critical edges in $G$ produced by removing $\lit_i$ are $\guard_i\rightarrow\lit_i$ in group $-\group2 \; (< -\group8)$.
            \item \textbf{(Removing $\loser$).}  First we note that $\first(G\setminus\{\loser\})=0$, hence $\SSV(G\setminus\{\loser\})\neq X_j, \forall j$ by \Cref{lem:guards}. Therefore the critical edge must be $\lit_i\rightarrow\loser$ for some $i$, but these edges are in group $-\group1 \; (<-\group8)$.
        \end{enumerate}
        
        These are all possible deletions. Therefore $\loser\to \guard_r$ is the maximum critical edge of $G_J$, and
            \[ \SSV(G_J)=\loser. \]
        This proves $P(t)$. Taking $J=I$ gives $\operatorname{SSV}(G_I)=\loser$, and adding back the single literals in $B$ using \Cref{lem:adding_lits_C_D_wins} yields $\SSV(G)=\loser$.

    \end{proof}

\end{claim}

\begin{claim}[Without $\loser$ or $\clause$, $\winner$ wins]\label{lem:c_wins_without_D_or_L}
Let $G$ be an admissible graph such that $\loser,\clause_k\notin V(G)$ for all $k\in[m],$ and $\winner\in V(G).$ Then $\SSV(G)=\winner.$

Let $S \subseteq \{\t_i,\f_i,\guard_i : i\in\Vars\}$ be a subset of literals and guards such that $\first(G_S)=0$ and, for every $i\in\Vars$, at most one of $\{\t_i,\f_i\} \in S$.
Let $G$ be the induced subtournament on $V(G)=\{\winner\}\cup S.$ Then $\SSV(G)=\winner$.

    \begin{proof}

       Let $G^{\loser}{=(G\setminus\{\winner\}})\cup\{\loser\}.$ Then $V(G^{\loser})=\{\loser\}\cup S$, and $G^{\loser}$ satisfies the conditions of \Cref{lem:d_wins_without_C_or_L}. Thus $\SSV(G^{\loser})=\loser.$

       Next, we will show that $G$ and $G^\loser$ are order-isomorphic. First, the set $S$ is the same in both graphs; all edges with both endpoints in $S$ are unchanged. Thus, we only need to consider edges between $S$ and $\winner$ in $G$, respectively $\loser$ in $G^{\loser}.$ By construction, these are exactly the edges in $\group1$ and $\group8$, and replacing $\winner$ by $\loser$ preserves both their orientations and their relative order in the weight ranking. Hence the subtournaments induced by $\{\winner\}\cup S$ and by $\{\loser\}\cup S$ are identical after relabeling $\winner$ as $\loser$. Therefore, $\SSV(G)=\winner.$

    \end{proof}

\end{claim}

\subsubsection{\texorpdfstring{$\loser$}{D} wins in Admissible Graphs without \texorpdfstring{$\winner$}{C}}

\DwithoutC*
    \begin{proof}
        Using \Cref{obs:decomposition}, we are going to decompose into the different types of nodes. Let $S_I$ be the paired core of literals and guards and $B$ the set of extra single literals. And let $K\subseteq[m]$ be the present clauses. By \Cref{lem:adding_lits_C_D_wins}, it suffices to prove the claim for the paired core. Let $G_{I,K}$ be the induced subtournament on $V(G_{I,K})=\{\loser\}\cup\{L_k:k\in K\}\cup S_I$. We will show that $\SSV(G_{I,K})=\loser$ and thus $\SSV(G)=\loser$.

        We prove the statement by induction on the number of clauses. For each integer $t\ge 0$, define the statement $P(t)$ as follows:
        
        \medskip
        \noindent
        $P(t)$: $\forall \, J\subseteq I, \forall \, K'\subseteq K$ with $|K'|=t$, the $\SSV(H_{J,K'})=\loser$ where $H_{J,K'}$ is the induced subtournament on 
        \[V(H_{J,K'})=\{\loser\}\cup \{\lit_j,\guard_j:j\in J\} \cup \{\clause_k:k\in K'\}.\]

        \paragraph{Base case: $(t=0)$.}
        Let $J\subseteq I$ and let $K'=\emptyset$. Then $V(H_{J,K'})=\{\loser\}\cup \{\lit_j,\guard_j:j\in J\}$, so $\SSV(H_{J,K'})=\loser$ by \Cref{lem:d_wins_without_C_or_L}. Hence $P(0)$ holds.
        
        \paragraph{Induction hypothesis.}
        Fix $t\ge 1$ and assume that $P(r)$ holds for every $r<t$.
        
        \paragraph{Induction step.}
        We prove $P(t)$. Fix $J\subseteq I$, $K'\subseteq K$ with $|K'|=t$, and $H\coloneqq H_{J,K'}$. We assume $|J|>0$, since $|J|=0$ implies that $\loser$ is a Condorcet winner of the corresponding graph. By definition, since every present guard in $H$ has its corresponding literal also present, we have $\first(H)=0$. We will show that the largest critical edge of $H$ is of the form $\loser\to \clause_k$ for some $k\in K'$, which implies $\SSV(H)=\loser$.
        
        We consider the critical edge produced by removing each type of node.
        
        \begin{enumerate}
            \item \textbf{(Removing a clause $\clause_k$ .)} The graph $H\setminus \{\clause_k\}$ is exactly of the form covered by $P(t-1)$, with the same set $J$ and clause set $K'\setminus\{k\}$. Since $|K'\setminus\{k\}|=t-1$, the induction hypothesis gives $\SSV(H\setminus\{ \clause_k\})=\loser$. Therefore, $\loser\to \clause_k$ is a critical edge of $H$, and it lies in group $\group7$.
        
            \item \textbf{(Removing a literal $\lit_i$ .)} By assumption $|J|>0.$ Since $\guard_i$ is still present and $\lit_i$ was the unique present literal of variable $i$, we have $\first(H\setminus \{\lit_i\})=i$. By \Cref{lem:guards}, $\SSV(H\setminus \{\lit_i\})=\guard_i$. Hence the critical edge produced by removing $\lit_i$ is $\guard_i\to \lit_i$, which lies in group $-\group2$. Hence, this critical edge is smaller than $\loser\to \clause_k$ in $\group_7.$
        
            \item \textbf{(Removing $\loser$ .)} Since removing $\loser$ does not affect literals or guards, we still have $\first(H\setminus \{\loser\})=0$. Hence, by \Cref{lem:guards}, no guard can win $H\setminus \{\loser\}$. We distinguish two cases.
        
            \begin{enumerate}
                \item If $H\setminus \{\loser\}$ is UNSAT, then by \Cref{lem:clauses}, $\SSV(H\setminus \{\loser\})=\clause_{k^*}$, where $\clause_{k^*}=\firstU(H\setminus \{\loser\})$ is the first unsatisfied clause. Thus the critical edge produced by removing $\loser$ is $\clause_{k^*}\to \loser$, which lies in group $-\group7$.
        
                \item If $H\setminus \{\loser\}$ is SAT, then $\firstU(H\setminus \{\loser\})=0$, so by \Cref{lem:clauses} no clause can win $H\setminus \{\loser\}$. As no guard can win $H\setminus \{\loser\}$, the winner must be some literal $\lit_j$, and the critical edge produced by removing $\loser$ is $\lit_j\to \loser$, which lies in group $-\group1$.
            \end{enumerate}
        
            In either case, the critical edge produced by removing $\loser$ is negative, hence strictly smaller than every positive group $\group7$ edge.

           \item \textbf{(Removing a guard $\guard_i$.)}
            Let $H_i:=H\setminus\{\guard_i\}$. We show that the critical edge
            produced by removing $\guard_i$ has weight $<\group7$. First, $\first(H_i)=0$, so by \Cref{lem:guards}, no guard can win $H_i$.
            
            We claim that, for every $k\in K'$, $\SSV(H_i\setminus\{\clause_k\})=\loser.$ Indeed, $H_i\setminus\{\clause_k,\lit_i\} = H_{J\setminus\{i\},K'\setminus\{k\}},$ so by the induction hypothesis it is won by $\loser$. Since $\guard_i\notin V(H_i\setminus\{\clause_k\})$, adding back the single
            literal $\lit_i$ does not change a designated winner, by \Cref{lem:adding_lits_C_D_wins}. Hence
            \[
                \SSV(H_i\setminus\{\clause_k\})=\loser .
            \]
            Thus every edge $\loser\to\clause_k$, $k\in K'$, is a critical edge
            of $H_i$, in group $\group7$.
            
            It remains to rule out $\SSV(H_i)=\lit_i$. If $\lit_i$ won $H_i$, then its winning critical edge would have to beat the group $\group7$ edge $\loser\to\clause_k$. Since $\guard_i\notin V(H_i)$, the group $\group2$ edge $\lit_i\to\guard_i$ is absent. The only remaining edges out of $\lit_i$ that can beat group $\group7$ are edges $\lit_i\to\clause_k$ in group $\group3$. But for every $\clause_k\in K'$, we just proved that $\SSV(H_i\setminus\{\clause_k\})=\loser,$ so the critical edge produced by deleting $\clause_k$ is $\loser\to\clause_k$, not $\lit_i\to\clause_k$. Therefore $\lit_i$ cannot win $H_i$.
            
            Thus neither a guard nor $\lit_i$ can win $H_i$. Hence the critical edge into $\guard_i$ is not the positive edge $\lit_i\to\guard_i$; all remaining possible edges into $\guard_i$ have weight $<\group7$.
        \end{enumerate}
        We have now examined all possible removals from $H$. Removing any clause $\clause_k$ produces a critical edge $\loser\to \clause_k$ in group $\group7$, while every critical edge produced by removing a literal, a guard, or $\loser$ has weight strictly smaller than $\group7$. Therefore, the largest critical edge of $H$ is of the form $\loser\to \clause_k$, and consequently $\SSV(H)=\loser$. This concludes the proof of $P(t)$.

        \bigskip
        Applying $P(|K'|)$ to $J=I$ and $K'=K$, we get $\SSV(G_{J,K})=\loser$. We can apply \Cref{lem:adding_lits_C_D_wins} to add the literals in $B$, to get $\SSV(G)=\loser$.
    \end{proof}

\subsubsection{\texorpdfstring{$\winner$}{C} wins Admissible Graphs without clauses }

\CwinsNoClause*
\begin{proof}

Using \Cref{obs:decomposition}, we write $S=S_I\cup B$, where $S_I$ is the paired core and $B$ is the set of extra single literals.  Let $G_I$ be the induced subtournament on $V(G_I)=\{\winner,\loser\}\cup S_I$. We will show that $\SSV(G_I)=\winner$. By \Cref{lem:adding_lits_C_D_wins}, it follows that $\SSV(G)=\winner$.

We do this by proving that the maximum critical edge of $G_I$ is $\winner\to\loser$, and comparing it to the critical edges produced by removing every other node of $G_I$.

        \begin{enumerate}
            \item \textbf{(Removing $\loser$.)} $\SSV(G_I\setminus \{\loser\}) = \winner$ by \Cref{lem:c_wins_without_D_or_L}, producing a critical edge in group $\group12$.
            \item \textbf{(Removing $\winner$.)} $\SSV(G_I\setminus \winner) = \loser$ by \Cref{lem:d_wins_without_C_or_L}, producing a critical edge in group $-\group12$.
            \item \textbf{(Removing $\lit_i$.)} Since in $G_I$ all the literals are paired with their guards, we get  $\first(G_I\setminus\{\lit_i\}) = i$, implying $ \guard_i \rightarrow \lit_i$ is the critical edge produced by removing $\lit_i$, in group $-\group2$.
            \item \textbf{(Removing $\guard_i$.)} It suffices to show that the critical edge produced by removing $\guard_i$ has weight $< \group12$. That is, the edge $w \rightarrow \guard_i$ has weight less than $\group12,$ where $w=\SSV(G_I\setminus \{\guard_i\})$ Suppose, for contradiction, that the critical edge has weight $>\group12$. The only two candidate sets for the critical edges are thus $\guard_{j<i} \rightarrow \guard_i$ ($\group10$) and $\lit_i \rightarrow \guard_i$ ($\group2$). We show that $w\neq\guard_j$ and $w\neq\lit_i.$ 
            
            \begin{enumerate}[
                align=left,
                itemindent=0pt,
                labelindent=0pt,
                labelwidth=\casewdii,
                labelsep=0.4em,
                leftmargin=\dimexpr\casewdii+0.5em\relax
                ]
                \item[($w=\guard_{j<i}$)] By \Cref{lem:guards}, $\guard_j$ can only win if its corresponding literal is removed, and by assumption, $\first(G_I)=0$, hence $\first(G_I\setminus\{\guard_i\})=0.$
                \item[($w=\lit_i$)]  Since the winner is $w=\lit_i,$ the largest critical edge of $G_I\setminus \{\guard_i\}$ must be of the form $\lit_i \rightarrow u$, for some $u\in G_I\setminus\{\guard_i\}$. We proceed by case analysis:
                \begin{enumerate}[
                    align=left,
                    itemindent=0pt,
                    labelindent=0pt,
                    labelwidth=\casewdiii,
                    labelsep=0.5em,
                    leftmargin=\dimexpr\casewdiii+0.5em\relax
                ]
                    \item[($u=\lit_{j\neq i}$)] By \Cref{lem:guards}, the winner would be  in $\guard_j$, since $\first(G\setminus\{\guard_i,\lit_j\})=j.$
                    \item[($u=\guard_{j\neq i}$)]  Then the critical edge of $G\setminus\{\guard_i\}$ would be in group $-\group5.$  
                    \item[($u=\winner$)] By \Cref{lem:d_wins_without_C_or_L}, $\loser$ wins. Hence, a critical edge of $G\setminus\{\guard_i\}$ is $\loser\to\winner$ that is in group $-\group12 > -\group5.$  
                    \item[($u=\loser$)] By \Cref{lem:c_wins_without_D_or_L}, $\winner$ wins. Hence, a critical edge of $G\setminus\{\guard_i\}$ is $\winner\to\loser$ which is in group $\group12>-\group5.$
                \end{enumerate}
                Thus, $\lit_i$ cannot win in $G_I\setminus \{\guard_i\}$ since there exists a critical edge ($\winner\to\loser$) in $G_I\setminus\guard_i$ with a larger weight than the largest critical edge of the form $\lit_i\to u.$
            \end{enumerate}
            Thus, the critical edge produced by removing $\guard_i$ has weight $<0<\group12$.
        \end{enumerate}
        Thus, the highest critical edge in $G_I$ is $\winner \rightarrow \loser$, and $\winner$ wins.

    \end{proof}

\subsubsection{ \texorpdfstring{$\loser$}{D} wins Admissible Graphs with \texorpdfstring{$\winner$}{C} and a single UNSAT clause}

\DBaseCase*
    \begin{proof}
        We will show that the largest critical edge is $\loser\rightarrow\winner$, by comparing it to all other critical edges produced by removing any node in $V(G)\setminus\{\loser\}.$ We begin by removing all unpaired literals $\lit_i$: literals for whom no guard $\guard_i$ is present. A simple argument uses \Cref{lem:adding_lits_C_D_wins} to add these back at the end.
            \begin{enumerate}
            \item \textbf{(Removing $\winner$.)} $\SSV(G\setminus \winner) = \loser$ by \Cref{lem:D-wins-with-clauses-no-C-general}, producing a critical edge in group $-\group12$.
            \item \textbf{(Removing $\loser$.)} $\SSV(G\setminus \{\loser\}) = \clause_k$ by \Cref{lem:clauses} since $\clause_k$ is unsatisfied and $\loser$ is not in the graph. Hence, the critical edge produced is in group $-\group7$.
            \item \textbf{(Removing $\clause_k$.)} 
            $\SSV(G\setminus\{\clause_k\})=\winner$ by \Cref{lem:c-wins-without-clauses}, hence the critical edge produced is $\winner\to\clause_k$ with weight in $-\group7.$
            \item \textbf{(Removing $\lit_i$.)} $\first(G\setminus\{\lit_i\}) = i \implies \SSV(G\setminus\{\lit_i\}) = \guard_i$, by \Cref{lem:guards}. Hence, $\guard_i \rightarrow \lit_i$ is the critical edge produced, which is in group $-\group2$.
            \item \textbf{(Removing $\guard_i$.)} It suffices to show that the critical edge produced by removing $\guard_i$ has weight $<- \group12$. That is, the edge $\SSV(G\setminus \{\guard_i\}) \rightarrow \guard_i$ has weight less than $-\group12.$ Suppose, for contradiction, that the critical edge has weight $>-\group12$. The only two candidate sets for the critical edges are thus $\guard_{j<i} \rightarrow \guard_i$ ($\group10$) and $\lit_i \rightarrow \guard_i$ ($\group2$).
            
            \begin{enumerate}[
                align=left,
                itemindent=0pt,
                labelindent=0pt,
                labelwidth=\casewdiv,
                labelsep=0.4em,
                leftmargin=\dimexpr\casewdiv+0.5em\relax
                ]
                \item[($\guard_{j<i}$)] By \Cref{lem:guards}, $\guard_j$ can only win if its corresponding literal is removed, and by assumption, $\first(G)=0$.
                \item[($\lit_i$)] Assume, for the sake of contradiction, that $SSV(G\setminus \{\guard_i\}) = \lit_i$. Then the largest critical edge of $G\setminus \{\guard_i\}$ must be of the form $\lit_i \rightarrow u$, for some $u\in V(G\setminus\guard_i)$. We proceed by case analysis:
                \begin{enumerate}[
                    align=left,
                    itemindent=0pt,
                    labelindent=0pt,
                    labelwidth=\casewdiii,
                    labelsep=0.5em,
                    leftmargin=\dimexpr\casewdiii+0.5em\relax
                ]
                    \item[($u=\lit_{j\neq i}$)] By \Cref{lem:guards}, the winner would be  in $\guard_j$.
                    \item[($u=\guard_{j\neq i}$)]  Then the critical edge of $G\setminus\guard_i$ would be in group $-\group5.$  
                    \item[($u=\winner$)] By \Cref{lem:D-wins-with-clauses-no-C-general}, $\loser$ wins. Hence, a critical edge of $G\setminus\guard_i$ is $\loser\to\winner$ that is in group $-\group12 > -\group5.$  
                    \item[($u=\loser$)] By \Cref{lem:clauses}, $\clause_k$ wins. Hence, a critical edge of $G\setminus\guard_i$ is $\clause_k\to\loser$, which is in group $-\group7<-\group12,$ which is smaller than the critical edge $\loser\to\winner$ in group $-\group12$.
                    \item[($u=\clause_k$)] by \Cref{lem:c-wins-without-clauses} $\winner$ wins. Hence, a critical edge of $G\setminus\guard_i$ is $\winner\to\clause_k$, which is in group $-\group7<-\group12,$ which is smaller than the critical edge $\loser\to\winner$ in group $-\group12$.
                \end{enumerate}
                Thus, $\lit_i$ cannot win in $G\setminus \{\guard_i\}$ since there exists a critical edge ($\loser\to\winner$) in $G\setminus\guard_i$ with a larger weight than any critical edge of the form $\lit_i\to u.$
            \end{enumerate}
            Thus, the critical edge produced by removing $\guard_i$ has weight $<-\group12$.
        \end{enumerate}
        As mentioned at the start, we use the simple argument implied by \Cref{lem:adding_lits_C_D_wins} to add back the unpaired literals and prove the claim.
        
        Thus, the highest critical edge in $G$ is $\loser \rightarrow \winner$, and $\loser$ wins.

    \end{proof}

\subsubsection{C wins SAT Admissible Graphs without D}

\CwinsNoD*

\begin{proof}
    We proceed by induction on $t := |V(G)|$.

    \paragraph{Base case} ($t = 1$).
    $V(G) = \{\winner\}$, so $\SSV(G) = C$.
    
    \paragraph{Inductive hypothesis.}
    Fix $t \ge 2$ and assume the claim holds for every admissible graph $G'$ with $\winner \in V(G')$, $\loser \notin V(G')$, $\firstU(G') = 0$, and $|V(G')| < t$. 
    
    \paragraph{Inductive step.}
    Let $G$ be such a graph with $|V(G)| = t$. Since $G$ is an admissible graph, $\first(G) = 0$. We consider all possible critical edges produced by removing each type of node.

    \begin{enumerate}
        
        \item \textbf{(Removing $\guard_i$.)}
        Since $\first(G) = 0$, the literal $\ell_i$ is present in $G$. Removing $\guard_i$ does not expose any other guard, hence $\first(G \setminus \{\guard_i\}) = 0$. Removing a guard does not affect clause satisfaction, so $\firstU(G \setminus \{\guard_i\}) = 0$. Thus $G \setminus \{\guard_i\}$ is a strictly smaller admissible graph with $\winner$ present, $\loser$ absent, and $\firstU = 0$. By the inductive hypothesis, $\SSV(G \setminus \{\guard_i\}) = \winner$. Therefore, $\winner \to \guard_i$ is a critical edge of $G$ with weight in group $-\group8$.
        
        \item \textbf{(Removing $\clause_k$.)}
        Removing a clause does not affect guards or literals, so $\first(G \setminus \{\clause_k\}) = 0$ and
        $\firstU(G \setminus \{\clause_k\}) = 0$. By the inductive hypothesis, $\SSV(G \setminus \{\clause_k\}) = \winner$. Therefore, $\winner \to \clause_k$ is a critical edge of $G$ with weight in group $-\group7$.
        
        \item \textbf{(Removing $\ell_i$ with $\guard_i \in V(G)$.)}
        We have $\first(G \setminus \{\ell_i\}) = i > 0$. By \Cref{lem:guards}, $\SSV(G \setminus \{\ell_i\}) = \guard_i$. The critical edge is $\guard_i \to \ell_i$ with weight in group $-\group2$. However, $\winner\to\guard_i$ is also a critical edge with a larger weight. 
        
        \item \textbf{(Removing a literal $\ell_i$ with $\guard_i \notin V(G)$.)}
        Since $\guard_i$ is absent, $\first(G \setminus \{\ell_i\}) = 0$, so $G \setminus \{\ell_i\}$ is an admissible graph with $\winner$ present and $\loser$ absent.
        
        \begin{itemize}
          \item \textbf{Case 1: $\firstU(G \setminus \{\ell_i\}) = 0$.}
          By the inductive hypothesis, $\SSV(G \setminus \{\ell_i\}) = \winner$. The critical edge is $\winner \to \ell_i$ with weight in group $\group1$.
        
          \item \textbf{Case 2: $\firstU(G \setminus \{\ell_i\}) = r > 0$.}
          The literal $\ell_i$ was the unique satisfying literal for clause $\clause_r$, and $\loser \notin V(G \setminus \{\ell_i\})$. By  \Cref{lem:clauses}, $\SSV(G \setminus \{\ell_i\}) = \clause_r$. Since $\ell_i$ satisfied $\clause_r$, the edge $\ell_i \to \clause_r$ lies in group $\group3$, so the critical edge $\clause_r \to \ell_i$ has weight in group $-\group3$. But the edge $\winner\to\clause_r$ is also critical with a larger weight. 
        \end{itemize}
        
        \item \textbf{(Removing $\winner$.)}
        Removing $\winner$ does not affect guards or literals, so $\first(G \setminus \{\winner\}) = 0$. By \Cref{lem:guards}, no guard wins $G \setminus \{\winner\}$. Since $\firstU(G \setminus \{\winner\}) = 0$ and $\loser \notin V(G \setminus \{\winner\})$,  \Cref{lem:clauses} implies no clause wins $G \setminus \{\winner\}$ either. Hence, $\SSV(G \setminus \{\winner\}) = \ell_j$ for some literal, and the critical edge $\ell_j \to \winner$ has weight in group $-\group1$. But, as we showed in the previous cases, there is always a critical edge from $\winner$ with a larger weight.   
    \end{enumerate}

By the case analysis above, every critical edge that is not from $\winner$ is dominated by a critical edge from $\winner$. Therefore, the maximum critical edge is from $\winner$, and hence $\SSV(G)=\winner$.

\end{proof}

\newpage
\section[Deferred Proofs of Section 5]{Deferred Proofs of \Cref{sec:stable_voting}}\label{apx:sv-extension}

\Reachability*

\begin{proof}[Proof of \Cref{claim:reachability}]
    The evaluation starts at $G(\Phi) = G_{\alpha_{<1}}$, which is a prefix graph. We show that whenever the rule is evaluated on a graph $G$ in $(\mathcal{P})\cup(\mathcal{A})\cup(\mathcal{F})$, every subtournament $G\setminus\{B\}$ on which it recurses also lies in $(\mathcal{P})\cup(\mathcal{A})\cup(\mathcal{F})$; the claim then follows by induction on the depth of the recursion.
    
    \begin{itemize}
        \item \textbf{$G = G_{\alpha_{<j}}$ is a prefix graph.} If $j = n+1$ then $G$ is a full-assignment graph, hence admissible, and the admissible case applies. So let $j \le n$. The proof of \Cref{thm:main_result} shows that the scan on $G$ is resolved within the group $\group1$ edges of variables $1,\ldots,j$: for $i<j$ the pairs $(\winner,\lit_i)$ and $(\loser,\lit_i)$ are not critical, since $\first(G\setminus\{\lit_i\}) = i$ and so $\SSV(G\setminus\{\lit_i\}) = \guard_i$ by \Cref{lem:guards}; and among the four pairs of variable $j$, one is critical, because $G\setminus\{\f_j\}$ and $G\setminus\{\t_j\}$ are the prefix graphs of the two extensions of $\alpha_{<j}$ and each is won by $\winner$ or $\loser$ by $P(j{+}1)$. Hence the evaluated subtournaments are $G\setminus\{\lit_i\} \in (\mathcal{F})$ for $i<j$, and $G\setminus\{\f_j\}, G\setminus\{\t_j\} \in (\mathcal{P})$.
        
        \item \textbf{$G$ is admissible.} Let $B \in V(G)$ be arbitrary; we show $G\setminus\{B\}\in(\mathcal{F})\cup(\mathcal{A})$ regardless of whether the pair is scanned. Deleting a candidate cannot make both literals of a variable present, so $G\setminus\{B\}$ contains at most one literal per variable. Since $\first(G)=0$, every present guard has its literal present. If $B = \lit_i$ is a literal with $\guard_i \in V(G)$, then deleting $B$ uncovers exactly the guard $\guard_i$, so $\first(G\setminus\{B\}) = i \neq 0$ and $G\setminus\{B\}\in(\mathcal{F})$. For every other $B$ (a designated candidate, a clause, a guard, or a literal whose guard is absent), every remaining guard keeps its literal, so $\first(G\setminus\{B\}) = 0$ and $G\setminus\{B\}$ is admissible.
    
        \item \textbf{Graphs with $\first(G) = \is \neq 0$.} By \Cref{lem:guards}, $\SSV(G) = \guard_{\is}$, so the maximum critical edge of $G$ has the form $\guard_{\is} \to u$. For every $B \neq \guard_{\is}$, the guard $\guard_{\is}$ remains present in $G\setminus\{B\}$ with both of its literals absent, so $\first(G\setminus\{B\}) \in \{1,\ldots,\is\}$, and $G\setminus\{B\}\in(\mathcal{F})$. It remains to consider $B = \guard_{\is}$, and we distinguish two cases.
        \begin{itemize}
            \item If $G$ contains at most one literal per variable: deleting $\guard_{\is}$ cannot uncover a guard with a smaller index (those were covered in $G$ and keep their literals), so $\first(G\setminus\{\guard_{\is}\})$ is either $0$, in which case $G\setminus\{\guard_{\is}\}$ is admissible, or some $j > \is$, in which case $G\setminus\{\guard_{\is}\}\in(\mathcal{F})$. 
            \item If $G$ contains both literals $\t_m, \f_m$ of some variable $m$ (note $m \neq \is$, as $\t_{\is},\f_{\is}\notin V(G)$), we show that the pair $(Z, \guard_{\is})$ is never scanned, for any $Z$, so $G\setminus\{\guard_{\is}\}$ never arises. Deleting $\t_m$ leaves $\f_m$ present, so no guard becomes uncovered and $\first(G\setminus\{\t_m\}) = \is$; by \Cref{lem:guards}, $\SSV(G\setminus\{\t_m\}) = \guard_{\is}$, and the edge $\guard_{\is}\to\t_m$, which lies in group $\group5$, is critical. Hence the maximum critical edge of $G$ has weight at least that of a group $\group5$ edge. On the other hand, since $\t_{\is},\f_{\is}\notin V(G)$, the only positive edges entering $\guard_{\is}$ come from guards $\guard_j$ with $j<\is$, in group $\group{10}$. Since $\group5 > \group{10}$, every pair $(Z, \guard_{\is})$ has weight strictly below the maximum critical edge, so the scan stops before reaching it.
        \end{itemize}
    \end{itemize}
\end{proof}

\winnersAdmissible*
\begin{proof}[Proof of \Cref{claim:admissible-winners}]
    As in \Cref{obs:decomposition}, write $B := \{\lit_j \in V(G) : \guard_j \notin V(G)\}$ for the single literals whose guards are absent, and let the \emph{core} of $G$ be the induced subtournament on $V(G)\setminus B$. The core is admissible, contains no literal of any variable $j$ with $\lit_j \in B$, and consists of the designated candidates $\{\winner, \loser\}$ present in $G$, the paired guards and literals $\{\guard_i, \lit_i : i \in I\}$, and the clause candidates of $G$. By \Cref{lem:adding_lits_C_D_wins}, if the winner of the core is $\winner$ or $\loser$, then it is also the winner of $G$.
    
    \begin{enumerate}[label=(\alph*)]
        \item Suppose $\winner, \loser \in V(G)$. The core then has exactly the form $G(I,K)$ appearing in the induction statements established in the proofs of \Cref{lem:sat-C-wins,lem:unsat-D-wins}. If every clause of the core is satisfied by a present (paired) literal, the statement in the proof of \Cref{lem:sat-C-wins} gives $\SSV(\mathrm{core}) = \winner$ (when no clause is present, this is \Cref{lem:c-wins-without-clauses}); if some clause of the core is unsatisfied, the statement in the proof of \Cref{lem:unsat-D-wins} gives $\SSV(\mathrm{core}) = \loser$. In both cases \Cref{lem:adding_lits_C_D_wins} lifts the winner to $G$, so $A \in \{\winner,\loser\}$.

        \item Suppose $A = \winner$. If $\loser \in V(G)$: by the dichotomy in (a), $A = \winner$ forces every clause of the core to be satisfied by a paired literal; since paired literals are present in $G$ and $G$ has the same clauses as its core, $\firstU(G) = 0$.
        
        If $\loser \notin V(G)$: by \Cref{lem:clauses}, the winner is a clause candidate if and only if $\firstU(G) \neq 0$; since $A = \winner$, $\firstU(G) = 0$.

        \item  Suppose $A = \clause_{\ks}$. First, we show $\loser \notin V(G)$. If $\loser \in V(G)$ and $\winner \notin V(G)$, then \Cref{lem:D-wins-with-clauses-no-C-general} gives $A = \loser$, a contradiction; and if $\winner, \loser \in V(G)$, then (a) gives $A \in \{\winner,\loser\}$, again a contradiction. Now that $\loser \notin V(G)$, \Cref{lem:clauses} applies and gives $\firstU(G) = \ks$.

        \item  Suppose $A$ is a literal. As in (c), $\loser \in V(G)$ leads to $A \in \{\winner,\loser\}$ or $A = \loser$, a contradiction; so $\loser \notin V(G)$. If $\winner \in V(G)$: were $\firstU(G) = 0$, \Cref{lem:C-wins-SAT-no-D} would give $A = \winner$; were $\firstU(G) \neq 0$, \Cref{lem:clauses} would give a clause winner; either way a contradiction, so $\winner \notin V(G)$. Finally, by \Cref{lem:clauses} (with $\loser \notin V(G)$), the winner is a clause candidate if and only if $\firstU(G) \neq 0$; since $A$ is a literal, $\firstU(G) = 0$.
    \end{enumerate}
\end{proof}

\newpage
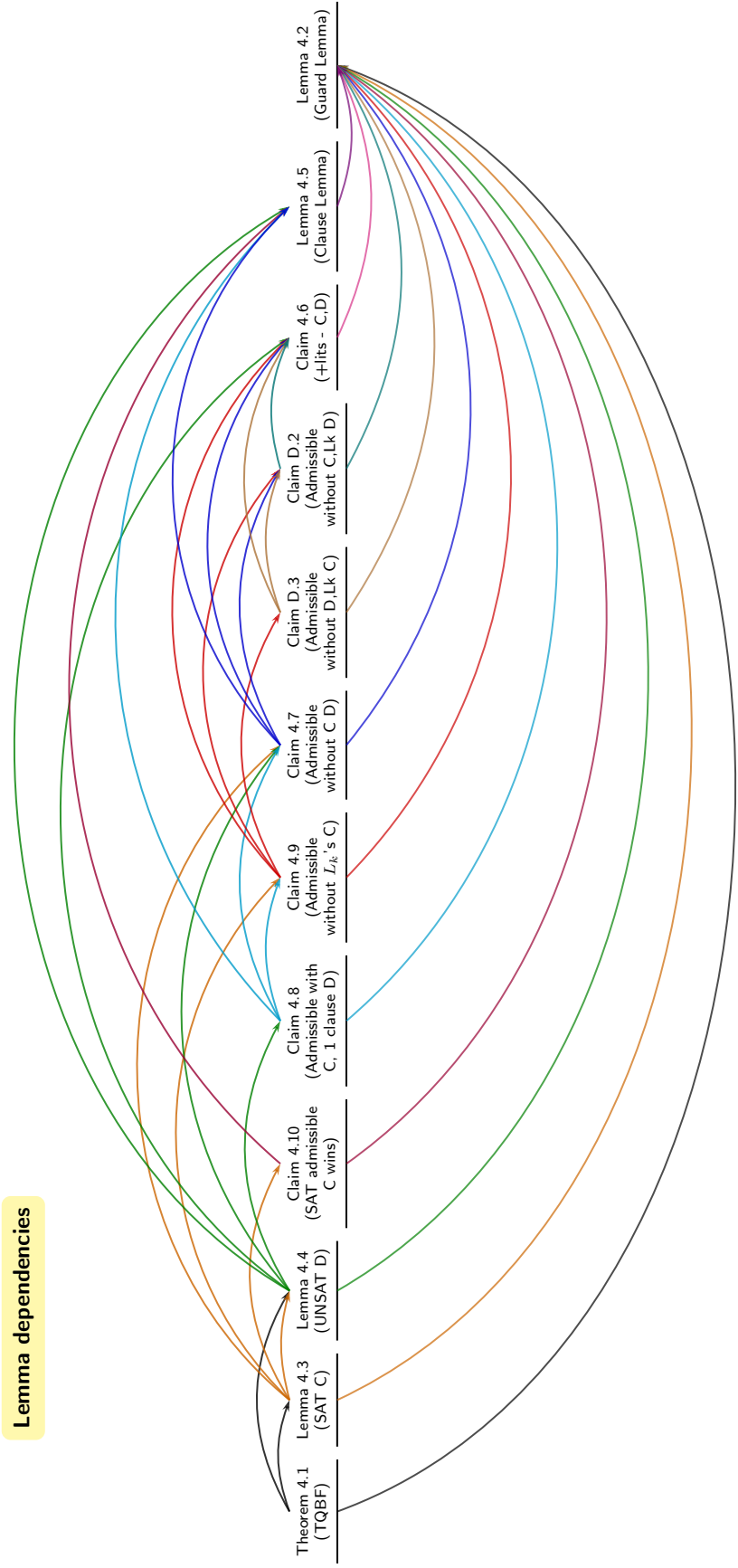
\begin{sidewaysfigure}
    \centering
    \resizebox{\linewidth}{!}{%
    \begin{tikzpicture}[
        node distance=0.5cm and 0.2cm,
        every node/.style={align=center, font=\sffamily\scriptsize},
        c_node/.style={fill=orange!30, rounded corners, inner sep=3pt},
        d_node/.style={fill=green!20, rounded corners, inner sep=3pt},
        base_node/.style={inner sep=3pt}
    ]

    \node[fill=yellow!40, rounded corners, font=\sffamily\normalsize\bfseries, inner sep=5pt] (title) at (3, 4.5) {Lemma dependencies};

    \node[base_node] (thm41) {\Cref{thm:main_result} \\ (TQBF)};
    \node[base_node, right=of thm41] (lem43) {\Cref{lem:sat-C-wins} \\ (SAT C)};
    \node[base_node, right=of lem43] (lem44) {\Cref{lem:unsat-D-wins} \\ (UNSAT D)};
    \node[base_node, right=of lem44] (clm410) {\Cref{lem:C-wins-SAT-no-D} \\ (SAT admissible \\ C wins)};
    \node[base_node, right=of clm410] (clm48) {\Cref{lem:D-wins-with-C-and-one-unsat-clause} \\ (Admissible with \\ C, 1 clause D)};
    \node[base_node, right=of clm48] (clm49) {\Cref{lem:c-wins-without-clauses} \\ (Admissible \\ without $L_k$'s C)};
    \node[base_node, right=of clm49] (clm47) {\Cref{lem:D-wins-with-clauses-no-C-general} \\ (Admissible \\ without C D)};
    \node[base_node, right=of clm47] (clmC3) {\Cref{lem:c_wins_without_D_or_L} \\ (Admissible \\ without D,Lk C)};
    \node[base_node, right=of clmC3] (clmC2) {\Cref{lem:d_wins_without_C_or_L} \\ (Admissible \\ without C,Lk D)};
    \node[base_node, right=of clmC2] (clm46) {\Cref{lem:adding_lits_C_D_wins} \\ (+lits - C,D)};
    \node[base_node, right=of clm46] (lem45) {\Cref{lem:clauses} \\ (Clause Lemma)};
    \node[base_node, right=of lem45] (lem42) {\Cref{lem:guards} \\ (Guard Lemma)};

    \foreach \n in {thm41, lem43, lem44, clm410, clm48, clm49, clm47, clmC3, clmC2, clm46, lem45, lem42} {
        \draw[thick] (\n.south west) -- (\n.south east);
    }

    \begin{scope}[>={Stealth[length=5pt, width=3pt]}, thick, opacity=0.8]

        \draw[->] (thm41.north) to[bend left=20] (lem43.north);
        \draw[->] (thm41.north) to[bend left=30] (lem44.north);

        \draw[->, orange!80!black] (lem43.north) to[bend left=15] (lem44.north);
        \draw[->, orange!80!black] (lem43.north) to[bend left=30] (clm410.north);
        \draw[->, orange!80!black] (lem43.north) to[bend left=45] (clm49.north);
        \draw[->, orange!80!black] (lem43.north) to[bend left=50] (clm47.north);

        \draw[->, green!50!black] (lem44.north) to[bend left=30] (clm48.north);
        \draw[->, green!50!black] (lem44.north) to[bend left=40] (clm47.north);
        \draw[->, green!50!black] (lem44.north) to[bend left=55] (clm46.north); %
        \draw[->, green!50!black] (lem44.north) to[bend left=60] (lem45.north);

        \draw[->, purple!80!black] (clm410.north) to[bend left=50] (lem45.north);

        \draw[->, cyan!80!black] (clm48.north) to[bend left=20] (clm49.north);
        \draw[->, cyan!80!black] (clm48.north) to[bend left=30] (clm47.north);
        \draw[->, cyan!80!black] (clm48.north) to[bend left=45] (lem45.north);

        \draw[->, red!80!black] (clm49.north) to[bend left=30] (clmC3.north);
        \draw[->, red!80!black] (clm49.north) to[bend left=40] (clmC2.north);
        \draw[->, red!80!black] (clm49.north) to[bend left=45] (clm46.north); %

        \draw[->, blue!80!black] (clm47.north) to[bend left=30] (clmC2.north);
        \draw[->, blue!80!black] (clm47.north) to[bend left=40] (clm46.north); %
        \draw[->, blue!80!black] (clm47.north) to[bend left=45] (lem45.north);

        \draw[->, brown!90!black] (clmC3.north) to[bend left=20] (clmC2.north);
        \draw[->, brown!90!black] (clmC3.north) to[bend left=30] (clm46.north); %

        \draw[->, teal!90!black] (clmC2.north) to[bend left=20] (clm46.north); %
        
    \end{scope}

    \begin{scope}[>={Stealth[length=5pt, width=3pt]}, thick, opacity=0.7]
        \draw[->] (thm41.south) to[bend right=70] (lem42.south);
        \draw[->, orange!80!black] (lem43.south) to[bend right=65] (lem42.south);
        \draw[->, green!50!black] (lem44.south) to[bend right=60] (lem42.south);
        \draw[->, purple!80!black] (clm410.south) to[bend right=55] (lem42.south);
        \draw[->, cyan!80!black] (clm48.south) to[bend right=50] (lem42.south);
        \draw[->, red!80!black] (clm49.south) to[bend right=45] (lem42.south);
        \draw[->, blue!80!black] (clm47.south) to[bend right=40] (lem42.south);
        \draw[->, brown!90!black] (clmC3.south) to[bend right=35] (lem42.south); %
        \draw[->, teal!90!black] (clmC2.south) to[bend right=30] (lem42.south);
        \draw[->, magenta!90!black] (clm46.south) to[bend right=25] (lem42.south); %
        \draw[->, violet!90!black] (lem45.south) to[bend right=20] (lem42.south);
    \end{scope}

    \end{tikzpicture}%
    }
    \caption{Lemma dependencies dictating the recursive evaluation of the graph}
    \label{fig:lemma_dependencies}
\end{sidewaysfigure}

\end{document}